\documentclass[11pt, oneside]{article}   	
\usepackage{amsmath,amssymb,amsthm,color,url,framed}
\usepackage{amscd}
\usepackage{enumerate}
\usepackage{comment}
\usepackage[margin=2cm]{geometry}                		
\usepackage{graphicx}				
\usepackage{tikz-cd}								
\usepackage[color=yellow]{todonotes}
\usepackage{hyperref}
\hypersetup{
colorlinks=true,
}
\def\contract{\makebox[1.2em][c]{\mbox{\rule{.6em}
{.01truein}\rule{.01truein}{.6em}}}}
\usepackage{makeidx}
\usepackage{fancyhdr}
\numberwithin{equation}{section}
\theoremstyle{plain}
\newtheorem{theorem}{Theorem}[section]     
\newtheorem{corollary}[theorem]{Corollary}             
\newtheorem{lemma}[theorem]{Lemma}              
\newtheorem{proposition}[theorem]{Proposition}

\theoremstyle{definition}
\newtheorem{definition}[theorem]{Definition}             
\theoremstyle{remark}
\newtheorem{remark}[theorem]{Remark}              
\newtheorem{example}{Example}                    

\DeclareMathOperator{\Div}{div}

\def\dt{\textnormal{dt}}

\def\diff{\textnormal{d}}

\def\bx{{\bf x}}

\def\bm{{\bf m}}

\def\p{{\partial}}

\def\RR{{\bf R}}
\def\bu{\mathbf{u}}

\def\U{{\bf u}}

\def\u{{\bf{u}}}

\def\x{{\bf{x}}}

\def\y{{\bf{y}}}

\def\div{{\textrm{div}}}
\def\curl{{\textrm{curl}}}

\def\rmd{{\mathrm{d}}}

\newcommand{\scp}[2]{{\big\langle {#1}\, , \, {#2}\big\rangle}}
\newcommand{\Scp}[2]{{\Big\langle {#1}\, , \, {#2}\Big\rangle}}
\newcommand{\SCP}[2]{{\left\langle {#1}\, , \, {#2}\right\rangle}_{L^2}}

\newcommand{\dede}[2]{\frac{\delta {#1}}{\delta {#2}} }

\newcommand{\wh}[1]{\widehat{#1}}

\newcommand{\mb}[1]{\mbox{\boldmath{$#1$}}}

\newcommand{\bs}[1]{\boldsymbol{#1}}
\newcommand{\mbs}[1]{\boldsymbol{#1}}

\newcommand{\mbf}[1]{\mathbf{#1}}
\newcommand{\mc}[1]{\mathcal{#1}}

\newcommand{\E}[1]{\mathbb{E}\left[{#1}\right]}

\usepackage[scr=boondoxo, scrscaled=1]{mathalfa}

\newenvironment{answerwide}[1][Answer]
{\bigskip \noindent \textbf{#1.} }
{\mbox{ } \hfill{$\blacktriangle$} }

\begin{document}

\title{\textbf{The Geometry of Stochastic Fluid Dynamics}}
\author{Darryl D. Holm
\thanks{Department of Mathematics, Imperial College London}  
\footnotesize
d.holm@ic.ac.uk
}
\date{}

\maketitle

\begin{abstract}
Stochastic geometric mechanics (SGM) is known for its potential utility in quantifying uncertainty in  global climate modelling of the Earth's ocean and atmosphere while also preserving the fundamental advective transport properties of ideal fluid flow. This paper is a pedagogical review  of the recent developments of the mathematical framework of stochastic geometric mechanics obtained from Lie group-invariant stochastic variational principles  in the context of model building for upper ocean dynamics, 
\end{abstract}
\vspace{-2mm}

The paper is divided into the following five parts. 
\begin{enumerate}[(I)]
\item
The first part discusses the origins of geometric mechanics applications in deterministic fluid dynamics based on Poincar\'e's two-page paper \cite{poincare1901forme} and Noether's theorem \cite{noether1918invariante}. These properties underlie the efficacy of geometric mechanics on the manifold of diffeomorphisms in applications to ideal fluid dynamics.

\item
The second part focuses on the example of the deterministic 3D Euler Boussinesq (EB) equations. The 3D EB equations are 
the source of several well-known approximate models applied in geophysical fluid dynamics (GFD). 
See Figure \ref{fig:ErwinTree} for an idea 
of the genealogy of approximate GFD models descending from the 3D EB equations. 
The entire genealogy of approximations of the 3D EB equations in Figure \ref{fig:ErwinTree} may be derived by making a sequence of approximations 
of its Lie group invariant Lagrangian in Hamilton's variational principle. This second part provides the differential geometric background for understanding the fundamental properties of the equations discussed in rest of the parts.

\item
The third part adds stochastic transport to the 3D Euler Boussinesq (EB) and derives its SALT equations. 
SALT is the abbreviation of Stochastic Advection by Lie Transport derived from  
stochastic geometric mechanics (SGM) based on stochastic variational principles. SGM is known for its potential utility in quantifying uncertainty for global climate modelling of the Earth's ocean and atmosphere interactions, while also preserving the fundamental advective transport properties of ideal fluid flow. After an introduction to the origins of SALT transport noise, this part describes the mathematical development of the framework of stochastic geometric mechanics in the context of fluid flow and wave dynamics obtained from Lie group-invariant variational principles.

\item
The fourth part focuses on Lagrangian Averaged Stochastic Lie Transport, abbreviated as LA-SALT.
In SALT, atmospheric `weather' produces uncertainty in advection arising
from motion on unresolved time scales. In LA-SALT, atmospheric `climate' is taken as the
expectation, and the atmospheric `weather' is treated as a field of pathwise fluctuations, 
as discussed in Ed Lorenz's famous lecture \cite{Lorenz1995}. 

\item
The fifth part focuses on SALT and LA-SALT to create stochastic Ocean--Atmosphere Models, abbreviated as SOAM. 
The SOAM approach brings us back to Hasselmann's paradigm, which decomposes a general climate model into deterministic 
and stochastic parts \cite{hasselmann1976stochastic}. Recent numerics with this model \cite{sharma2026structure} verifies 
the red-shift in the ocean energy predicted by  \cite{hasselmann1976stochastic}.
\end{enumerate}

\tableofcontents

\part{The Geometry of Deterministic Fluid Dynamics}
\section{Introduction} 
\subsection{The origins of geometric mechanics for ideal fluid dynamics}

Euler's classic theory of ideal fluid dynamics represents the motion of all the fluid particles in a container as a time-dependent curve on the manifold of volume-preserving smooth invertible maps, now called diffeomorphisms. Moreover, this curve describing the sequential actions of the volume-preserving diffeomorphisms on the fluid domain is a special optimal curve that distills the fluid motion into a single statement. Namely, ``A fluid moves to get out of its own way.'' Put more mathematically, an Euler fluid flow occurs along a time-dependent curve in the manifold of volume-preserving diffeomorphisms which is a geodesic with respect to the metric on its tangent space supplied by its kinetic energy, as noted by \cite{arnold1966geometrie}.


Thus, the problem of determining the kinematics of an incompressible fluid was transformed by \cite{arnold1966geometrie} into a geometric mechanics problem which can be expressed mathematically as a variational principle, defined by $\delta S=0$ where $S:=\int_0^T\ell(u)\,dt$, in which the Lagrangian $\ell(u):=\tfrac12\|u\|_{L^2}^2$ is equal to the kinetic energy of the fluid. Variational principles are not restricted to the kinematic motion of ideal fluids, though. If the forces arise from a potential, then variational principles can also be extended to derive the equations of compressible fluids with advected quantities. The Lagrangian is the key to extending fluid models beyond the kinematic case. The Lagrangian needed to include the potential energy and thermodynamics of advected quantities is a functional whose stationary variations determine reversible dynamical transformations between kinetic and potential energy \cite{HMR1998}. 

Fluid dynamics may be represented in two equivalent ways. The first way is the Lagrangian representation, in which fluid parcels carry labels $l^A(x_t) = l^A(g_t x_0)$ as an array of scalar functions moving along Lagrangian paths in fixed space $x_t=g_t x_0$ generated by a smooth invertible map $g_t$ parametrised by time $t$. In a flow domain  ${\cal D}$ the action of the flow is given by $G\times{\cal D}\to {\cal D}$ where $G$ is a smooth diffeomorphism and $g_t\in G$ satisfies the flow criterion $g_tg_s = g_{t+s}$, in which composition of smooth functions is written as concatenation $g_t\circ g_s=g_tg_s$.  The second way is the Eulerian representation, in which the labels of the fluid parcels  $l^A(x_t)=l^A(x_0) g_t^{-1}=: g_{t*} l^A(x_0)$ satisfy the following partial differential equation 
\index{Lagrangian representation!fluid parcels} \index{fluid parcels!Lagrangian representation}
\index{fluid parcels!Lagrangian paths} \index{Lagrangian representation!fluid parcel labels} 
\index{Lagrangian representation!smooth invertible flow maps}
\index{Eulerian representation!fixed spatial coordinates} \index{Eulerian representation!right invariant velocity} 
\begin{align}
\begin{split}
\p_t l^A(x_t) &= \p_t (g_{t*} l^A(x_0)) = \p_t \big(l^A(x_0) g_t^{-1}\big) 
= -  \, l^A(x_0) g_t^{-1}\dot{g}_t  g_t^{-1}  
\\& = -\, l^A(x_t) \dot{g}_t  g_t^{-1}
=: - {\cal L}_{ \dot{g}_t  g_t^{-1} } l^A(x_t) =: - {\cal L}_{u } l^A(x_t) 
\,.\end{split}
\label{push-forward}
\end{align}
\index{Lie derivative}\index{Lie derivative!chain rule}
In this chain rule calculation, $\p_t (g_{t*} l^A(x_0)) = - \,g_{t*} ( {\cal L}_{ \dot{g}_t  g_t^{-1} } l^A(x_0) ) $ defines  ${\cal L}_{ \dot{g}_t  g_t^{-1} }$ as the Lie derivative with respect to the velocity vector field $u := \dot{g}_t  g_t^{-1}$ \cite{holm1998euler,cotter2013noether}. For Eulerian advected scalar quantities such as the components of the fluid labels $l^A$, equation \eqref{push-forward} is simply $\p_t l^A(x,t) = - \, \bu\cdot\nabla l^A(x,t)$. 

Note that the fluid transport velocity vector field $u := \dot{g}_t  g_t^{-1} $ is invariant under the right action of an arbitrary smooth constant invertible map, since $ \dot{g}_t h (g_t h)^{-1}= \dot{g}_t  g_t^{-1}$. This right action transforms the Lagrangian path formula as $x_t =g_t h (x_0)$, which simply relabels the starting points of the Lagrangian paths as $x_0\to h (x_0)$ and leaves Eulerian fluid variables invariant.

As we will see later, this particle relabelling symmetry of the Eulerian fluid representation induces an infinite family of conserved functionals defined on closed loops of fluid material moving with the flow. These functionals defined on material loop space are known as fluid circulations and they are determined from the Kelvin-Noether theorem in geometric mechanics, as shown, e.g., in \cite{holm1998euler,cotter2013noether}. In fact, a fluid could be defined as a dynamical system whose equation of motion may be formulated in terms of a Kelvin--Noether circulation theorem.

At this point, we have seen that the Eulerian transport velocity $u := \dot{g}_t  g_t^{-1} $ in the fluid Lagrangian $\ell(u)$ is right invariant under the action of smooth invertible maps which form a Lie group via their flow condition, $g_t\circ g_s=g_tg_s$. If the Lagrangian is invariant under the action of a Lie group, then a classical theorem by Noether \cite{noether1918invariante} determines a quantity known as the momentum map \index{momentum map} corresponding to that Lie symmetry. The state space of the fluid can now be split into two sets: the momentum map and the advected quantities (such as the mass density) which are constant along Lagrangian particle paths. Newton's law in terms of the momentum map for fluid dynamics then becomes, ``The time rate of change of the momentum map equals the sum over forces expressed in terms of the advected quantities.'' Likewise, the ratio of the momentum map and the mass density turns out to be the integrand in the Kelvin--Noether theorem on the space of material loops. This sort of equation derived from a variational principle with a continuous group symmetry is called an Euler--Poincar\'e equation, after Poincar\'e's two-page paper \cite{poincare1901forme}. The origins of geometric mechanics and its applications in fluid dynamics are based on Poincar\'e's two-page paper \cite{poincare1901forme} and Noether's theorem \cite{noether1918invariante}. 
\index{Noether's theorem!momentum map}
\index{Noether's theorem!Kelvin-Noether theorem} 
\index{Geometric Mechanics!Kelvin-Noether theorem}

Geometric mechanics is particularly useful in the context of geophysical fluid dynamics (GFD). This is because GFD applies at planetary scales such as Earth's ocean and atmosphere. Motion at planetary scales involves large masses whose fluid motions are essentially unaffected by viscosity. Hence, provided that the GFD models are energetically closed, Hamilton's principle is available for deriving fluid models whose underlying geometry helps to understand their conserved quantities and can be used as a guide for numerical discretisation. However, formulating appropriate GFD models for atmospheric or oceanic processes remains a challenging problem, since one must deal with vast ranges of spatial and temporal scales whose associated dynamical interactions involve a large number of disparate and often unknown processes. Such a lack of information can be interpreted as a representation error. Representation errors are not the only source of modelling errors, though. Together with numerical errors and observation errors, many  existing physical parameterisations such as the air-sea interaction dynamics contain uncontrolled approximations. The many sources of uncertainty arising from incomplete information 
motivate the usage of stochastic parametrisations. However, these stochastic parametrisations should preserve the fundamental geometric structure that underlies the GFD models, particularly the Kelvin-Noether circulation theorem. Stochastic parametrisations that preserve this geometric structure can be achieved via stochastic variational principles.

\subsection{The momentum map in Noether's theorem} \index{momentum map!Noether's theorem}
\subsection*{Key points}
\begin{center}
    $\bullet\,$ Lie symmetries $\qquad$ 
    $\bullet\,$ Noether's theorem $\qquad$
    $\bullet\,$ Momentum map $\qquad$
    $\bullet\,$ Reduction by symmetry  \index{momentum map!reduction by symmetry}
\end{center}
Geometric mechanics deals with group-invariant variational principles. Its origin goes back to the early 1900s, when Poincar\'e used Hamilton's variational principle to show in a two-page paper \cite{poincare1901forme} that when a Lie algebra acts locally transitively on the configuration space $M$ of a Lagrangian mechanical system, then the Euler-Lagrange equations are equivalent to a new system of differential equations defined via the action of the Lie algebra on the configuration space. These equations are now known as the Euler-Poincar\'e equations and they play a crucial role in fluid dynamics. The original work of Poincar\'e is presented in modern language in the paper \cite{marle2013henri}. 

Consider a Lagrangian $L(q,v)$ defined for a configuration manifold $M$ on its tangent bundle $TM$ with coordinates $(q,v)\in TM$ as  $L:TM\to \mathbb{R}$ in Hamilton's action integral $S=\int L(q,v)dt$ for Hamilton's variational principle $\delta S = 0$. Noether's theorem \cite{noether1918invariante} states that each continuous symmetry of such as Lagrangian implies a conserved quantity for the corresponding Euler-Lagrange equations. The proof of Noether's theorem is a straightforward application of Hamilton's principle:
\begin{align}
\begin{split}
0 = \delta S &= \delta \int_0^T L(q,v) + \scp{p}{\frac{dq}{dt}-v}_{T^*M\times TM}\,dt 
\\&= \int_0^T \scp{\frac{\delta L}{\delta q}-\frac{dp}{dt}}{\delta q} + \scp{ \delta p }{ \frac{dq}{dt} - v } + \scp{ \frac{\delta L}{\delta v} - p } {\delta v} \,dt  + \scp{ p }{\delta q }\Big|_0^T
\end{split}
\label{Noether-proof}
\end{align}
Note that the definitions $p=\delta L / \delta v$ and $v=\frac{dq}{dt}$ arise as variational constraints imposed by the pairing $\scp{}{}_{T^*M\times TM}$. Suppose the Euler--Lagrange equations hold for a Lagrangian $L(q,\frac{dq}{dt})$ in the action integral $S$ that is invariant under the infinitesimal push-foward transformation $\delta q := \frac{d \phi_{*,\epsilon} q(t)}{d\epsilon}\big|_{\epsilon=0} =: -  {\cal L}_\xi q$ which defines Lie derivative ${\cal L}_\xi $.%
\footnote{The transformation $\phi_\epsilon: G\times M\to M$ is taken to be an action on $M$ of a Lie group $G$ with parameter $\epsilon=0$ at the identity. Consequently, the quantity $\xi$ is an arbitrary element of the Lie algebra $\xi\in T_eG=:\mathfrak{g}$.} 
Then the following Noether endpoint quantity must vanish for every solution of the Euler--Lagrange equations,
\begin{align}
\scp{ p }{\delta q }_{T^*M\times TM}\big|_0^T = \scp{ p }{-  {\cal L}_\xi q}_{T^*M\times TM}\big|_0^T 
\,.\label{Noether-quantity}
\end{align}
Hence,  the quantity 
\begin{align}
\scp{ p }{-  {\cal L}_\xi q}_{T^*M\times TM}=:\scp{p\diamond q}{\xi}_{\mathfrak{g}^*\times\mathfrak{g}}
\,.\label{momap}
\end{align}
is a constant of motion for these Euler--Lagrange equations.
Moreover, the quantity $p\diamond q\in\mathfrak{g}^*$ takes 
values in the dual $\mathfrak{g}^*$ of the Lie algebra symmetry algebra $\mathfrak{g}$ with respect to the pairing 
$\scp{}{}_{\mathfrak{g}^*\times\mathfrak{g}}$ introduced by the \emph{diamond operation} $\diamond: T^*M\to \mathfrak{g}^*$. 
\index{diamond operation $(\diamond)$!momentum map} \index{diamond operation $(\diamond)$!broken symmetry force}  
\index{momentum map}

The quantity $J(p,q):=p\diamond q \in \mathfrak{g}^*$ in \eqref{momap} that transforms variables from the cotangent bundle $T^*M$ to the dual $\mathfrak{g}^*$ of the Lie algebra $\mathfrak{g}$ associated with the Lie symmetry group $G$ is known as a \emph{momentum map}. Historically, the momentum map was introduced by Sophus Lie himself, as explained in \cite{weinstein1983sophus} who refers to \cite{lie1890theorie}, page 237. It has been further developed at various levels of generality by \cite{kirillov1962unitary}, \cite{kostant1970quantization}, \cite{souriau1970structure} and \cite{smale1970topologya, smale1970topologyb}. The momentum map is now a central object in geometric mechanics.  \index{momentum map}

A caveat should be mentioned: In general, the configuration manifold $M$ is not a Lie group. However, as discussed by Poincar\'e in  \cite{poincare1901forme}, when a Lie group $G$ acts transitively on a configuration manifold, $M$, the proof of Noether's theorem induces a cotangent-lift momentum map $J: T^*M\to\mathfrak{g}^*$. 
The cotangent lift momentum map $J: T^*M\to\mathfrak{g}^*$ is ${\rm Ad}^*$-equivariant and Poisson, even if $G$ is not a Lie symmetry of the Lagrangian in Hamilton's principle. Momentum maps naturally lead from the Lagrangian side of the dynamics to the Hamiltonian side. The Hamiltonian dynamics on $T^*M$ involves symplectic transformations. However, as we shall discuss below, for the class of Hamiltonians which can be defined as $H\circ J: \mathfrak{g}^Chap2*\to \mathbb{R}$, the momentum map induces Euler-Poincar\'e motion on the Lagrangian side and Lie-Poisson motion on the Hamiltonian side. The momentum map connects the Hamiltonian reduction techniques of  \cite{marsden1974reduction} with the Lagrangian reduction techniques of \cite{holm1998euler}. To illustrate this equivalence, we consider the situation in which the configuration manifold, $M$, is a Lie group, $G$.

One may reconstruct the solution on $G$ from its representation on $T^*G\setminus G\simeq\mathfrak{g}^*$. In that case, solving the equations describing the evolution of the momentum map on the dual Lie algebra $\mathfrak{g}^*$ is equivalent to solving the equations on the cotangent bundle $T^*G$ when the configuration manifold is $G$. When the Lie group $G$ acts transitively, freely and properly on the configuration manifold $M$, then one may reconstruct the solution on $M$ from its representation on $T^*G\setminus G\simeq\mathfrak{g}^*$. The last statement is proved for finite-dimensional Lie groups $G$ in, e.g., \cite{abraham1978foundations}. Provided the Lagrangian or Hamiltonian is hyperregular (invertible), the Legendre transform is a diffeomorphism. The Euler-Poincar\'e reduction procedure can then be expressed in terms of the cube of linked commutative diagrams shown in Figure \ref{fig:GM-Cube1}. 
\begin{figure}[h!]
\centerline{\includegraphics[width=6in]{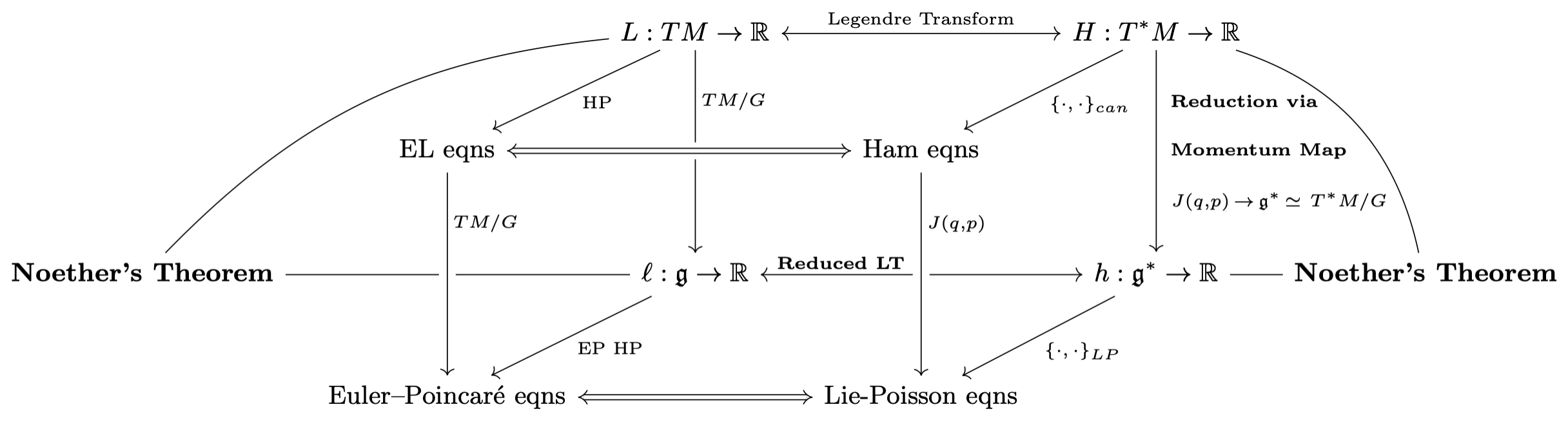} }
\caption{The relations among the fundamental concepts and transformations in 
geometric mechanics can be envisioned as a cube, on which each face is a commuting diagram.
Euler-Poincar\'e reduction (on the left side) and Lie-Poisson reduction (on the right side) are both indicated by the arrows pointing down,
provided the Legendre transformation \index{Legendre transformation} and reduced Legendre transformation are both invertible.
}\label{fig:GM-Cube1} 
\end{figure}\vspace{-3mm}

The notation in Figure \ref{fig:GM-Cube1} is as follows: $M$ denotes the configuration manifold which is assumed to be isomorphic to a Lie group, $G$; $TG$ is the tangent bundle; $T^*G$ is the cotangent bundle; $TG\setminus G \simeq \mathfrak{g}$ is the Lie algebra: and $T^*G\setminus G\simeq \mathfrak{g}^*$ is the dual of the Lie algebra. The Lagrangian is a functional $L:TG\to\mathbb{R}$ and the Hamiltonian is a functional $H:T^*G\to\mathbb{R}$. Euler-Poincar\'e reduction takes advantage of Lie group symmetries to transform the Lagrangian and Hamiltonian into group-invariant variables, which leads to a reduced Lagrangian $\ell:\mathfrak{g}\to\mathbb{R}$ and a reduced Hamiltonian $\hslash:\mathfrak{g}^*\to\mathbb{R}$. The diagram comprising the face of the cube involving these functionals in Figure \ref{fig:GM-Cube1} commutes if the Legendre transform is a diffeomorphism. This is guaranteed if the Lagrangian or Hamiltonian is hyperregular. The Euler-Lagrange equations and Hamilton's equations are related via an invertible change of variables, which also holds for the Euler-Poincar\'e equations and the Lie-Poisson equations. The invertibility of change of variables follows from the Legendre transform being a diffeomorphism. Many finite dimensional mechanical systems may be described naturally in this framework. The classic example is the rotating rigid body, discussed from the viewpoint of symmetry reduction by Poincar\'e in \cite{poincare1901forme}. In his 1901 paper, Poincar\'e also raised the issue of \emph{symmetry breaking}, by introducing the vertical acceleration of gravity, which breaks the  $SO(3)$  symmetry for free rotation and restricts it to  $SO(2)$ for rotations about the vertical axis. \index{symmetry breaking} 


Figure \ref{fig:GM-Cube1} describes geometric mechanics in the finite dimensional setting without symmetry breaking, but it cannot describe the infinite dimensional setting due to presence of the mass density, which breaks symmetry. We will show this in the next section, where we investigate the extension of Figure \ref{fig:GM-Cube1} to the setting of broken symmetries and stochasticity. 

Stochasticity can be included in the framework of Euler-Poincar\'e reduction by symmetry. The first attempt to include noise consistently in finite-dimensional symplectic Hamiltonian mechanics was by \cite{bismut1982mecanique} and reduction by symmetry of stochastic systems was studied by \cite{lazaro2008stochastic}. In the present work, we will review Euler-Poincar\'e reduction of stochastic infinite dimensional variational systems with symmetry breaking. 

The space of invertible maps with $H^s(M)$ smoothness, denoted $\mathfrak{D}^s$, is the configuration space for continuum mechanics and each map $\phi\in\mathfrak{D}^s$ is called a configuration. A fluid trajectory starting from $x_0\in M$ at time $t=0$ is given by $x(t):=\phi_t(x_0):=\phi(x_0,t)$, with $\mathfrak{D}^s\ni \phi:M\times\mathbb{R}^+\to M$ a continuous one-parameter subgroup of $\mathfrak{D}^s$. In continuum dynamics, one assumes the trajectory maps satisfy the \emph{flow property} under composition of maps. Namely, one assumes that $\phi_{t_1} \circ \phi_{t_2} = \phi_{t_1+ t_2}$, where the symbol $\circ$ denotes composition of functions.  Computing the time derivative of this one-parameter subgroup yields the \emph{reconstruction equation}, given by
\begin{equation}
\frac{\partial}{\partial t}\phi_t(x_0) = u(\phi_t(x_0),t).
\label{eq:reconstructiondeterministic1}
\end{equation}
The deterministic reconstruction equation \eqref{eq:reconstructiondeterministic1} defines Eulerian velocity in terms of Lagrangian velocity.

The infinite dimensional case is interesting because it is the natural setting for fluid dynamics, quantum mechanics and elasticity. We will discuss the infinite dimensional case in context of geophysical fluid dynamics, where symmetry under the smooth invertible maps of the flow domain is broken by the spatial dependence of the initial mass density.

\begin{definition}[Advected quantity]
A fluid variable is said to be \emph{advected}, if it keeps its value along Lagrangian particle trajectories. Advected quantities are sometimes called \emph{tracers}, because the evolution histories of scalar advected quantities with different initial values (labels) trace out the Lagrangian particle trajectories of each label, or initial value, via the \emph{push-forward} of the full diffeomorphism group, i.e., $a_t=\phi_{t\,*}a_0= (\phi_t^{-1})^*(a_0)$, where $\phi_t$ is a curve parametrised by $t$ on the manifold of diffeomorphisms that represents the fluid flow.
\end{definition}\smallskip

\begin{remark}[Advected quantities break symmetries] \index{advected quantities}
When several advected quantities are involved, the space $V^*$ is the direct sum of several vector spaces, in which each summand space hosts a different advected quantity. In general, each additional advected quantity decreases the dimension of the isotropy subgroup of the diffeomorphisms. For example, consider an ideal deterministic fluid with a scalar buoyancy variable, $b$. In this case, the Lagrangian corresponding to the model will depend on the mass form $\rho\,d\mu = (\phi^{-1})^*(d\mu)=\phi_*(d\mu)$%
\footnote{We denote the operations pull-back by a smooth invertible map $\phi$ as $\phi^*$ and its push-forward by $\phi_*$.}
 and on the buoyancy $b$ in a parametric manner. This Lagrangian will be right-invariant under the action of the isotropy subgroup $\mathfrak{D}^s_{\mu,b} = \{\phi\in\mathfrak{D}^s|\, \phi_*(d\mu)=d\mu \text{ and } \phi_*b=b\}$. Hence, each additional advected quantity breaks more symmetry. For brevity in notation, one usually writes $\mathfrak{D}^s_{a_0}$ for the isotropy subgroup, no matter how many advected quantities are present. One then writes $a_t=\phi_{t\,*}a_0$ to represent the push-forward by the flow map $\phi_t$ of all advected quantities and $a_0$ to denote the initial values of the advected quantities. The corresponding vector space for the $a_0$ is the coset denoted $a_0\in \mathfrak{D}^s/\mathfrak{D}^s_{a_0}\in V^*$.
\end{remark}

\begin{remark}[Coadjoint action and the diamond operator]
The coadjoint action is an important operator in geometric mechanics and representation theory. It was shown by \cite{kirillov1962unitary} and in further work by \cite{kostant1970quantization} and \cite{souriau1970structure} that the coadjoint orbits of a Lie group $G$ have the structure of symplectic manifolds and are connected with Hamiltonian mechanics. See \cite{kirillov1999merits} for a review. The infinitesimal adjoint and coadjoint Lie algebra actions $\mathrm{ad}: \mathfrak{g}\times \mathfrak{g}\to \mathfrak{g}$ and $\mathrm{ad}^*: \mathfrak{g}\times \mathfrak{g}^*\to \mathfrak{g}^*$ for a semidirect product group are valuable for fluid mechanics, as well, because they introduce the two fundamental operators that appear in the equations of motion. Namely, the Lie derivative is responsible for transport of tensors along vector fields and its dual action is given by the diamond operation $(\diamond)$ which arose in our discussion of Noether's theorem in equation \eqref{momap}. The diamond operation $(\diamond)$ is defined via two different real-valued symmetric non-degenerate pairings, $\scp{}{}_{V^*\times V}: V^*\times V \to \mathbb{R}$ and $\scp{}{}_{\mathfrak{g}^*\times\mathfrak{g}}: \mathfrak{g}^*\times\mathfrak{g}\to \mathbb{R}$, by the relation
\begin{equation}\label{def: diamond}
\scp{b\diamond a}{\xi}_{\mathfrak{g}^*\times\mathfrak{g}} := \scp{b}{- \mathcal{L}_\xi a}_{V^*\times V}
\,,\end{equation}
where  $ \mathcal{L}_\xi$ is the Lie derivative defined as the tangent at the identity of the Lie chain rule in \eqref{push-forward}. This is recalled in the present notation as \index{Lie derivative!chain rule}
\begin{align}
  \mathcal{L}_\xi f := \frac{d}{d\epsilon}({\phi_\epsilon}^* f )\Big|_{\epsilon=0} \,.
\label{def: LieDerivative2}
\end{align}
As we shall see, the diamond operation $(\diamond)$ encodes the forces arising from symmetry breaking in the Euler-Poincar\'e equations of motion. \index{symmetry breaking!forces}  \index{symmetry breaking!circulation}  \index{symmetry breaking!semidirect product action}
\end{remark}

\subsection{Geometric mechanics with diffeomorphisms}
\subsection*{Key points}
\begin{center}
    $\bullet\,$ Lagrangians and Hamiltonians $\qquad$ \index{diffeomorphism!Lagrangian} \index{diffeomorphism!Hamiltonian}
    $\bullet\,$ Legendre transform $\qquad$ \\ \index{diffeomorphism!Legendre transform} \index{Legendre transform!diffeomorphism}
    $\bullet\,$ Reduction by symmetry $\qquad$ \index{reduction by symmetry!reconstruction equation}
    $\bullet\,$ Reconstruction equation  \index{reconstruction equation!reduction by symmetry}
\end{center}

Euler-Poincar\'e reduction for a semidirect product group $\mathfrak{D}^s\times V$ as developed in \cite{holm1998euler} is sketched below in Figure \ref{fig:cubesdp} \cite{holm2019stochastic}. \index{semidirect product group}
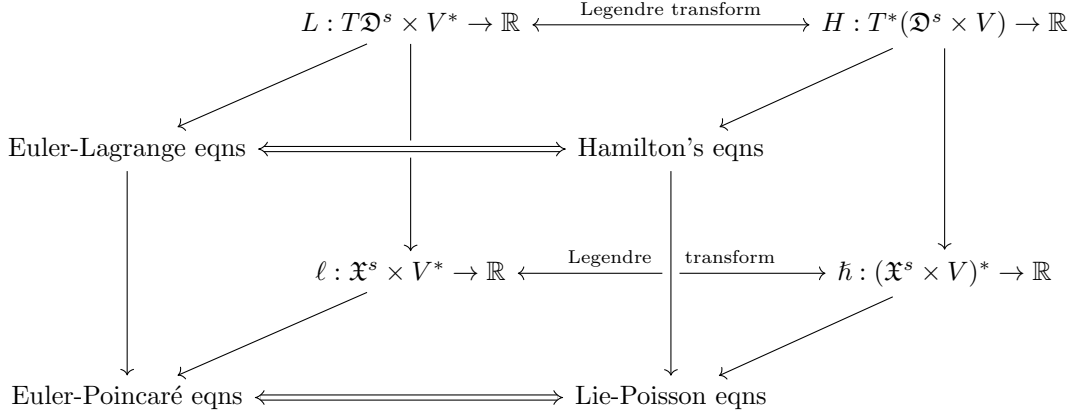
\begin{figure}[h] 
\small
\centering
\begin{tikzcd}[row sep=3em, column sep=small]
& 
L:T\mathfrak{D}^s\times V^*\to\mathbb{R} \arrow[dl] \arrow[rr,  "\text{Legendre transform}", leftrightarrow] \arrow[dd] 
& 
& 
H:T^*(\mathfrak{D}^s\times V)\to\mathbb{R} \arrow[dl] \arrow[dd]
\\
\text{Euler-Lagrange eqns} \arrow[rr, crossing over, Leftrightarrow] 
& 
& \text{Hamilton's eqns}
\\
&
\ell:\mathfrak{X}^s\times V^*\to\mathbb{R} \arrow[dl] \arrow[rr, "\text{Legendre \hspace{0.25cm} transform}", leftrightarrow] 
& 
& 
\hslash:(\mathfrak{X}^s\times V)^*\to\mathbb{R} \arrow[dl] 
\\
\text{Euler-Poincar\'e eqns} \arrow[rr, Leftrightarrow] \arrow[from=uu, crossing over]
& 
& 
\text{Lie-Poisson eqns} \arrow[from=uu, crossing over]
\end{tikzcd}
\caption{The cube of continuum mechanics in the semidirect product group setting. Reductions by Lie group symmetry are indicated as arrows pointing down.}
\label{fig:cubesdp}
\end{figure}
 Comparison of Figure \ref{fig:GM-Cube1}  with Figure \ref{fig:cubesdp} shows several new features arise in semidirect product Lie group reduction which differ from Euler-Poincar\'e reduction by symmetry when the configuration space itself is a Lie group. These differences can be conveniently explained by introducing the physical concept of an \emph{order parameter}. 
 
The order parameters in continuum mechanics are the elements of $V^*$ which are advected by the action of the diffeomorphism group $\mathfrak{D}^s$. The advection is defined simply as the semidirect product action on the elements of $V^*$. The introduction of each additional advected state variable (or, order parameter) into the physical problem reduces the symmetry group $\mathfrak{D}^s$ of the Lagrangian in Hamilton's principle down to the isotropy subgroup $\mathfrak{D}^s_{a_0}$ of the initial conditions, $a_0$, for the entire set of advected quantities, $a$. The action of the diffeomorphism group $\mathfrak{D}^s$ on these initial conditions then describes their advection as the action of $\mathfrak{D}^s$ on its coset space $\mathfrak{D}^s\setminus\mathfrak{D}^s_{a_0}=V^*$. 
 
 Once the initial values of the order parameters, $a_0$,  have been set, one must still define a Legendre transform to pass from the Lagrangian formulation into the Hamiltonian formulation and vice versa. The Legendre transform in the setting of semidirect products is a partial Legendre transform, since it transforms between $T\mathfrak{D}^s$ and $T^*\mathfrak{D}^s$ or $T\mathfrak{D}^s\setminus\mathfrak{D}^s_{a_0} \simeq \mathfrak{X}^s$ and $T^*\mathfrak{D}^s\setminus\mathfrak{D}^s_{a_0} \simeq\mathfrak{X}^{s*}$ only after having fixed the value $a_0$ of the order parameters, which live in $V^*$. This coset reduction is what Figure \ref{fig:cubesdp} shows. The remaining right-invariance of a functional under the action of the isotropy subgroup is called its \emph{particle relabelling symmetry}.

Continuing on the Lagrangian side in Figure \ref{fig:cubesdp}, consider a Lagrangian $L:T\mathfrak{D}^s\times V^*\to\mathbb{R}$. By fixing the value of $a_0\in V^*$, one can construct $L_{a_0}:T\mathfrak{D}^s\to\mathbb{R}$. If this Lagrangian is right-invariant under the action of the isotropy subgroup $\mathfrak{D}_{a_0}^s$, then one can construct 
\begin{equation}
\begin{aligned}
L\left(\frac{d}{dt}\phi\circ \phi^{-1},e,a_0\right) &= L_{a_0}\left(\frac{d}{dt}\phi\circ \phi^{-1},e\right)\\
&= \ell_{a_0}\left(\frac{d}{dt}\phi\circ \phi^{-1}\right) = \ell\left(\frac{d}{dt}\phi\circ \phi^{-1}, \phi_*a_0\right).
\end{aligned}
\label{eq:lagrangians}
\end{equation}
Here $\circ$ means composition of functions. The same procedure applies to the Hamiltonian. Since the coadjoint action is known, it is straightforward to formulate the Lie-Poisson equations. The details of Hamiltonian semidirect product reduction and also more information on the Lagrangian semidirect product reduction may be found in \cite{holm1998euler,HolmGM2025}. 

The coadjoint action of the Lie algebra on its dual is also required for the Lagrangian semidirect product reduction. One can use the deterministic reconstruction equation to see that the argument of the Lagrangians in \eqref{eq:lagrangians} is
\begin{equation}
\frac{d}{dt}\phi\circ \phi^{-1} = u.
\label{reconstruct-eqn}
\end{equation}
Using this information, one may integrate the Lagrangian in time to construct the action functional. By requiring the variational derivative of the action functional to vanish, one may compute the equations of motion. However, due to the removal of symmetries, the variations are no longer free. 

\begin{remark}
After the deterministic background material discussed in Part \ref{Determ-EBeqns} of this paper, in Part \ref{EB-SALT} we will replace the deterministic reconstruction equation \eqref{reconstruct-eqn} by a semimartingale and formulate the stochastic Euler-Poincar\'e theorem, which is the basis for most of the modelling work performed in the STUOD project.
\end{remark}

\newpage

\part{Deterministic Euler Boussinesq (EB) equations}\label{Determ-EBeqns}

\section*{Key words}
\begin{center}
$\bullet\,$ Geometric mechanics $\qquad$ $\bullet\,$ Fluid dynamics $\qquad$ $\bullet\,$ Stochastic partial differential equations\\
$\bullet\,$ Lie groups $\qquad$ $\bullet\,$ Diffeomorphism group $\qquad$ $\bullet\,$ Sobolev spaces $\qquad$ $\bullet\,$ Momentum maps
\end{center}
 
\section{Elements of Geometric Mechanics for deterministic EB equations}\index{EB equations}

The 3D Euler-Boussinesq (EB) equations for the dynamics of a rotating, stratified, incompressible fluid are the fundamental equations for oceanography. 
This section follows the implications for the EB equations of the Lie symmetry known as fluid parcel relabelling in the Eulerian representation. Importantly, the Lie symmetry under fluid parcel relabelling in the Eulerian representation underlies the variational derivation of the 3D EB equations. Although we emphasise the EB equations, it turns out that sequential approximations of the Lagrangian in Hamilton's principle for the EB equations also produce the \emph{flow chart} of familiar approximate models for particular regimes of ocean physics, shown in Figure \ref{fig:ErwinTree}. \index{Hamilton's principle}
In fact, all of the approximate ocean models on the flow chart in Figure \ref{fig:ErwinTree} admit the same consequences of the Lie symmetry we investigate here for the EB equations.

The 3D EB models are derived from the Euler--Poincar\'e theorem discussed in this section for a fluid Lagrangian in Hamilton's principle that depends 
on \emph{advected quantities} such as mass density and buoyancy under gravity. These material properties
are carried along by the EB fluid's transport velocity. \index{advected quantities} \index{Euler--Poincar\'e!theorem} 

In general, the Euler--Poincar\'e theorem follows from variational principles written in terms of variables that are invariant under a Lie group action \cite{HMR1998}.  As mentioned earlier, a Lie group is a group of transformations that depends smoothly on a set of parameters. For Eulerian fluids, the Lie symmetry group comprises the smooth invertible maps acting to relabel fluid parcels while leaving the Eulerian fluid variables invariant. \index{Lie group!definition}

A variational principle seeks critical points of a time integral defined on the tangent space of a manifold. Well known examples of Lie group invariant variational principles include Fermat's principle for geometric optics, Euler's equations for rigid body motion, Euler's equations for ideal fluid dynamics, and Einstein's equations for general relativity. In all of these examples, Lie symmetries introduce conservation laws via Noether's theorem \cite{noether1983invariante}. Thus, the solutions are constrained to stay on level sets of these conservation laws.

Figure \ref{fig:LRBS} depicts the central role played by the Lagrangian in Hamilton's variational principle for fluid dynamics. Namely, approximations in the Lagrangian preserve the mathematical structure of the parent model. Thus, approximations in the Lagrangian comprise a foundational approach in designing approximate fluid equations. 
 
As shown in Figure \ref{fig:GM-Cube}, the Euler--Poincar\'e theorem reveals a \emph{framework of related results} in dynamical systems following from Lie group invariant variational principles. Figure \ref{fig:GM-Cube} envisions these related results as a cube whose six faces comprise commuting diagrams of smooth equivariant maps. These 
commuting diagrams are unfolded in Figure \ref{fig:GM-Unfold}. \index{Euler--Poincar\'e!theorem}  

This section begins by reviewing the properties of Lie derivatives. It then uses Lie symmetries to prove the Euler--Poincar\'e theorem for fluid dynamics 
with advected quantities. For the 3D EB equations, the advected quantities are mass density and buoyancy. 
\index{advected quantities}

The Euler--Poincar\'e theorem for fluid dynamics with advected quantities systematically formulates and derives fluid models which obey
the Kelvin--Noether theorem for the circulation dynamics on fluid material loops carried along (advected) by the fluid flow, as depicted in Figure \ref{fig:KelvinThm}.
\index{Kelvin--Noether!theorem}
\begin{figure}[h!]
\centerline{{\includegraphics[width=5in]{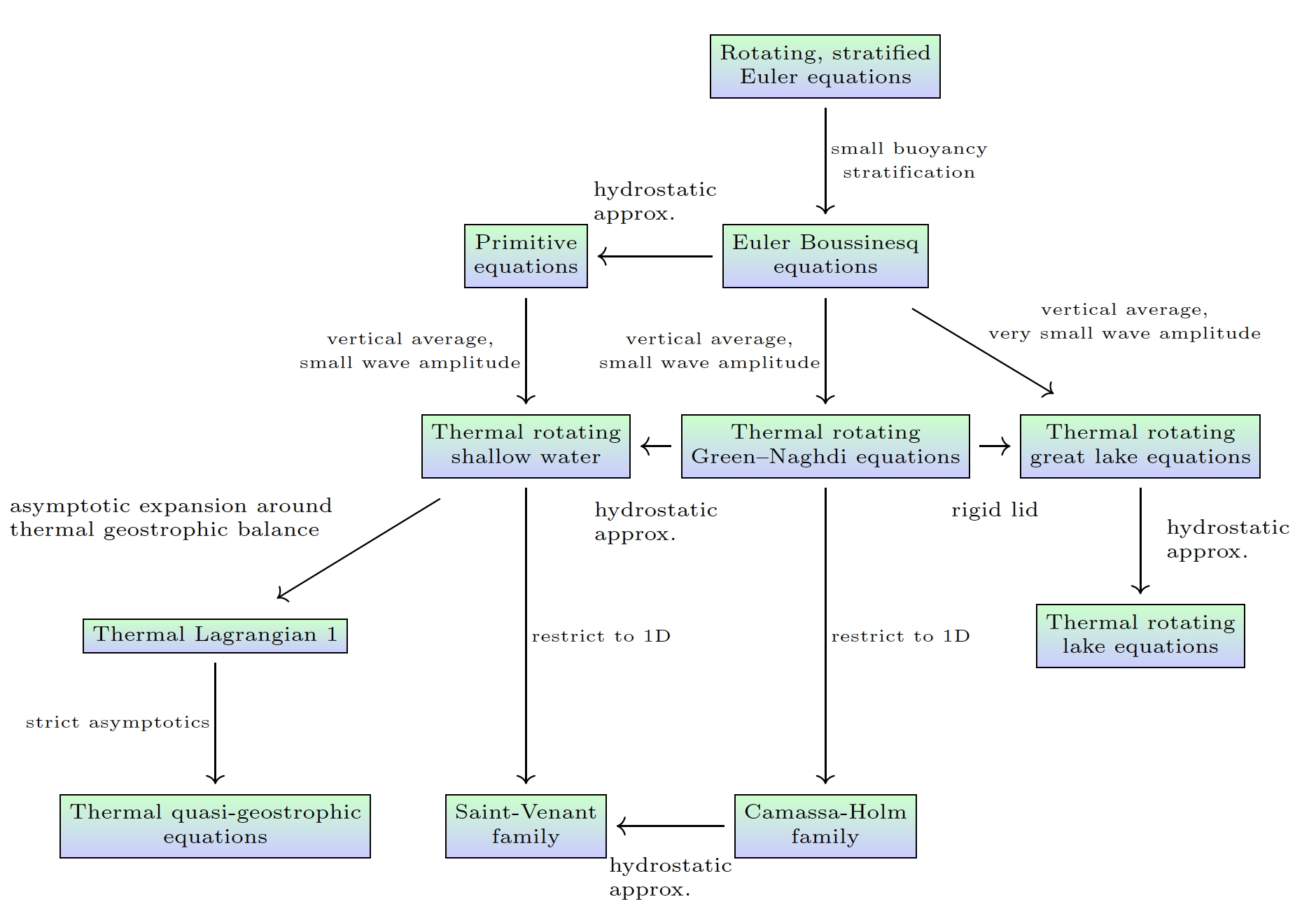}}}
\caption{The 3D Euler Boussinesq (EB) equations studied here have many approximate `descendants' along a `flow chart' of models obtained by
making sequential approximations of the Lagrangian for Hamilton's principle for EB flow. In doing so, each `descendant' in the flow chart
inherits the geometrical properties of 3D EB. Figure courtesy of E. Luesink.}\label{fig:ErwinTree}
\end{figure}

\begin{figure}[h!]
\centerline{{\includegraphics[width=5in]{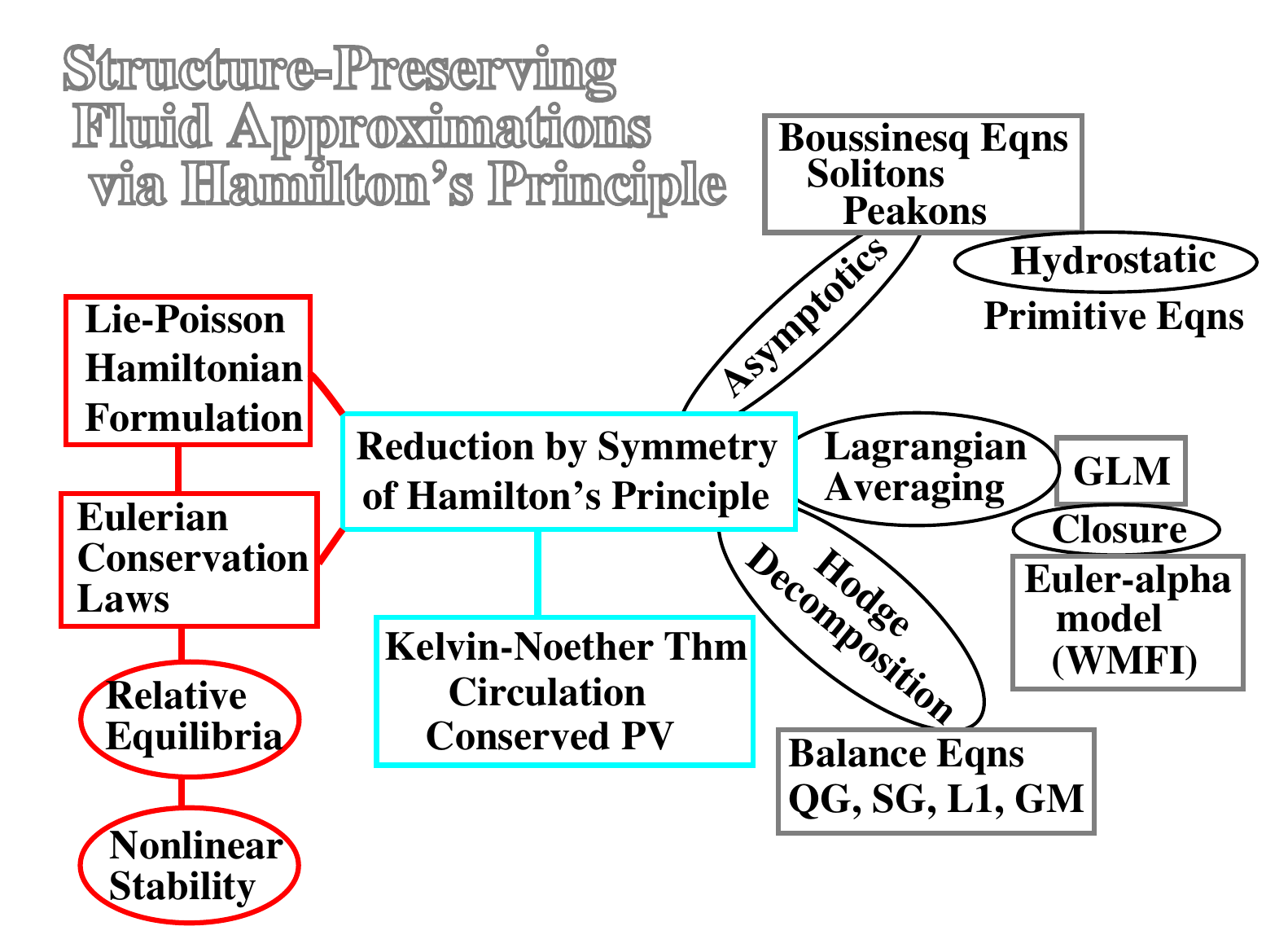}}}
\caption{All of the models in Figure \ref{fig:ErwinTree} stem from approximating the highest level Lagrangian in Hamilton's principle.
This figure depicts the central role played by the right-invariant Euler--Poincar\'e Hamilton principle as a foundational approach in designing approximate models 
in fluid dynamics which preserve its mathematical structure. The Euler--Poincar\'e theorem guarantees the Kelvin--Noether form of the fluid motion equations.
Moreover, each model's Legendre transformation \index{Legendre transformation} yields its Lie--Poisson Hamiltonian formulation, 
whose constants of motion in turn yields equilibrium solutions and play a role in their stability analysis. } \label{fig:LRBS}
\end{figure}
\index{Kelvin--Noether!form}

\begin{figure}[h!]
\centerline{{\includegraphics[width=6in]{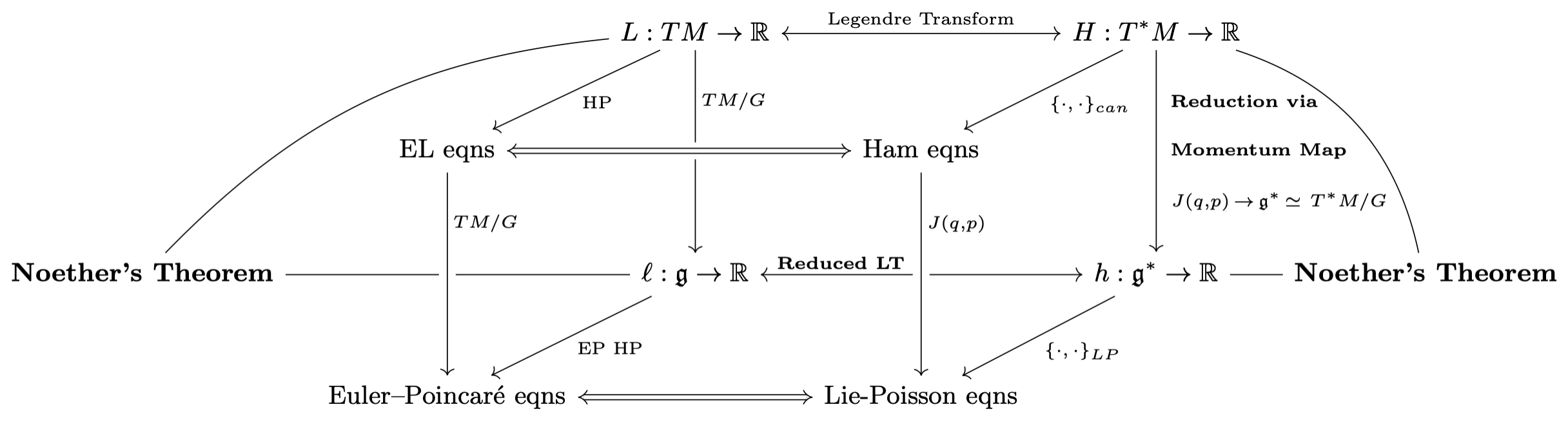}}}
\caption{The relations among the fundamental concepts in 
geometric mechanics can be envisioned as a cube, on which each face is a commuting diagram.
As we will see, though, fluid dynamics with advected quantities requires the corresponding structure in 
Figure  \ref{fig:GM-Cube1}.}\label{fig:GM-Cube}
\end{figure}

\begin{figure}[h!]
\centerline{{\includegraphics[width=6in]{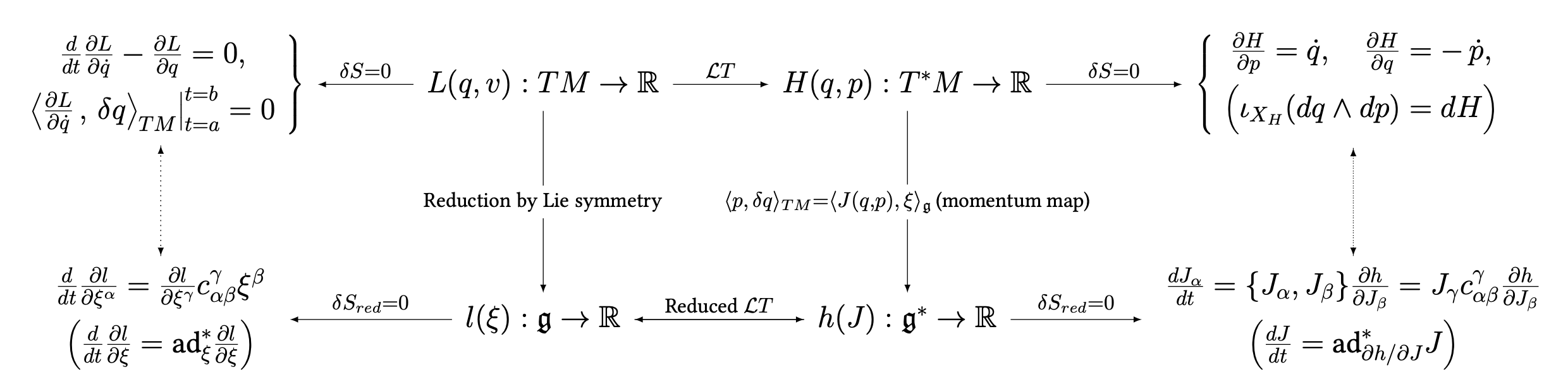}}}
\caption{Unfolding the cube of commuting diagrams for Geometric Mechanics shows how the 
classical equations descend from symplectic phase space to Euler--Poincar'e equations defined on the dual of the Lie algebra for 
the symmetry group of Eulerian fluid dynamics under relabelling of Lagrangian fluid parcels.}\label{fig:GM-Unfold}
\index{Geometric Mechanics!commuting diagrams}
\end{figure}

\begin{figure}[h!]
\centerline{{\includegraphics[width=3in]{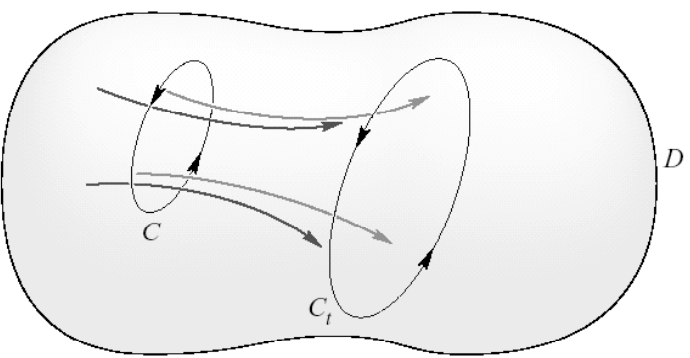}}}
\caption{The Kelvin--Noether theorem for the circulation dynamics on fluid material loops carried along (advected) by the fluid flow is a basic element of the mathematical structure of fluid dynamics.\label{fig:KelvinThm} 
}\index{advected quantities} \index{Kelvin--Noether!theorem}
\end{figure}



\subsection{Infinitesimal  \& finite Lie group actions} \index{Lie group!actions on a manifold}
Recall that a Lie group is a group of transformations which depends smoothly on a set of parameters. A Lie group is also a manifold, 
i.e., a space on which the rules of calculus apply. 

The \emph{action} of a Lie group $G$ on a
manifold $M$ is a group of transformations of $M$ associated to 
elements of the group $G$, whose composition acting on $M$ corresponds
to group multiplication in $G$.

\begin{definition}[Left and Right actions] Let $M$ be a manifold and let $G$ be a Lie group. A
\emph{left action} of a Lie group \index{left action!Lie group} $G$ on $M$ is a smooth mapping 
$\Phi{:\ }  G\times M\to M$ such that 
\begin{enumerate}
\item[(i)] $\Phi(e, x) = x \hbox{ for all } x \in M$, 
\item[(ii)] $\Phi(g, \Phi(h, x)) = \Phi(gh, x)$ for all $g, h \in G$ and $x \in
M$, and 
\item[(iii)] $\Phi(g, \cdot)$ is a diffeomorphism on $M$ 
for each $g \in G$. 
\end{enumerate}
One often uses the convenient notation $gx$ for $\Phi(g, x)$ and
thinks of the group element $g$ acting on the point $x \in M$. The
associativity condition (ii) above then simply reads $(gh)x = g(hx)$.
\end{definition}
Similarly, one can define a \emph{right action}, \index{right action!Lie group}
which is a map $\Psi{:\ } 
M\times G\to M$  satisfying $\Psi(x,e) = x$ and $\Psi(\Psi(x,g),h) =
\Psi(x,gh)$. The convenient notation for right action is $xg$ for
$\Psi(x, g)$, the right action of a group element $g$ on the point
$x\in M$. Associativity $\Psi(\Psi(x,g),h) = \Psi(x,gh)$ is
then be expressed conveniently as $(xg)h=x(gh)$.

\newpage

Consider a Lie group acting transitively on a manifold $M$ as $G\times M\to M$. 
Follow the flow of finite transformations of $G$ along a path $\phi_\epsilon\in G$ which is smoothly parametrised by $\epsilon\in \mathbb{R}$. 
The corresponding infinitesimal  transformation in the neighbourhood
of the identity $\epsilon=0$ is obtained from the linear term in the Taylor expansion of the finite Lie group transformation acting on the initial point $q_0\in M$ as,
\begin{align}
\begin{split}
q_\epsilon := \phi_\epsilon (q_0) &:= 
q_0 + \epsilon \left[\frac{d}{d\epsilon}\phi_\epsilon (q_0)\right]_{\epsilon=0}
+ O(\epsilon^2)
\\&=:
q_0 + \epsilon \Phi(q_0)
+ O(\epsilon^2)
\,.
\end{split}
\label{Taylor-exp}
\end{align}
The linear term in the Taylor expansion \eqref{Taylor-exp} is denoted as 
\[\delta q  = q_\epsilon' |_{\epsilon=0} = \frac{dq_\epsilon}{d\epsilon} \Big|_{\epsilon=0} =: \Phi(q)\,.\]

\textbf{Lie chain rule.} When the group operation $q_\epsilon := \phi_\epsilon (q_0)$ above is given by composition of smooth functions $f\in C^\infty(M)$, one obtains the \emph{pull-back relation} known as the \emph{Lie chain rule}:
\index{pull-back relation} \index{Lie chain rule} \index{vector field!definition}
\begin{align}
\frac{d}{d\epsilon}({\phi_\epsilon}^* f ) = {\phi_\epsilon}^*( \mathcal{L}_{w_\phi} f)\,,
\label{LieChainRule}
\end{align}
where $ \mathcal{L}_{w_\phi}$ is the Lie derivative with respect to the vector field ${w_\phi}\in \mathfrak{X}(M)$ whose characteristic curves generate the smooth flow $\phi_\epsilon$.%
\footnote{The vector field ${w_\phi}=\phi_\epsilon'\phi_\epsilon^{-1}|_{\epsilon=0}$ generates right-invariant action 
and ${w_\phi}=\phi_\epsilon^{-1}\phi_\epsilon'|_{\epsilon=0}$ generates left-invariant action of the smooth flow $\phi_\epsilon$ defined for smooth functions on $M$.}
Upon assuming that $\phi_\epsilon|_{\epsilon=0}=Id$, one defines the Lie
 derivative of a function as the tangent to the Lie chain rule at the identity, $\epsilon=0$. Namely, the Lie  derivative $\mathcal{L}_{w_\phi}$ with respect to vector field ${w_\phi}\in \mathfrak{X}(M)$ of a smooth function $f$ on manifold $M$ is given by
\begin{align}
 \mathcal{L}_{w_\phi} f := \frac{d}{d\epsilon}({\phi_\epsilon}^* f )\Big|_{\epsilon=0} \,.
\label{def: LieDerivative}
\end{align}
\index{Lie chain rule!Lie derivative}\index{Lie derivative!definition!Lie chain rule}

The properties of Lie derivatives are of central importance in the Euler--Poincar\'e theory of ideal fluid dynamics \cite{HMR1998}. In fact, the properties of Lie derivatives underlie the following sections of this paper which explain the derivation of the Euler--Poincar\'e theory of stochastic fluid mechanics. Hence, the next section is devoted to the actions of Lie derivatives on advected quantities.
\index{Lie derivative!action}
 \medskip
 
\subsection{Properties of Lie derivatives}

\begin{definition}[Lie derivative of a differential $k$-form] \index{Lie derivative!properties}
Differential forms are quantities one integrates, such as line elements, surface elements and volume elements. 
The Lie group $G$ for ideal fluid flows is the manifold of smooth invertible maps, aka \textit{diffeomorphisms}. 
The \emph{Lie derivative} of a differential $k$-form $\Lambda^k(M)$ by a vector field $w\in\mathfrak{X}(M)$ is defined 
by linearising the pull-back action $\phi^*_t\Lambda^k$ by the flow $\phi_t\in G$ around the identity map at $t=0$.%
\footnote{At this point, we begin denoting the map parameter as $t$ in $\phi_t$ to invoke dynamics. 
Later, we will denote infinitesimal variations using the map parameter $\epsilon$ in $\phi_\epsilon$.}
Namely,
on the domain of flow $M$,
\begin{eqnarray}
 \mathcal{L}_w\Lambda^k = \frac{d}{dt}\bigg|_{t=0}\phi^*_t\Lambda^k
\quad\hbox{maps}\quad
\mathfrak{X}\times \Lambda^k \mapsto  \Lambda^k
\,.
\end{eqnarray}
\end{definition}
\begin{corollary}[Product rule for the Lie derivative of a wedge product of $k$-forms]  \index{Lie derivative!product rule} \index{Lie derivative!$k$-forms}
\begin{equation}
 \mathcal{L}_w(\alpha\wedge\beta)
=
 \mathcal{L}_w\alpha\wedge  \beta
+
\alpha\wedge  \mathcal{L}_w\beta
\,.
\label{Lie product rule}
\end{equation}
\end{corollary}
\begin{proof}
The pull-back action on differential forms is natural, i.e.,
\[
\phi_t^*(\alpha\wedge\beta) = (\phi_t^*\alpha\wedge\phi_t^*\beta).
\]
The Lie derivative is defined in \eqref{def: LieDerivative} as the tangent at the identity of the Lie chain rule in \eqref{LieChainRule}. 
Consequently, the Lie derivative of a wedge product of differential forms satisfies the product rule in \eqref{Lie product rule}.  
\end{proof}

Pull-backs of vector fields also lead to Lie derivative expressions. \index{vector field!pull-back} \index{vector field!push-forward}
\begin{definition}[Lie derivative of a vector field]
The \emph{Lie derivative} of a vector field $Y\in\mathfrak{X}$ by another vector field $X\in\mathfrak{X}$ is defined 
by linearising the pull-back by $\phi_t$ of vector field $X$ around the identity $t=0$.\index{push-forward}%
\footnote{The push-forward of a vector field -- denoted as ${\phi_t}_*W$ -- is the pull-back by the inverse $\phi_t^{-1}$ 
for a smooth map $\phi_t$ acting on a vector field $W\in \mathfrak{X}(M)$. For the push-forward, one finds $ \frac{d}{dt}\big|_{t=0}{\phi_t}_*Y= -  \mathcal{L}_XY$.}
\begin{eqnarray*}
 \mathcal{L}_XY = \frac{d}{dt}\bigg|_{t=0}\phi^*_tY
\quad\hbox{maps}\quad
 \mathcal{L}_X\in\mathfrak{X} \mapsto  \mathfrak{X}
\,.
\end{eqnarray*}
\end{definition}
\begin{theorem}
The Lie derivative $ \mathcal{L}_XY$ of a vector field $Y$ by a vector field $X$ satisfies
\begin{equation}
 \mathcal{L}_XY = \frac{d}{dt}\bigg|_{t=0}\phi^*_tY
= [X,\,Y]
\,,
\label{LieBrkt}
\end{equation}
where $[X,\,Y]=XY-YX$ is the commutator of the vector fields $X$ and $Y$.
\end{theorem}
\index{vector field!commutator}

\begin{proof}
Denote the vector fields in components as
\begin{eqnarray*}
X = X^i(q)\frac{\partial}{ \partial q^i} = \frac{d}{dt}\bigg|_{t=0}\phi^*_t
\quad\hbox{and}\quad
Y = Y^j(q)\frac{\partial}{ \partial q^j}
\,.
\end{eqnarray*}
By the pull-back relation $\phi_t^*\big[Y,\,Z\big] = \big[\phi_t^*Y,\,\phi_t^*Z\big]$ a direct computation yields, on using the matrix identity $dM^{-1}=-\,M^{-1}dM M^{-1}$,
\begin{align*}
 \mathcal{L}_XY 
&= \frac{d}{dt}\bigg|_{t=0}\phi^*_tY
= \frac{d}{dt}\bigg|_{t=0} \left(Y^j(\phi_tq)\frac{\partial}{ \partial (\phi_tq)^j}\right)
\\
&= \frac{d}{dt}\bigg|_{t=0} \left(Y^j(\phi_tq)
\left[\frac{ \partial (\phi_tq)}{\partial q}^{-1}\right]^k_j
\frac{\partial}{ \partial q^k}\right)
\\
&= 
\left(
X^j \frac{ \partial Y^k}{\partial q^j}
-
Y^j \frac{ \partial X^k}{\partial q^j}
\right)\frac{\partial}{ \partial q^k}
= 
[X,\,Y]
\,.
\end{align*}
\end{proof}

\begin{corollary}
The Lie derivative of the pull-back of the commutator of vector fields $\phi_t^*\big[Y,\,Z\big] = \big[\phi_t^*Y,\,\phi_t^*Z\big]$ 
see, e.g., \cite{HolmGM2025} implies the \emph{Jacobi identity} condition for the vector fields to form an algebra. \index{Jacobi identity}
\end{corollary}

\begin{proof}
By the product rule  \index{Lie derivative!product rule!Lie bracket}
\begin{align*}
\tfrac{d}{dt}\phi_t^*\big[Y,\,Z\big] = \big[\tfrac{d}{dt}\phi_t^*Y,\,\phi_t^*Z\big] +  \big[\phi_t^*Y,\,\tfrac{d}{dt}\phi_t^*Z\big]
\,,\end{align*}
and the definition of the Lie bracket \eqref{LieBrkt} we have
\begin{align*}
\frac{d}{dt}\bigg|_{t=0}\phi^*_t\big[Y,\,Z\big]
&=
\big[X,\big[Y,\,Z\big]\big]
\\&=
\big[ [X,Y],\,Z\big] + \big[ Y,\,[X,Z]\big]
=
\frac{d}{dt}\bigg|_{t=0}\big[\phi^*_tY,\,\phi^*_t Z\big]
\,.\end{align*}
This is the \emph{Jacobi identity} for vector fields.
\end{proof}
\index{vector field!Jacobi identity}
\index{Lie derivative!Cartan form}


\begin{remark}
An equivalent definition of Lie derivative of  $B\in\Lambda^k$ is given \textit{in Cartan form} by $ \mathcal{L}_uB= w\contract dB + d(u\contract B)$, where the operator 
$\contract$ reduces the degree of a $k$-form:  $\contract\times\Lambda^k \to \Lambda^{k-1}$.
\end{remark}\index{Lie derivative!Cartan form}


\subsection{Examples of Lie derivatives for 3D EB} \label{examp: LieDeriv}%
\index{Lie derivative!Examples for 3D EB}
\begin{description} 

\item [(Functions)]
$(\partial_t + \mathcal{L}_u)b(\mathbf{x},t)
=
\partial_t b + \mathbf{u}\cdot \nabla b \hbox{ for a vector field} \quad u=\dot{\phi}_t \phi_t^{-1}
$\,, 

\item [(1-forms)]
$(\partial_t + \mathcal{L}_u)( \mathbf{v}(\mathbf{x},t)\cdot d\mathbf{x})
= \big(\partial_t\mathbf{v}+ \mathbf{u}\cdot\nabla \,\mathbf{v} + v_j\nabla u\,^j\big)\cdot d\mathbf{x}$
\\ 
\hspace{1.82cm}$=\big(\partial_t\mathbf{v} 
-\,\mathbf{u}\times{\rm curl}\,\mathbf{v} + \nabla(\mathbf{u}\cdot\mathbf{v})\big)\cdot d\mathbf{x} 
$\,,  in Cartan form.
\smallskip

\item [(Vector fields)] 
$(\partial_t + \mathcal{L}_u) w(\mathbf{x},t)
=
\partial_t w - \mathrm{ad}_uw = \partial_t w - [u, w ] $\\
$= \big(\partial_t \mathbf{w} + \mathbf{u}\cdot\nabla \mathbf{w} - \mathbf{w}\cdot\nabla \mathbf{u} \big)\cdot\nabla$
\smallskip

\item [(2-forms)]
$(\partial_t + \mathcal{L}_u)({\boldsymbol{\omega}}(\mathbf{x},t)\cdot d\mathbf{S})
=
\big(\partial_t{\boldsymbol{\omega}} -\,{\rm curl}\,(\mathbf{u} \times
{\boldsymbol{\omega}} )
+
\mathbf{u}\,{\rm div}\,{\boldsymbol{\omega}}\big)\cdot d\mathbf{S}$\,,
\smallskip

\item [(3-forms)]
$(\partial_t + \mathcal{L}_u)(D(\mathbf{x},t)\, d\,^3x)
=
(\partial_tD + {\rm div}\,D\,\mathbf{u})\, d\,^3x$\,.
\end{description}
\begin{remark}
The action of the flow map $\phi_t\in G$ is assumed to satisfy the Lie group properties defined by composition
of maps: $\phi_s\circ \phi_t = \phi_{t+s}$. Also notice that the Lie derivative by vector field  $u=\mathbf{u}\cdot\nabla$ of a 1-form density 
$m:=\mathbf{m} \cdot d\mathbf{x} \otimes d\,^3x$ may be obtained from the above table of examples by applying the product rule.
\end{remark}
After this brief review, let's apply the definitions and key properties of Lie derivatives to ideal fluid dynamics. 

\subsection{Actions of Lie derivatives in EB fluid dynamics}\index{Lie derivative!action}
The Euler--Boussinesq (EB) equations for the incompressible motion of an ideal
flow of a stratified fluid and velocity $\mathbf{u}$ satisfying  
$\mathrm{div}(\mathbf{u})=0$  in a rotating frame with Coriolis 
parameter $\mathrm{curl}\,\mathbf{R}=2\boldsymbol{\Omega}$ are given by
\index{Ertel's theorem}
\begin{equation}
\underbrace{\
\partial_t \,\mathbf{u}
+
\mathbf{u}\cdot\nabla\mathbf{u}\
}_{\hbox{acceleration}}
=
\underbrace{\
-\,gb\nabla{z}\
}_{\hbox{buoyancy}}
+\
\underbrace{\
\mathbf{u}\times 2\boldsymbol{\Omega}\
}_{\hbox{Coriolis}}
\
-
\underbrace{\quad
\nabla p\quad
}_{\hbox{pressure}}
\label{EulBouss-motion-eqn-Newton}
\end{equation} 
where $-g\nabla{z}$ is the constant downward acceleration of gravity and $b$ is the buoyancy, a scalar function of space and time which satisfies the \textbf{ advection relation},
\begin{equation} 
\partial_t \,b
+
\mathbf{u}\cdot\nabla b 
=
0
\,.
\label{buoy-advection}
\end{equation} 
As for Euler's equations without buoyancy, requiring preservation of the divergence-free (volume-preserving) constraint $\nabla\cdot\mathbf{u}=0$ results in a Poisson equation for pressure $p$,  
\begin{equation} 
-\,\Delta \left(p + \frac{1}{2}|\mathbf{u}|^2\right)
=
\mathrm{div} (-\,\mathbf{u}\times\mathrm{curl}\,\mathbf{v}) 
+ g \partial_z b
\,.
\label{Lamb-pressure-eqn-buoy}
\end{equation}
The Poisson equation for pressure $p$ satisfies a Neumann boundary condition obtained from $\mathbf{u}\cdot \mathbf{\hat{n}}|_{\partial D}=0$
because the velocity $\mathbf{u}$ must be tangent to the boundary.

The Newton's law form of the EB equations (\ref{EulBouss-motion-eqn-Newton}) may be rearranged as
\begin{eqnarray} 
\partial_t \,\mathbf{v}
-\,\mathbf{u}\times\mathrm{curl}\,\mathbf{v} 
+
gb\nabla{z}
+
\nabla \left(p + \frac{1}{2}|\mathbf{u}|^2\right)
&=
0
\,,
\label{EulBous-motion-eqn-Lie}
\end{eqnarray} 
where
$ 
\mathbf{v}
\equiv
\mathbf{u}+\mathbf{R}
\,,$
$
\boldsymbol{\omega}=\mathrm{curl}\,\mathbf{v}=\mathrm{curl}\,\mathbf{u}+2\boldsymbol{\Omega}
\,,
$
and $\nabla\cdot\mathbf{u}=0$.


\begin{theorem} \label{Stokes-thm-strat}
The EB equations in \eqref{EulBous-motion-eqn-Lie} satisfy the following vector form  of the Kelvin-Stokes theorem,
\begin{align} 
\begin{split}
\frac{d}{dt}\oint_{c(\mathbf{u})} \mathbf{v} \cdot d \mathbf{x} 
&=
\frac{d}{dt}
\int\!\!\!\!\int_{S(\mathbf{u})}
\mathrm{curl}\,\mathbf{v} \cdot d \mathbf{S}
=
\int\!\!\!\!\int_{S(\mathbf{u})}
\bigg(
\frac{\partial  }{\partial t} 
+ 
 \mathcal{L}_{u}
\bigg)
(\mathrm{curl}\,\mathbf{v} \cdot d \mathbf{S})
   \\&\hspace{-1cm}
   = 
\int\!\!\!\!\int_{S(\mathbf{u})}
\Big(\partial_t \,\boldsymbol{\omega}
-
\mathrm{curl}\,(\mathbf{u}\times \boldsymbol{\omega}) \Big) \cdot d \mathbf{S} 
= 
\int\!\!\!\!\int_{S(\mathbf{u})}
\Big(-\,g\nabla{b}\times\nabla{z}\Big) \cdot d \mathbf{S} 
 \,,
\end{split}
 \label{Stokes-strat-vector}
\end{align} 
where the surface $S(\mathbf{u})$ is bounded by an arbitrary circuit $\partial S=c(\mathbf{u})$ moving with the fluid. 
Thus, non-alignment of the gradient of buoyancy $\nabla{b}$ with the vertical $\nabla{z}$ creates circulation. 
\end{theorem} \index{Stokes theorem} \index{Helmholtz vorticity equation}

Geometrically, equation (\ref{EulBous-motion-eqn-Lie}) may be written in Lie derivative form as 
 \begin{equation}
\big(\partial_t +  \mathcal{L}_u\big)v^\flat  + gbdz + d\varpi =0
\,,
\label{EulBousseqns-1form}
\end{equation}
where $v^\flat = \mathbf{v}\cdot d \mathbf{x}$ is the circulation 1-form and $\varpi$ is an augmented pressure variable,
 \begin{equation}
 \varpi: = p+\frac{1}{2}|\mathbf{u}|^2
-\mathbf{u}\cdot\mathbf{v}
\,.
\label{varpi-def}
\end{equation}
In addition, the buoyancy is advected as a scalar and the volume form is preserved
 \begin{equation}
\big(\partial_t +  \mathcal{L}_u\big)b = 0
\,,
\quad\hbox{with}\quad
 \mathcal{L}_u\,d\,^3x = \mathrm{div}\mathbf{u}\,d\,^3x = 0
\,.
\label{EulBousseqns-tracers}
\end{equation}

\begin{remark} \label{Stokes-thm-strat}
The EB equations in \eqref{EulBousseqns-1form} satisfy the following geometric form of the Kelvin-Stokes theorem,
\begin{align} 
\begin{split}
\frac{d}{dt}\oint_{c(\mathbf{u})} v^\flat
&=\frac{d}{dt}
\int\!\!\!\!\int_{S(\mathbf{u})} dv^\flat
=
\int\!\!\!\!\int_{S(\mathbf{u})}
d\bigg(
\Big(
\frac{\partial  }{\partial t} 
+ 
 \mathcal{L}_{u}
\Big)
v^\flat
\bigg)
= 
 - \, g
\int\!\!\!\!\int_{S(\mathbf{u})}
db\wedge dz 
\,,
\end{split}
 \label{Stokes-strat-geom}
\end{align} 
where the surface $S(\mathbf{u})$ is bounded by an arbitrary circuit $\partial S=c(\mathbf{u})$ moving with the fluid.  
\end{remark}

\begin{remark}[\emph{Conservation of  potential vorticity (PV) on EB Lagrangian particle trajectories}]$\,$\\
The fluid velocity vector field is denoted as $u=\mathbf{u}\cdot \nabla$ and the circulation one-form as $v^\flat=\mathbf{v}\cdot d\mathbf{x}$. 
The exterior derivatives of the two equations in (\ref{EulBousseqns-1form}) are written as
 \begin{equation}
\big(\partial_t +  \mathcal{L}_u\big)dv^\flat = -gdb\wedge dz
\quad\hbox{and}\quad
\big(\partial_t +  \mathcal{L}_u\big)db = 0
\,.
\label{EulBousseqns-2form}
\end{equation}
Consequently, the product rule for the Lie derivative of a wedge product in \eqref{Lie product rule} implies that%
\footnote{Equation \eqref{EulBouss-PV-def} introduces the Hodge star operator on $\mathbb{R}^3$ discussed next. \label{Hodge-Laplace-DeRham-sec}}
 \begin{equation}
\big(\partial_t +  \mathcal{L}_u\big)\star(db \wedge dv^\flat) = 0
\quad\hbox{or}\quad
\partial_t \,q
+
\mathbf{u}\cdot\nabla q 
=
0
\,,
\label{EulBousseqns-3form}
\end{equation}
in which the scalar quantity
 \begin{equation}
 q  = \star(db \wedge dv^\flat) =  \nabla b\cdot \mathrm{curl}\,\mathbf{v}
 \,,
\label{EulBouss-PV-def}
\end{equation}
is called \emph{potential vorticity} and is abbreviated as PV.
The potential vorticity is an important diagnostic for many processes in geophysical fluid dynamics. 
\begin{remark}
Conservation of PV on fluid parcels is called \emph{Ertel's theorem}.\index{potential vorticity!PV}\end{remark} \index{Ertel's theorem}
\end{remark}

\begin{definition}\index{Hodge star operator}
The \emph{Hodge star operator} establishes a linear correspondence between the space of $k$-forms and the space of $(3-k)$-forms, for $k=0,1,2,3$. This correspondence may be defined by its usage: 
\begin{eqnarray*}
*1&:=& d\,^3x = dx^1\wedge dx^2\wedge dx^3\,,\\
*d\mathbf{x} &:=& d\mathbf{S}\,,\\
(*dx^1,*dx^2,*dx^3) &:=& (dS_1,dS_2,dS_3)\,,\\
&:=&(dx^2\wedge dx^3,\,dx^3\wedge dx^1,\,dx^1\wedge dx^2)\,,\\
*d\mathbf{S} &=&d\mathbf{x} \,,\\
(*dS_1,*dS_2,*dS_3)&:=&(dx^1,dx^2,dx^3)  \,,\\
*d\,^3x &=& 1
\,,
\end{eqnarray*}
in which each formula admits cyclic permutations of the set $\{1,2,3\}$. 
\end{definition}

\begin{remark}
Note that $*\!*\alpha=\alpha$ for these $k$-forms.
\end{remark}

\begin{definition}[$L^2$ inner product of forms]\label{innerproductforms-def}
The Hodge star induces an \emph{inner product} $(\,\cdot\,,\,\cdot\,):\Lambda^k(M)\times\Lambda^k(M)\to\mathbb{R}$ on the space of $k$-forms. Given two $k$-forms $\alpha$ and $\beta$ defined on a smooth manifold $M$, one defines their $L^2$ inner product as
\begin{eqnarray} \index{inner product!differential forms}
(\alpha,\,\beta)
: =
\int_M \alpha\wedge *\beta = \int_M \langle \alpha, \beta \rangle\;d\,^3x
\,,\label{innerproductforms-eqn}
\end{eqnarray}
where $d\,^3x$ is the volume form. The main examples of the inner product are for $k=0,1$. These are given by the $L^2$ pairings,
\begin{eqnarray*}
(f,\,g) = \int_M f\wedge *g&:=& \int_M fg\,d\,^3x \,,\\
(\mathbf{u}\cdot d\mathbf{x},\,\mathbf{v}\cdot d\mathbf{x})
=
\int_M \mathbf{u}\cdot d\mathbf{x}
\wedge *(\mathbf{v}\cdot d\mathbf{x})
&:=& \int_M \mathbf{u}\cdot \mathbf{v}\,d\,^3x\,.
\end{eqnarray*}
\end{definition}

Combining the Hodge star operator with the exterior derivative yields the following vector calculus operations:
\begin{eqnarray*}
*\,d*(\mathbf{v}\cdot d\mathbf{x}) &=& \mathrm{div}\,\mathbf{v}
\,,\\
*\,d(\mathbf{v}\cdot d\mathbf{x}) &=& (\mathrm{curl}\,\mathbf{v})\cdot d\mathbf{x}
\,,\\
d*d*(\mathbf{v}\cdot d\mathbf{x}) &=& (\nabla\mathrm{div}\,\mathbf{v})\cdot d\mathbf{x}
\,,\\
*\,d*d(\mathbf{v}\cdot d\mathbf{x}) &=&
\mathrm{curl}\,(\mathrm{curl}\,\mathbf{v})\cdot d\mathbf{x}
\,.
\end{eqnarray*}%

\index{Ertel's theorem}
\begin{remark}[\emph{Ertel's theorem for the vorticity vector field}]$\,$\\ \textrm
Writing the vorticity vector field $\omega=\boldsymbol{\omega}\cdot\nabla = ( \star dv^\flat)^\sharp$, we have the 
vorticity equation in terms of vector fields
 \begin{align*}
\big(\partial_t +  \mathcal{L}_u\big)\omega 
=
\partial_t\omega + [u,\,\omega]  = g(\star( dz\wedge db)) )^\sharp := g\nabla{z}\times\nabla{b}\cdot\nabla
\,. \end{align*}
Thus, conservation of the potential vorticity $q= \star(db \wedge dv^\flat) = \omega \contract db  =  \nabla b\cdot \mathrm{curl}\,\mathbf{v} $ 
may also be proved by the product rule upon regarding $\omega$ as  a vector field,
 \begin{align*}
 \big(\partial_t +  \mathcal{L}_u\big)q
 &=
 \big(\partial_t +  \mathcal{L}_u\big)\star(db \wedge dv^\flat)
 = 
\big(\partial_t +  \mathcal{L}_u\big)(\omega \contract db )
\\&=
\big(\big(\partial_t +  \mathcal{L}_u\big)\omega\big) \contract db 
+ \omega\contract  \big(\partial_t +  \mathcal{L}_u\big) db = 0
\,. \end{align*}
 \end{remark}
 
 \begin{remark}[\emph{Material derivative formulation}]$\,$\\ \textrm
Denoting with standard musical notation with superscripts $\flat: \mathfrak{X}\to \Lambda^1$ and $\sharp: \Lambda^1 \to \mathfrak{X}$
\[
\frac{D}{Dt} := \partial_t +  \mathcal{L}_u
\quad\hbox{and}\quad
\omega := ( \star dv^\flat  )^\sharp = \boldsymbol{\omega}\cdot\nabla
\]
provides an intuitive expression of Ertel's theorem (\ref{EulBousseqns-3form}) that helps understand it in terms of the time derivative $\frac{D}{Dt}$ following the flow of the fluid particles. Namely, it suggests writing in vector form
\[
\frac{D\omega}{Dt}
=
\frac{D}{Dt}(\boldsymbol{\omega}\cdot\nabla )
= g\nabla{z}\times\nabla{b}\cdot\nabla
= g (\star( dz\wedge db)) )^\sharp
\quad\hbox{and}\quad
\frac{Db}{Dt} = 0
\,,
\]
so that the product rule for derivatives yields conservation of PV on fluid parcels, as 
\[
\frac{Dq}{Dt}
=
\frac{D}{Dt}(\boldsymbol{\omega}\cdot\nabla{b})
=
\Big(\frac{D}{Dt}(\boldsymbol{\omega}\cdot\nabla)\Big)b
+
(\boldsymbol{\omega}\cdot\nabla)\frac{Db}{Dt}
= g\nabla{z}\times\nabla{b}\cdot\nabla{b} + 
(\boldsymbol{\omega}\cdot\nabla)\frac{Db}{Dt} = 0
\,.
\]
 
 \end{remark}

\begin{remark}[\emph{Ertel's theorem implies conserved quantities}]$\,$\\ \textrm
The constancy of the scalar quantities $b$ and $q$ on fluid parcels implies conservation of the spatially integrated quantity,
 \begin{equation}
 C_\Phi
 =
 \int_D \Phi(b,q)\,d\,^3x
 \,,
\label{EulBouss-conlaws}
\end{equation}
for any differentiable function $\Phi$  for which the integral exists. 
\index{potential vorticity!PV}\index{Ertel's theorem}
\begin{proof}
\begin{align*}
 \frac{d}{dt}C_\Phi
 &=
 \int_D \Phi_{,b}\partial_t b + \Phi_{,q}\partial_t q\,d\,^3x
=
- \int_D \Phi_{,b}\mathbf{u}\cdot\nabla{b} 
+ \Phi_{,q}\mathbf{u}\cdot\nabla{q}\,d\,^3x 
\\&=
- \int_D \mathbf{u}\cdot\nabla\Phi(b,q)\,d\,^3x 
=
- \int_D \nabla\cdot\big(\mathbf{u}\,\Phi(b,q)\big)\,d\,^3x 
\\&=
- \oint_{\partial D} \Phi(b,q)\,\mathbf{u}\cdot\mathbf{\hat{n}}\,d\,S 
=
0\,,
\end{align*}
when the normal component of velocity $\mathbf{u}\cdot\mathbf{\hat{n}}$ vanishes at the boundary $\partial D$. 

\end{proof}
\end{remark}

\subsection{Variational formulae in three dimensions}  \label{sec-2.1}
In fluid dynamical applications, the advected Eulerian variables $b$
and $D\,d^3x$ represent the buoyancy $b$ (or specific entropy, for the
compressible case) and volume element (or mass density) $D\,d^3x$,
respectively. In the Euler--Poincar\'e theory, the variations of the tensor
functions $a$ at fixed $\mathbf{x}$ and $t$ are also given by advective Lie
derivatives, namely
$\delta a = - \mathcal{L}_w\,a$, for a vector field $w=(\delta g_\epsilon) g_\epsilon^{-1}$.
In particular, we have
\begin{align}
\begin{split}
\delta b
&= - \mathcal{L}_w\, b = -\mathbf{w}\cdot\nabla\,b\,,  \hfil \hbox{for a vector field }\quad w:=\partial_\epsilon g_\epsilon g_\epsilon^{-1}\big|_{\epsilon=0}
\\
\delta (D\ d^3x)&=- \mathcal{L}_w\,(D\,d^3x)
= -\nabla\cdot(D\mathbf{w})\ d^3x\,,
\\
\delta u = \delta(\mathbf{u}\cdot\nabla) &= \partial_t w - \mathrm{ad}_u w 
= \big(\partial_t \mathbf{w} + \mathbf{u}\cdot\nabla \mathbf{w} - \mathbf{w}\cdot\nabla \mathbf{u} \big)\cdot\nabla\,.
\\ 
\end{split}
\label{vars: u-D-b}
\end{align}
Hence, Hamilton's principle with this
dependence and these variations yields the general form for GFD models, \index{Hamilton's principle!GFD models}
\begin{align}
\begin{split}
0 &=\delta \int _0^T\ell(u, b,D)\,dt
\\
&=\int _0^T
\scp{\frac{\delta \ell}{\delta {u}}}{\delta{u}}
+ \scp{\frac{\delta \ell}{\delta b}}{\delta b}
+\scp{\frac{\delta \ell}{\delta D}}{\delta D}\,dt
\\
&=\int _0^T \!\!\! \int _{\cal D}\bigg[\frac{\delta \ell}{\delta \mathbf{u}}
\cdot \Big(\frac{\partial \mathbf{w}}{\partial t}
-\mathrm{ad}_{\mathbf{u}}\,\mathbf{w}\Big)
-\frac{\delta \ell}{\delta b}\ \mathbf{w}\cdot\nabla\,b
-\frac{\delta \ell}{\delta D}\ \Big(\nabla
\cdot(D\mathbf{w})\Big)\bigg]d\,^3x\,dt
\\
&=\int _0^T \!\!\! \int _{\cal D}\mathbf{w}\cdot
\bigg[-\frac{\partial }{\partial t}
\frac{\delta \ell}{\delta \mathbf{u}}
-\mathrm{ad}^*_{\mathbf{u}}\ \frac{\delta \ell}{\delta \mathbf{u}}
-\frac{\delta \ell}{\delta b}\ \nabla\,b
+D\ \nabla\frac{\delta \ell}{\delta D}\bigg]d\,^3x\,dt
\\
&=-\int _0^T \!\!\! \int _{\cal D}\mathbf{w}\cdot
\bigg[\Big(\frac{\partial }{\partial t}
 + \mathcal{L}_u \Big)\frac{\delta \ell}{\delta \mathbf{u}}
+\frac{\delta \ell}{\delta b}\ \nabla\,b
-D\ \nabla\frac{\delta \ell}{\delta D}\bigg]d\,^3x\,dt
\\
&=-\int _0^T \!\!\! \int _{\cal D}\mathbf{w}\cdot
\bigg[\Big(\frac{\partial }{\partial t}
 + \mathcal{L}_u \Big)\frac{1}{D}\frac{\delta \ell}{\delta \mathbf{u}}
+\frac{1}{D}\frac{\delta \ell}{\delta b}\ \nabla\,b
- \nabla\frac{\delta \ell}{\delta D}\bigg] D\, d\,^3x\,dt
\,\end{split}
\label{eq-EP-Eul}
\end{align}
where in the last step we have used the auxiliary Lie-derivative relation for the continuity equation
$({\partial}/{\partial t}+ \mathcal{L}_u)Dd^3x = 0$ to factor out $Dd^3x$. 
We have also dropped boundary terms arising from
integrations by parts, by invoking `natural boundary conditions'.
Specifically, we have imposed $\mathbf{\hat{n}}\cdot\mathbf{w}=0$ on the
boundary, where $\mathbf{\hat{n}}$ is the boundary's outward unit normal
vector and ${w} = \delta g_t \circ g_t^{-1}$ vanishes at
the endpoints in time.

\textbf{Question}
Compute $\mathrm{ad}^*_{\mathbf{u}} \frac{\delta \ell}{\delta \mathbf{u}} $ by integrating by parts 
in passing from line 3 to line 4 of equation \eqref{eq-EP-Eul};
\[
\int _{\cal D} \mathbf{w} \cdot \mathrm{ad}^*_{\mathbf{u}} \frac{\delta \ell}{\delta \mathbf{u}}  \,d\,^3x
=
 \int _{\cal D} \frac{\delta \ell}{\delta \mathbf{u}} \cdot {\mathrm{ad}_{\mathbf{u}}\,\mathbf{w}}  \,d\,^3x
= 
- \int _{\cal D} \frac{\delta \ell}{\delta \mathbf{u}} \cdot \big(\mathbf{u}\cdot\nabla \mathbf{w} - \mathbf{w}\cdot\nabla \mathbf{u} \big)\,d\,^3x
\,.\]

\textbf{Answer}
In equation \eqref{eq-EP-Eul}, set $\frac{\delta \ell}{\delta \mathbf{u}}=: \mathbf{m}$ and compute $\mathrm{ad}^*_{\mathbf{u}} \mathbf{m} $
from the following definition,
\begin{align*}
\int _{\cal D} \mathbf{w} \cdot \mathrm{ad}^*_{\mathbf{u}} \mathbf{m} \,d\,^3x
&:=
\int _{\cal D} \mathbf{m} \cdot {\mathrm{ad}_{\mathbf{u}}\,\mathbf{w}}  \,d\,^3x
 =
- \int _{\cal D} \mathbf{m} \cdot \big(\mathbf{u}\cdot\nabla \mathbf{w} - \mathbf{w}\cdot\nabla \mathbf{u} \big)\,d\,^3x
\\ \scp{  \mathrm{ad}^*_u m}{w}  & = 
- \int _{\cal D} m_j \big(u^k w_{,\,k}^j - w^lu_{,\,l}^j \big)\,d\,^3x
= \int _{\cal D} \big(\partial_k(m_j u^k) + m_k u_{,\,j}^k \big)w^j\big)\,d\,^3x
\\&= 
 \int _{\cal D} \big( (\partial_km_j + m_k \partial_j ) u^k  \big) w^j\,d\,^3x
 =:
  \langle  \mathcal{L}_u m\,,\,w\rangle
\,.\end{align*}
Thus, the definitions \emph{coincide}: $\mathrm{ad}^*_um = \mathcal{L}_um := (\partial_km_j + m_k \partial_j ) u^k dx^j \otimes d^3x$ 
is the Lie derivative by vector field $u:=\mathbf{u}\cdot \nabla$ of 1-form density $m:=\mathbf{m} \cdot d\mathbf{x} \otimes d\,^3x$.
A 1-form density $m$ is \emph{dual} to a vector field $u$ in the $L^2$ integral pairing $\scp{}{}$ on $\mathcal{D}\in \mathbb{R}^n$,
in the sense that $\scp{u} = \int_{\cal D} u\contract m =  \int_{\cal D} u^j\partial_j\contract m_kdx^k\otimes d^3x  =  \int_{\cal D} \mathbf{u}\cdot \mathbf{m}\,d^3x$, since $\partial_j \contract dx^k = \delta_j^k$. 

\subsection{Euler--Poincar\'e framework for 3D continuum motion}
The Euler--Poincar\'e equations for continua may
now be summarized in vector form for advected Eulerian variables $a$ in the
example set \ref{examp: LieDeriv}. The last step in the calculation \eqref{eq-EP-Eul} 
expresses the Euclidean components of the Euler--Poincar\'e equations for continua
in Kelvin--Noether theorem form \index{Euler--Poincar\'e!equation} \index{Euler--Poincar\'e!equation} \index{Kelvin--Noether!theorem}
\begin{equation}
\Big(\frac{\partial }{\partial t} +  \mathcal{L}_u\Big)
\Big(\frac{1}{D}\frac{\delta \ell}{\delta
\mathbf{u}}\cdot d\mathbf{x}\Big)
\,+\,\frac{1}{D}\frac{\delta \ell}{\delta b}\nabla b \cdot d\mathbf{x}
\,-\nabla\Big(\frac{\delta \ell}{\delta D}\Big)\cdot d\mathbf{x}
 = 0,
\label{EP-Kthm}
\end{equation}
in which the variational derivatives of the Lagrangian $\ell(\mathbf{u}, b,D)$ are
computed as Fr\'echet derivatives, according to the usual physical conventions. 
Formula \eqref{EP-Kthm} is the Kelvin--Noether form of the equation of motion for ideal continua.
Hence, we have the explicit Kelvin theorem expression, \index{Kelvin--Noether!theorem}\index{Kelvin--Noether!form} 
\begin{equation} \label{KN-theorem-bD}
\frac{d}{dt}\oint_{{c}_t(\mathbf{u})} \frac{1}{D}\frac{\delta \ell}{\delta \mathbf{u}}\cdot d\mathbf{x}
=
\oint_{{c}_t (\mathbf{u})} \Big(\frac{\partial }{\partial t} +  \mathcal{L}_u\Big) \Big(\frac{1}{D}\frac{\delta \ell}{\delta \mathbf{u}}\cdot d\mathbf{x}\Big)
= 
-\oint_{{c}_t(\mathbf{u})}
\frac{1}{D}\frac{\delta \ell}{\delta b}\nabla b \cdot d\mathbf{x}\,,
\end{equation}
where the closed curve ${c}_t(\mathbf{u})$ moves with the fluid velocity
$\mathbf{u}$. 

Then, by applying Stokes' theorem to the last term in \eqref{KN-theorem-bD}, one finds that the Euler equations in \eqref{EP-Kthm} generate circulation of
$\mathbf{v}:=(D^{-1}\delta{l}/\delta\mathbf{u})$
unless the gradients~$\nabla b$ and $\nabla(D^{-1}\delta{l}/\delta{b})$ are functionally dependent and thus collinear. 

The corresponding \emph{conservation of potential vorticity} $q$ on fluid
parcels is given by
\begin{equation} \label{pv-cons-EP}
\frac{\partial{q}}{\partial{t}}+\mathbf{u}\cdot\nabla{q} = 0,
\quad \hbox{where}\quad
q:=\frac{1}{D}\nabla{b}\cdot\mathrm{curl}
\left(\frac{1}{D}\frac{\delta \ell}{\delta \mathbf{u}}\right).
\end{equation}
The continuity equation for $D$ and the advection of $b$ and $q$ combine to conserve the following integral quantity, 
\begin{equation} \label{pv-Casimir}
C_\Phi = \int D \Phi(b,q)\,d^3x
\,,\end{equation}
for any differential function $\Phi(b,q)$.

\textbf{Question:}
Prove PV advection equation in \eqref{pv-cons-EP} either by direct calculus manipulations, or by using the following properties: 
product rule for the Lie derivative $(\partial_t +  \mathcal{L}_u)$; the continuity equation for $D\,d\,^3x$; the commutativity of differential $d$ 
and Lie derivative $ \mathcal{L}_u$; the antisymmetry of the wedge product, 
$da\wedge db=-db\wedge da=(\nabla a\times \nabla b)\cdot d\mathbf{S}$; the relation $d^2=0$, 
for example $d^2b = d(\nabla b \cdot d\mathbf{x})=({\rm curl} \nabla b )\cdot d\mathbf{S}=0$, and 
the Kelvin theorem form of the EP equation motion in \eqref{EP-Kthm}.

\begin{answerwide}
PV may be expressed naturally as a density which is intimately associated with the circulation in Kelvin's theorem, 
\begin{align*}
qD\,d^3x = db\wedge d\Big(\frac{1}{D}\frac{\delta \ell}{\delta \mathbf{u}}\cdot d\mathbf{x}\Big)
\,.\end{align*}
Taking the Eulerian transport time derivative $(\partial_t  +  \mathcal{L}_u)$, applying the product rule and 
substituting the continuity equation $(\partial_t  +  \mathcal{L}_u)(D\,d^3x)=0$ and advection of buoyancy $(\partial_t  +  \mathcal{L}_u)b=0$ yields 
\begin{align*}
\Big(\partial_t  +  \mathcal{L}_u\Big)(qD\,d^3x) &= \big(\partial_t q + \mathbf{u}\cdot\nabla q) (D\,d^3x) 
\\&= 
db\wedge d \bigg( \big(\partial_t  +  \mathcal{L}_u\big)\Big(\frac{1}{D}\frac{\delta \ell}{\delta \mathbf{u}}\cdot d\mathbf{x}\Big)\bigg)
= 
- db\wedge d \bigg( 
\frac{1}{D}\frac{\delta \ell}{\delta b} db\bigg)
\\&= 
- db\wedge d \bigg( 
\frac{1}{D}\frac{\delta \ell}{\delta b}\bigg) \wedge db
= 0
\,.\end{align*}
This derivation shows the robustness of PV advection. Namely, PV advection is independent of the details
of how the Lagrangian depends on the buoyancy, $b$. 
\end{answerwide}

\subsection*{Two equivalent representations of EP GFD motion equations}
Equations \eqref{EP-Kthm}--\eqref{pv-cons-EP} embody most of the panoply
of equations for GFD.  The vector form of equation \eqref{EP-Kthm} is,
\begin{align}\label{vec-EP-eqn1}
\underbrace{\Big({\partial\over \partial t} + \mathbf{u}\cdot\nabla\Big)
\Big({1\over D}{\delta l\over \delta \mathbf{u}}\Big)
+{1\over D}{\delta l\over \delta u^j}\nabla u^j}
_{\hbox{Geodesic Nonlinearity: Kinetic energy}}
= \underbrace{\nabla {\delta l\over \delta D}
-{1\over D} {\delta l\over \delta b}\nabla b}_{\hbox{Potential energy}}
\,.\end{align}
In geophysical applications, the Eulerian variable $D$ represents
the frozen-in volume element and $b$ is the buoyancy. In this case,
\emph{Kelvin's theorem} is
\[
\frac{d I}{dt}
=\int\!\!\!\int_{S(t)}
\nabla \left({1\over D} {\delta l\over \delta b}\right)
\times\nabla b \cdot d \mathbf{S} 
\,,\]
with circulation integral
\[
I=\oint_{\gamma(t)} {1\over D}{\delta l\over \delta \mathbf{u}}\cdot d\mathbf{x}
\,.\]
An alternative vector form of equation \eqref{vec-EP-eqn1} is obtained via a vector calculus identity,
\begin{align}\label{vec-EP-eqn2}
{\partial\over \partial t} \Big({1\over D}{\delta l\over \delta \mathbf{u}}\Big)
-  \mathbf{u}\times {\nabla}{\rm curl} \Big({1\over D}{\delta l\over \delta \mathbf{u}}\Big)
+ \nabla \Big(\mathbf{u}\cdot{1\over D}{\delta l \over \delta \mathbf{u}}\Big)
= \nabla {\delta l\over \delta D}
- {1\over D} {\delta l\over \delta b}\nabla b
\,.\end{align}
The first form of the Euler--Poincar\'e motion equation in \eqref{vec-EP-eqn1} uses the dynamic definition of Lie derivative,
while the equivalent second form in equation \eqref{vec-EP-eqn2} uses Cartan's geometric definition of Lie derivative.

\begin{remark}[\emph{Hamiltonian formulation for 3D EB}]$\,$\\ \textrm
In addition to the Casimir $C_\Phi$ in \eqref{pv-Casimir}, the Euler--Boussinesq fluid equations (\ref{EulBous-motion-eqn-Lie}) also conserve the total energy 
 \begin{equation}
 E =\int_D 
 \frac{1}{2}|\mathbf{u}|^2+bz 
 \
 d\,^3x
 \,,
\label{EulBouss-erg-def}
\end{equation}
which is the sum of the kinetic and potential energies.


\end{remark}

\section{Two additional examples of geometric fluid dynamics}

\subsection{Rotating shallow water (RSW) equations} 
Consider dynamics of rotating shallow water (RSW)  
on a two dimensional domain with horizontal planar coordinates $\mathbf{x}=(x,y)$.
This RSW motion is governed by the following nondimensional
equations for variables depending on $(\mathbf{x},t)$ comprising the horizontal fluid velocity vector $\mathbf{u}=(u,v)$ and the total depth is  
$\mathscr{D}(\mathbf{x},t)$,
\small
\begin{equation}   
\epsilon\frac{d}{dt}\mathbf{u}+ f (\mathbf{x}) \mathbf{\hat{z}} \times\mathbf{u}   
+ \nabla \zeta = 0\, ,   
\qquad 
\frac{\partial \mathscr{D}}{\partial t} + \nabla\cdot ( \mathscr{D}\mathbf{u} )= 0\, ,   
\label{rsw}   
\end{equation}
with notation
\begin{equation*}   
\frac{d}{dt}:=\left(\frac{\partial}{\partial t} + \mathbf{u}\cdot\nabla\right) 
\quad\hbox{and}\quad  
\zeta:=\bigg(\frac{\mathscr{D} - \mathscr{B}(\mathbf{x})}{\epsilon {\rm Fr}^2}\bigg)\, ,
\label{Notation}   
\end{equation*} 
where the Rossby number $\epsilon\ll1$ and Froude number ${\rm Fr}=O(1)$ are non-dimensional constants.
\normalsize

$\bullet$  These equations include spatially variable Coriolis parameter $f(\mathbf{x})\mathbf{\hat{z}}= \mathrm{curl}\mathbf{R}(\mathbf{x})$ and
mean depth $\mathscr{B}=\mathscr{B}(\mathbf{x})$.
\medskip

$\bullet$  Their geometric properties will be revealed efficiently as an exercise using the Euler--Poincar\'e approach.
\medskip

\textbf{Question}
\begin{enumerate}[(i)]
\item
Show that the RSW motion equation \eqref{rsw} follows as an Euler--Poincar\'e equation \eqref{EP-Kthm} \index{Euler--Poincar\'e!equation}
\[
\Big(\frac{\partial }{\partial t} +  \mathcal{L}_u\Big)
\Big(\frac{1}{D}\frac{\delta \ell}{\delta
\mathbf{u}}\cdot d\mathbf{x}\Big)
\,+\,\frac{1}{D}\frac{\delta \ell}{\delta b}\nabla b \cdot d\mathbf{x}
\,-\nabla\Big(\frac{\delta \ell}{\delta D}\Big)\cdot d\mathbf{x}
 = 0,
\]
arising from Hamilton's variational principle $ \delta S=0$ for the action integral,
\begin{align}
S=\int_0^T l(\mathbf{u},\mathscr{D}) dt \quad\hbox{and}\quad 
l(\mathbf{u},\mathscr{D}) = \int \frac{\epsilon}{2}\mathscr{D} |\mathbf{u}|^2 
+ \mathscr{D}\mathbf{u}\cdot\mathbf{R}(\mathbf{x}) - \frac{(\mathscr{D}-\mathscr{B}(\mathbf{x}))^2}{2\epsilon\cal{F}} \,d^2x\,.
\label{rsw-Lag}
\end{align}
in which $\mathscr{D}(\mathbf{x},t)\,d^2x$ is an advected quantity. 

\item
Use the Euler-Poincar\'e equations to show that the RSW equations satisfy Kelvin's circulation theorem
\[
\frac{d}{dt} \oint_{c_t} \mathbf{v} \cdot d \mathbf{x} = 0 \,,
\]
 with $\mathbf{v}=\epsilon \mathbf{u} + \mathbf{R}(\mathbf{x})$.

\item
Use the Euler-Poincar\'e equations to show that the RSW equations satisfy
\[
\big(\partial_t +  \mathcal{L}_u\big) d \big( \mathbf{v} \cdot d \mathbf{x}\big) = 0\,,
\]
 with $\mathbf{v}=\epsilon \mathbf{u} + \mathbf{R}(\mathbf{x})$. 

 \item
Show that $ d (\mathbf{v} \cdot d \mathbf{x}) = \omega \,d^2x$, with $\omega := \mathbf{\hat{z}}\cdot\mathrm{curl}\mathbf{v}$.

 \item
Use $(\partial_t +  \mathcal{L}_u)(\omega \,d^2x)=0$ obtained in the previous two parts to derive conservation of potential vorticity on fluid particles.

\end{enumerate}

\textbf{Answer}
\begin{enumerate}
\item
The Euler-Poincar\'e equations are 
\[
\big(\partial_t +  \mathcal{L}_u\big)\frac{1}{\mathscr{D}}\frac{\delta l}{\delta u} = d\frac{\delta l}{\delta \mathscr{D}}
\quad\hbox{and}\quad
\big(\partial_t +  \mathcal{L}_u\big) \big(\mathscr{D}\,d^2x\big) = 0\,,
\]
where $\mathscr{D}^{-1}\frac{\delta l}{\delta u}=(\epsilon \mathbf{u} + \mathbf{R}(\mathbf{x}))\cdot d \mathbf{x}=: \mathbf{v} \cdot d \mathbf{x}$ and 
$\nabla\Big(\frac{\delta \ell}{\delta D}\Big)\cdot d\mathbf{x}=d\frac{\delta l}{\delta \mathscr{D}}=d(\frac{\epsilon}{2} |\mathbf{u}|^2+\mathbf{u}\cdot\mathbf{R} - \zeta)$. Thus,
\[
\big(\partial_t +  \mathcal{L}_u\big) (\mathbf{v} \cdot d \mathbf{x})
= d\Big(\frac{\epsilon}{2} |\mathbf{u}|^2+\mathbf{u}\cdot\mathbf{R} - \zeta\Big)
\]
 with $\mathbf{v}=\epsilon \mathbf{u} + \mathbf{R}(\mathbf{x})$. 
\item
Integrating the previous equation around a loop moving with the fluid produces 
\[
\frac{d}{dt} \oint_{c_t} \mathbf{v} \cdot d \mathbf{x} =  \oint_{c_t} d\Big(\frac{\epsilon}{2} |\mathbf{u}|^2+\mathbf{u}\cdot\mathbf{R} - \zeta\Big) = 0 \,,
\]
 with $\mathbf{v}=\epsilon \mathbf{u} + \mathbf{R}(\mathbf{x})$. 
 \item
 The differential of the Euler-Poincar\'e equation yields with $\omega := \mathbf{\hat{z}}\cdot\mathrm{curl}\mathbf{v}$
\[
\big(\partial_t +  \mathcal{L}_u\big) (\omega d^2x)
= \big(\partial_t +  \mathcal{L}_u\big) d(\mathbf{v} \cdot d \mathbf{x})
= d^2\Big(\frac{\epsilon}{2} |\mathbf{u}|^2+\mathbf{u}\cdot\mathbf{R} - \zeta\Big) = 0
\]
upon commuting the differential $d$ with the Lie derivative and using $d^2=0$. 
\item
By direct computation,
\begin{align*}
d (\mathbf{v} \cdot d \mathbf{x}) = v_{i,j}dx^j\wedge dx^i 
&= v_{1,2}dx^2\wedge dx^1 + v_{2,1}dx^1\wedge dx^2
\\&= (v_{2,1} - v_{1,2})\,d^2x
= \mathbf{\hat{z}}\cdot\mathrm{curl}\mathbf{v}\,d^2x = \omega \,d^2x
\end{align*}
\item
We have $(\partial_t +  \mathcal{L}_u) (\omega d^2x)$ and $(\partial_t +  \mathcal{L}_u) (\mathscr{D} d^2x)$. Therefore, by the product rule for the evolutionary operator $(\partial_t +  \mathcal{L}_u)$ we have
\[
0 = (\partial_t +  \mathcal{L}_u)\Big( \frac{\omega}{\mathscr{D}} (\mathscr{D} d^2x) \Big) 
= \Big( (\partial_t +  \mathcal{L}_u) \frac{\omega}{\mathscr{D}}\Big) (\mathscr{D} d^2x)
+ \frac{\omega}{\mathscr{D}} (\partial_t +  \mathcal{L}_u) (\mathscr{D} d^2x)
\]
Since the second term vanishes via the continuity equation, $(\partial_t +  \mathcal{L}_u) (\mathscr{D} d^2x)$, the first term yields
\[
0 = (\partial_t +  \mathcal{L}_u) \frac{\omega}{\mathscr{D}} = \left(\frac{\partial}{\partial t} + \mathbf{u}\cdot\nabla\right)\frac{\omega}{\mathscr{D}}
\,.\quad \hbox{Hence,}\quad 
\frac{dq}{dt} = 0
\,,\quad \hbox{with}\quad q:= \omega/\mathscr{D}\,.
\]
This is conservation of potential vorticity on fluid particles.

\vspace{2cm}

\item
Derive the Euler-Poincar\'e equations when Hamilton's variational principle $ \delta S=0$ for the 
action integral.in \eqref{rsw-Lag} is modified to include the advected buoyancy $b(\mathbf{x},t)$ satisfying 
$\partial_t b + \mathbf{u}\cdot\nabla b = 0$, as follows,
\begin{align}
\begin{split}
S&=\int_0^T l(\mathbf{u},\mathscr{D},b) dt \quad\hbox{and}\quad 
\\
l(\mathbf{u},\mathscr{D},b) &= \int \frac{\epsilon}{2}\mathscr{D} |\mathbf{u}|^2 
+ \mathscr{D}\mathbf{u}\cdot\mathbf{R}(\mathbf{x}) - \frac{b\,\mathscr{D}\big(\mathscr{D}-2\mathscr{B}(\mathbf{x})\big)}{2\epsilon {\rm Fr}^2} \,d^2x\,.
\end{split}
\label{rsw-Lag-b}
\end{align}
In particular, write the modified evolution equation for $q=\omega/\mathscr{D}$.
\end{enumerate}

\subsection{Adiabatic compressible ideal fluid flow in 3D} 
In the case of adiabatic compressible ideal fluid flow in 3D, the action in Hamilton's
principle is given by
\begin{equation}
{S}_{red}=\int_0^T l(\mathbf{v},D,s)\,dt = \int_0^T   \left(\frac{ D}{2}\
|\mathbf{v}|^2 - De(D,s)\right)d^3x\,dt\,,
\label{mhdact}
\end{equation}
where $e(D,s)$ is the fluid's \emph{specific internal energy}, whose dependence
on the \emph{mass density} $D$ and \emph{specific entropy} $s$ is given as the \emph{equation of
state}. An isotropic medium satisfies the \emph{thermodynamic first law}, written in the form 
\index{specific internal energy}\index{equation of state}\index{thermodynamic first law}
\begin{equation}
de=-p\,d(D^{-1})+Tds
\,.\label{1st-law}
\end{equation}
The \emph{thermodynamic first law} defines pressure as $p(D,s):=-\,\partial e/\partial D^{-1}$ and temperature as $T(D,s):=\partial e/\partial s$. 
\\That is, ``volume flows to equalise pressure" and ``heat flows to equalise temperature". The variation of the Lagrangian $l(\mathbf{v},D,s)$ in (\ref{mhdact}) is
\begin{equation}
0 = \delta {S}_{red} = \int_0^T dt\, \int_0^T
D\mathbf{v}\cdot\delta\mathbf{v}-DT\,\delta s +\left(\frac{1}{2}
|\mathbf{v}|^2 - h(p,s)\right)\delta D\, d^3x\,.
\label{var-red}
\end{equation}
The quantity $h=e+p/D$ in \eqref{var-red} denotes the \emph{specific enthalpy}, which satisfies
$dh(p,s)=(1/D)dp+Tds$\,. \index{specific enthalpy}

The variations in equation \eqref{var-red} are given by
\begin{align}
\begin{split}
\delta\mathbf{v} &= \partial_t\xi - \mathbf{ad}_v \xi = \partial_t\xi -  \mathcal{L}_v \xi
\,,\quad 
\delta s = - \,\xi\cdot \nabla s = -  \mathcal{L}_\xi s
\,,
\\
\delta D &= - \,\nabla \cdot (D\xi) = -  \mathcal{L}_\xi (Dd^3x)
\,.
\end{split}
\end{align}
With these variations, Hamilton's principle yields \index{Hamilton's principle}
\begin{align*}
\begin{split}
0 &= \delta {S}_{red} = \int_0^T 
\Big\langle D{v} \,,\,\partial_t \xi - \mathrm{ad}_v \xi\Big\rangle\,dt 
\\&+  \int_0^T DT\,\xi \cdot \nabla s - \left(\frac{1}{2}
|\mathbf{v}|^2 - h(p,s)\right)\nabla \cdot (D\xi) \,d^3x\,dt
\\&= \int_0^T  \Big\langle -(\partial_t + \mathrm{ad}^*_v )(D{v} )\,,\,\xi \Big\rangle \,dt 
\\&+  \int_0^T \Big((DT\nabla s) \cdot \xi + \nabla\big(\tfrac{1}{2}
|\mathbf{v}|^2 - h(p,s)\big) \cdot (D\xi)\Big) \,d^3x\,dt
\\&= \int_0^T 
\Big\langle\underbrace{ - D\big(\partial_t +  \mathcal{L}_{v}\big)
\big({\mathbf{v}}\cdot d\mathbf{x}\big)
+ D Tds + Dd\big(\tfrac{1}{2} \mathbf{v}|^2 - h(p,s)\big)}_{\hbox{1-form density}}\,,\, \xi\Big\rangle \,dt \,.
\end{split}
\label{EPvar-Ham}
\end{align*}

The Euler--Poincar\'e formula in the Kelvin-Noether form
yields the compressible ideal fluid motion equation as
\begin{equation}
-\left(\frac{ \partial}{\partial t}+  \mathcal{L}_{v}\right)
\left({\mathbf{v}}\cdot d\mathbf{x}\right)
+ Tds
+ d\left(\frac{1}{2} |\mathbf{v}|^2 - h(p,s)\right) = 0\,,
\end{equation}
or, in three dimensional vector form after using $Tds - dh= - D^{-1}dp$ and 
$ \mathcal{L}_{v}{\mathbf{v}}\cdot d\mathbf{x}-d|\mathbf{v}|^2/2=\big((\mathbf{v}\cdot\nabla)\mathbf{v}\big)\cdot d\mathbf{x}$,
\begin{equation} 
\frac{ \partial\mathbf{v}}{\partial t} + (\mathbf{v}\cdot\nabla)\mathbf{v}
+\frac{ 1}{ D}\nabla p = 0\,.
\label{3DCcompressFlow}\end{equation}
By definition, the advected variables (specific entropy $s$ and mass density $Dd^3x$ here) satisfy 
the following Lie-derivative relations which close the ideal compressible ideal fluid 
system, 
\begin{align}
\begin{split}
\left(\frac{ \partial}{\partial t}+  \mathcal{L}_\mathbf{v}\right) s=0,
&\hbox{\quad or \quad}
\frac{ \partial s}{\partial t} + \mathbf{v}\cdot\nabla\,s = 0\,,
\\
\left(\frac{ \partial}{\partial t}+  \mathcal{L}_\mathbf{v}\right)(D d^3x)=0,
&\hbox{\quad or \quad}
\left(\frac{ \partial D}{\partial t} + \nabla\cdot(D\mathbf{v})\right)d^3x = 0,
\end{split}
\end{align}
and the pressure function $p(D,s) = D^2\partial e/\partial D$ is obtained from
the equation of state of the fluid, $e=e(D,s)$.
\smallskip


\section{Lie--Poisson brackets for deterministic ideal fluids}

\index{Lie algebra!semidirect product}\index{broken symmetry!semidirect product}\index{Lie--Poisson!bracket} 

\subsection{Semidirect products and symmetry breaking} \index{symmetry breaking!semidirect products}
The Hamiltonian representation of ideal fluid dynamics in general is dual to the \emph{semidirect product\/}, denoted $\circledS$,
of diffeomeorphisms acting on vector spaces (the advected quantities). \index{advected quantities}
In general, the Lie bracket for
semidirect product action $\mathfrak{g}\,\circledS\,V$ of Lie algebra $\mathfrak{g}$
on vector space $V$ is given by \index{semidirect product!Lie algebra}
\begin{align}
\big[(X,a), (\overline{X},\overline{a})\big]
=
\big([X,\overline{X}\,],\overline{X}(a)-X(\overline{a})\big)
\,,\label{SDPaction}
\end{align}
in which $X,\overline{X}\in\mathfrak{g}$ and $a,\overline{a}\in V$.
The action of the Lie algebra $\mathfrak{g}$ on the vector space $V$ is denoted, for
example, $X(\overline{a})$. Usually, the action $\mathfrak{g}\times V\to V$ would be the Lie derivative,
$ \mathcal{L}$, or its transpose  $ \mathcal{L}^T$. Let variables $\mu\in\mathfrak{g}^*$ and $b\in V^*$ be dual, 
respectively, to $X,\overline{X}\in\mathfrak{g}$ and $a,\overline{a}\in V$ as
\begin{align}
\left[\left(\frac{\delta F}{\delta \mu} , \frac{\delta F}{\delta b}\right), \left(\frac{\delta H}{\delta \mu},\frac{\delta H}{\delta b}\right)\right]
= \mp
\left( \left[\frac{\delta F}{\delta \mu},\frac{\delta H}{\delta \mu}\,\right] \,,\,
 \mathcal{L}^T_{\frac{\delta H}{\delta \mu}} \frac{\delta F}{\delta b} -  \mathcal{L}^T_{\frac{\delta F}{\delta \mu}} \frac{\delta H}{\delta b}\right)
\label{Dual-vars}
\end{align}
where in Lie symmetry reduction of the Lagrangian in Hamilton's principle one chooses $(-)$ resp. $(+)$ 
signs for Lie algebras which are invariant under right (resp. left) Lie-group action. 
\index{semidirect product!Lie-Poisson bracket} \index{Hamilton's principle!Lie symmetry reduction}
\index{Lagrangian!Lie symmetry reduction}

Consequently, one may write the Lie-Poisson bracket dual to the semidirect-product 
action of vector fields on vector spaces as, \vspace{-5mm}
\index{broken symmetry!semidirect product}
\begin{align}
\begin{split}
\Big\{ F,H \Big\}(\mu,b) :&= \mp\scp{(\mu,b)}{\left[\left(\frac{\delta F}{\delta \mu} , \frac{\delta F}{\delta b}\right), \left(\frac{\delta H}{\delta \mu},\frac{\delta H}{\delta b}\right)\right]}
\\&= \mp
\scp{(\mu,b)} { \left( \left[\frac{\delta F}{\delta \mu},\frac{\delta H}{\delta \mu}\,\right] ,
 \mathcal{L}^T_{\frac{\delta H}{\delta \mu}}\frac{\delta F}{\delta b} -  \mathcal{L}^T_{\frac{\delta F}{\delta \mu}}\frac{\delta H}{\delta b}\right) }
\\&= 
\pm\scp{\mu} {\mathrm{ad}_{\frac{\delta F}{\delta \mu}} \frac{\delta H}{\delta \mu} }
\pm\scp{ \mathcal{L}^T_{\frac{\delta H}{\delta \mu} }\frac{\delta F}{\delta b} -  \mathcal{L}^T_{\frac{\delta F}{\delta \mu}}\frac{\delta H}{\delta b}}{b}
\\&=
\pm\scp{\mathrm{ad}^*_{\frac{\delta H}{\delta \mu}}\mu} {\frac{\delta F}{\delta \mu}}
\pm \scp{\frac{\delta F}{\delta b}} {- \mathcal{L}_{\frac{\delta H}{\delta \mu} } b}
\mp \scp{\frac{\delta H}{\delta b}} {- \mathcal{L}_{\frac{\delta F}{\delta \mu}}b}
\\&=
\mp\scp{\mathrm{ad}^*_{\frac{\delta H}{\delta \mu}}\mu} {\frac{\delta F}{\delta \mu}}
\mp \scp { \mathcal{L}_{\frac{\delta H}{\delta \mu} } b} {\frac{\delta F}{\delta b}}
\mp \scp{\frac{\delta H}{\delta b}\diamond b} {\frac{\delta F}{\delta \mu} }
\,,\end{split}
\label{SDP-LPB}
\end{align}
where the diamond operation $(\diamond)$ in the last step is defined via integration by parts in the previous step.
The diamond operation $(\diamond)$ figures prominently in both the implications of symmetry of the Lagrangian 
in Hamilton's principle. The increase of that symmetry leads to conservation laws by Noether's theorem 
and the decrease of that symmetry leads to additional forces by the Euler--Poincar\'e theorem. 

For example, the Lie-Poisson Hamiltonian equations for semidirect-product action in equation \eqref{SDP-LPB} 
are written in their matrix operator form below.
In these equations, the diamond term $\frac{\delta H}{\delta b}\diamond b = - b \nabla \frac{\delta H}{\delta b}$ 
in the last line of \eqref{SDP-LPB} represents the force of vertical gravity acting on horizontal gradients of buoyancy, 
\index{Lie algebra!semidirect product} \index{Lie--Poisson bracket!semidirect product}
\index{Lie--Poisson!equations} \index{diamond operation $(\diamond)$} 
\begin{align}
\frac{d}{dt} 
\begin{bmatrix}
\mu \\ b
\end{bmatrix}
=
\mp
\begin{bmatrix}
\mathrm{ad}^*_{\Box}\mu & \Box \diamond b 
\\
 \mathcal{L}_\Box b & 0
\end{bmatrix}
\begin{bmatrix}
{\delta H}/{\delta \mu} \\  {\delta H}/{\delta b}
\end{bmatrix}
=
\mp
\begin{bmatrix}
\mathrm{ad}^*_{\frac{\delta H}{\delta \mu}}\mu + \frac{\delta H}{\delta b}\diamond b
\\  \mathcal{L}_{\frac{\delta H}{\delta \mu} } b
\end{bmatrix}
,
\label{SDP-LPB-eqns}
\end{align}
with $(-)$ resp. $(+)$ signs for Lie algebras invariant under right (resp. left) Lie-group action \cite{HolmGM2025}.
\index{broken symmetry!semidirect product}

\subsection{Example: Lie-Poisson Hamiltonian structure for deterministic 2D EB equations} 

\paragraph{2D EB Notation.}
Upon introducing local temperature notation $\Theta(x,z,t)$ as the buoyancy quantity in a vertical plane, the 
Euler-Boussinesq (EB) equations \eqref{EulBouss-motion-eqn-Newton} in a 2D vertical $(x,z)$ plane without rotation reduce to
\begin{align}
\begin{split}
\frac{\partial}{\partial t}\textbf{u} + \textbf{u}\cdot\nabla \textbf{u} &= -\nabla p +   \Theta\hat{\textbf{z}}
\,,\quad
(\p_t +  \mathcal{L}_u)u^\flat = d(\tfrac12|\textbf{u}|^2 - p) + \Theta\,dz
\,,\\
\frac{\partial}{\partial t} {\Theta} + \textbf{u}\cdot \nabla {\Theta} &= 0
\,,\hspace{27mm}
(\p_t +  \mathcal{L}_u) \Theta = 0
\,,\\
\div\mathbf{u} &= 0
\,,\hspace{25.66mm}
 \mathcal{L}_u (dx\wedge dz) = d(u\contract (dx\wedge dz)) = 0
\,.
\end{split}
\label{DEBeqns}
\end{align}
Taking the curl of the 2D EB motion equation in \eqref{DEBeqns} yields the equations,
\begin{align}
\begin{split}
\partial_t \omega + \bm{u}\cdot \nabla \omega
&= \Theta_x 
\,,
\quad \hbox{since}\quad \star d(\Theta dz) = \star \Theta_x dx\wedge dz = \Theta_x
\,,\\
\partial_t \Theta + \bm{u}\cdot \nabla \Theta &= 0
\,.\end{split}
\label{EB-eqns-def1}
\end{align}
In coordinates, one also has the relations
\begin{align}
\begin{split}
\bm{u} &= -\,\bs{\hat{y}}\times \nabla\psi =: \nabla^\perp\psi = (-\,\psi_z\,,\, \psi_x)\,,
\quad\hbox{so}\quad
{\div}\bm{u} = 0
\,,\\
\quad\hbox{with}\quad 
\omega &:= -\,\bs{\hat{y}}\cdot {\curl}\bm{u} = \Delta \psi  = \psi_{xx} + \psi_{zz}
\quad\hbox{and}\quad 
-\,\bs{\hat{y}}\cdot {\curl}(\Theta \hat{\textbf{z}}) = \Theta_x
\,.
\end{split}
\label{EB-eqns-def2}
\end{align}
Here, the quantity $\omega:=-\,\bs{\hat{y}}\cdot {\curl}\bm{u} = \Delta \psi$ is the $\bs{\hat{x}}\times\bs{\hat{z}}=-\,\bs{\hat{y}}$ component of the vorticity, pointing \emph{into the $(x,z)$ plane}. We may regard both $\omega$ and $\Theta$ as scalar functions. 
 
\paragraph{Jacobian notation.}
The misalignment of the gradients $\nabla a$ and $\nabla b$ of scalar functions $a,b$ in a vertical (x,z) plane may be expressed by the Jacobian, denoted
\begin{align}
J(a,b):=-\,\bs{\hat{y}}\cdot \nabla a\times \nabla b = a_x b_z-a_z b_x =: \{a,b\}
\quad\hbox{or equivalently,}\quad
J(a,b)dx\wedge dz = da\wedge db\,.
\label{def-Jac}
\end{align}
In terms of the Jacobian and the local temperature $\Theta(x,z,t)$, the 2D EB equations in \eqref{EB-eqns-def1} may be expressed as
\begin{align}
\begin{split}
\partial_t \omega  &= J(\omega,\psi) + J(\Theta,z)  \,,
\\
\partial_t \Theta &= J(\Theta, \psi )
\,.
\end{split}
\label{EB-eqns-def3}
\end{align}
The vertical boundary conditions are $\bu\cdot\bs{\wh{n}} = 0 $ and $\bs{\wh{n}}\times \nabla \Theta = 0$. Equivalently, one may take  $\psi\big|_{z=0}=0$, $\psi_{z=H}=0$ 
and $\Theta\big|_{z=0}=\Theta_0>0$,  $\Theta\big|_{z=H}=0$; and periodic in the horizontal coordinate $x$. 

\paragraph{Kelvin circulation theorem for 2D EB.}
2D EB circulation dynamics for velocity 
$\bu=\nabla^\perp\psi =  (-\,\psi_z\,,\, \psi_x)$ and  vorticity $\omega=\Delta \psi= \psi_{xx}+\psi_{zz}$ 
 follows immediately from equation \eqref{EB-eqns-def1} as 
\begin{align} 
\frac{d}{dt}\oint_{c(u)} \!\!\!\! \mbf{u}\cdot d\bx
= 
\frac{d}{dt}\int\!\!\!\!\int_{\p S = c(u)} \!\!\!\! \omega \,dxdz
=  \int\!\!\!\!\int_{\p S = c(u)} \!\!\!\!  dz\wedge d{\Theta}
\,,\label{EB-KelvinThm}
\end{align}
for a closed material loop, $c(u)$, moving with the flow velocity $\bu=\nabla^\perp\psi$. 
\smallskip

Hence, in 2D EB, circulation in a vertical plane is driven by misalignment between the gradients $\nabla z$ and $\nabla {\Theta}$.

\paragraph{Conservation Laws for 2D EB.}
The 2D EB equations in \eqref{EB-eqns-def2} preserve the sum of the kinetic and thermal energies of the EB system, made definite in sign by completing squares to find
\begin{align}
\begin{split}
H_{EB}(\omega,\Theta) 
&= \int  \frac12\,  \omega \Delta^{-1}\omega - z \Theta \,dxdz
\\&= \frac12 \int  \omega \Delta^{-1}\omega + \big(\Theta -   z\big)^2  dxdz
- \frac12\int(\Theta^2 +  z^2) \,dxdz
\,.
\end{split}
\label{EB-eqns-Ham}
\end{align}
The last (negative) term in the second line of \eqref{EB-eqns-Ham} is immaterial, because $\int z^2dxdz$ is a non-dynamical constant and, as mentioned next, $\int{\Theta}^2dxdz$ is a Casimir in the Lie-Poisson Hamiltonian formulation of these EB equations.

The EB equations in \eqref{EB-eqns-def2} also preserve an infinity of integral conservation laws, determined by two arbitrary differentiable functions of temperature $\Phi(\Theta)$ and $\Psi(\Theta)$ as
\begin{equation}
\mc{C}_{\Phi,\Psi}(\Theta,\omega) = \int_{\mathcal D}\Phi(\Theta ) + \omega \Psi(\Theta )\,\mathrm{d}x\mathrm{d}z\,.
\label{eq:casimirsEB}
\end{equation}
That the EB dynamics in \eqref{EB-eqns-def2} preserves the energy in \eqref{EB-eqns-Ham} and the family of integral quantities in \eqref{eq:casimirsEB} can be verified by direct computations. However, the infinity of conservation laws in \eqref{eq:casimirsEB} also indicates that the EB system in \eqref{EB-eqns-def1} possesses a rich mathematical structure which will help in diagnosing its solution behaviour. We investigate the geometrical aspects of this mathematical structure in the next few paragraphs.

\paragraph{Lie-Poisson Hamiltonian form of 2D EB.}
The conservation laws in $\mc{C}_{\Phi,\Psi}$ in \eqref{eq:casimirsEB} signal a deeper mathematical content. In particular, the EB equations \eqref{EB-eqns-def3} may be expressed in Lie-Poisson Hamiltonian form, in terms of the semidirect-product Lie-Poisson bracket given in equation \eqref{SDP-LPB}, as
\begin{equation}
 \frac{\partial (\omega \,; \Theta)}{\partial t}  = \big\{(\omega \,;  \Theta)\,,\, \mc{H}_{EB}\big\},
\label{short-note}
\end{equation}
written in the matrix form as
\begin{equation}
 \frac{\partial }{\partial t} 
\begin{bmatrix}
\omega \\ \Theta
\end{bmatrix}
=
\begin{bmatrix}
J(\omega,\,\cdot\,) & J(\Theta,\,\cdot\,) 
\\ 
J(\Theta,\,\cdot\,)  & 0
\end{bmatrix}
\begin{bmatrix}
{\delta \mc H_{EB}}/{\delta \omega} = \psi
\\ 
{\delta \mc H_{EB}}/{\delta \Theta} = \Theta - z
\end{bmatrix}
=
\begin{bmatrix}
J(\omega,\psi) + J(\Theta,z)  
\\
J(\Theta, \psi )
\end{bmatrix} 
\,.
\label{EB-LPHam}
\end{equation}

\paragraph{Energy conservation of 2D EB flows.}
\begin{itemize}
\item
The Hamiltonian $\mc{H}_{EB}$ is conserved because the Lie-Poisson matrix operator defined in \eqref{EB-LPHam} is skew symmetric in the $L^2$ pairing, so that $\frac{d}{dt}\mc{H}_{EB} = \{\mc{H}_{EB},\mc{H}_{EB}\}=0$. 
\item
The quantities $\mc{C}_{\Phi,\Psi}$ in \eqref{eq:casimirsEB} are conserved because their variational derivatives are null eigenvectors of the Lie-Poisson matrix operator in \eqref{EB-LPHam}. The type of conserved quantity such as  $\mc{C}_{\Phi,\Psi}$ that Poisson commutes with any Hamiltonian are called the \emph{Casimir functions} of the Lie-Poisson bracket \cite{HMR1998}.
\end{itemize} \index{Casimirs!Lie-Poisson bracket} \index{Lie-Poisson bracket!Casimirs}

\paragraph{Geometric considerations of 2D EB flows.} Upon identifying $(\mu,b)=(\omega,\Theta)$ and writing the matrix form of equation \eqref{SDP-LPB-eqns} to recover equation \eqref{EB-LPHam}  in terms of  $(\omega,\Theta)$, the EB model is found to fit into the standard semidirect-product Lie--Poisson Hamiltonian framework for all ideal fluids with advected quantities, including their Casimir conservation laws. Remarkably, the same universal geometric considerations of 2D EB flows transfer from these deterministic considerations to both the SALT and LA-SALT stochastic dynamics we shall consider next.

\subsection{Geometric properties of the Lie-Poisson bracket for 2D EB dynamics}

\begin{proposition}[Jacobi identity]\label{LPB-EB5}
The Lie-Poisson bracket in \eqref{short-note} satisfies the Jacobi identity.  
\end{proposition}

$\bullet$  This proposition could be demonstrated by direct computation using the properties of the Jacobian of function pairs. 
\medskip

$\bullet$  The proof given here, though, will illustrate the geometric properties of the 2D EB system in \eqref{EB-eqns-def2} and thereby place it into the wider class of ideal fluid dynamics with advected quantities. \medskip \index{advected quantities}

$\bullet$  In particular, the proof will identify the Poisson bracket in \eqref{short-note} as being defined over domain $\mathcal D$ on functionals of the \emph{dual space}\footnote{Dual with respect to the  $L^2$ pairing on ${\mathcal D}$} given by $(f_1\circledS f_2)^*$ of the Lie algebra  $(f_1\circledS f_2)$ of semidirect-product symplectic transformations. 
\medskip

$\bullet$  Thus, EB dynamics is understood as coadjoint motion generated by the semidirect-product action of the Lie algebra $(f_1\circledS f_2)$ on function pairs $(f_1;f_2)\in (f_1\circledS f_2)^*$. 
\medskip

$\bullet$  This proof of coadjoint motion also identifies the potential vorticity and buoyancy, $(\omega;b)$, as a semidirect-product momentum map.  \index{momentum map!semidirect product} \index{semidirect product!momentum map}

\begin{proof}
The Lie algebra commutator $[\,\cdot\,,\,\cdot\,]$ action for the adjoint (ad) representation of the action of the Lie algebra $f_1\circledS f_2$ on itself is defined by \cite{HMR1998}
\begin{align}
{\mathrm{ad}}_{\big(\overline{f}_1;\overline{f}_2\big)}\big(f_1;f_2\big) 
=
\Big[ \big(f_1;f_2\big) , \big(\overline{f}_1;\overline{f}_2\big)\Big] 
:=
\Big( \big[ f_1,\overline{f}_1 \big] ; \big[ f_1,\overline{f}_2 \big] - \big[ \overline{f}_1,f_2 \big]\Big)\,,
\label{SDP-ad-action-fns}
\end{align}
where the commutator $[\,\cdot\,,\,\cdot\,]$ is given by the Jacobian of the functions $f_1$ and $\overline{f}_2$. For example,
\begin{align}
[ f_1,\overline{f}_2 ]:= J(f_1,\overline{f}_2)
\,,\label{SDP-commutator-fns}
\end{align}
which is also the commutator of symplectic vector fields. Thus, the adjoint (ad) action in \eqref{SDP-ad-action-fns} is the semidirect product Lie algebra action among symplectic vector fields. 
\bigskip

The definition in \eqref{SDP-ad-action-fns} of the semidirect product Lie algebra action among functions defined on the plane $\mathbb{R}^2$ enables the Lie-Poisson bracket \eqref{EB-brkt3} for functionals of $(\omega,\Theta)$ defined on domain $\mathcal D$ to be identified with the coadjoint action of this Lie algebra. 

\smallskip
This is because the variational derivatives of such functionals live in the Lie algebra of symplectic vector fields on domain $\mathcal D$. Hence, the Lie--Poisson bracket will be defined by the pairing 
\begin{align}
\begin{split}
 \frac{\mathrm{d}}{\dt} { \mathcal{F}}(\omega,\Theta)
 =
\Big\{\mathcal{F}\,,\,{\mathcal{H}}\Big\}(\omega,\Theta)
&=
- \Bigg\langle
\big(\omega ; \Theta\big) \,,\,\Bigg[ \bigg(\frac{\delta \mathcal{F}}{\delta \omega} ; \frac{\delta \mathcal{F}}{\delta \Theta}\bigg), 
\bigg(\frac{\delta \mathcal{H}}{\delta \omega} ;\frac{\delta \mathcal{H}}{\delta \Theta} \bigg)\Bigg] \Bigg\rangle
\\&=:
- \Bigg\langle \big(\omega ; \Theta\big)
\,,\,
{\mathrm{ad}}_{\big(\frac{\delta \mathcal{H}}{\delta \omega};\frac{\delta \mathcal{H}}{\delta \Theta}\big)}
\bigg(\frac{\delta \mathcal{F}}{\delta \omega};\frac{\delta \mathcal{F}}{\delta \Theta}\bigg)
 \Bigg\rangle
\\&=:
- \Bigg\langle {\mathrm{ad}}^*_{\big(\frac{\delta \mathcal{H}}{\delta \omega};\frac{\delta \mathcal{H}}{\delta \Theta}\big)}\big(\omega ; \Theta\big)
\,,\,
\bigg(\frac{\delta \mathcal{F}}{\delta \omega};\frac{\delta \mathcal{F}}{\delta \Theta}\bigg)
 \Bigg\rangle
\,.
\end{split}
\label{SDP-LPB-EB}
\end{align}
The angle brackets $\langle\,\cdot\,,\,\cdot\,\rangle$ represent the pairing of the Lie algebra commutator of functions $f_1\circledS f_2$ in \eqref{SDP-ad-action-fns} with its dual Lie algebra $(f_1\circledS f_2)^*$ via the $L^2$ pairing on the domain $\mathcal D$ as will be shown below, in equation \eqref{EB-brkt3}. 

The calculation below will enable us to rewrite the EB equations in \eqref{EB-eqns-def1} in terms of this coadjoint action and thereby reveal its geometric nature as,
\begin{align}
\frac{\partial \big(\omega ; \Theta\big)}{\partial t}
= - \,{\mathrm{ad}}^*_{\big(\frac{\delta \mathcal{H}}{\delta \omega};\frac{\delta \mathcal{H}}{\delta \Theta}\big)}
\big(\omega ; \Theta\big)
= - \,{\mathrm{ad}}^*_{\big(\psi\,;\, (\Theta -z)\big)}
\big(\omega ; \Theta\big)
\,.
\label{SDP-EB-LP-ad-star}
\end{align}
This calculation showcases 2D EB dynamics as the semidirect-product coadjoint action shown in equation \eqref{SDP-LPB-EB} of symplectic vector fields acting on the momentum map $(\omega ; \Theta)$ defined in the dual space of vorticity and buoyancy functions. 

For completeness, it remains to write the coadjoint (ad$^*$) representation of the (right) action of the Lie algebra $f_1\circledS f_2$ on its dual Lie algebra explicitly in terms of the Jacobian operator between pairs of functions as 
\begin{align}
\begin{split}
\Big\langle \big(\omega ; \Theta\big) \,,\,
{\mathrm{ad}}_{\big(\overline{f}_1;\overline{f}_2\big)}\big(f_1;f_2\big) 
\Big\rangle
&:= -\,
\Big\langle \big(\omega ; \Theta\big) \,,\,
\Big[ \big(f_1;f_2\big) , \big(\overline{f}_1;\overline{f}_2\big)\Big] 
\Big\rangle
\\
\hbox{By \eqref{SDP-ad-action-fns} and \eqref{SDP-commutator-fns}}\quad
&:= -\,
\Big\langle \big(\omega ; \Theta\big) \,,\,
\Big( J\big( f_1,\overline{f}_1 \big) ; J\big( f_1,\overline{f}_2 \big) 
- J\big( \overline{f}_1,f_2 \big)\Big)
\Big\rangle
\\
\hbox{Defining the SDP pairing}\quad
&:= -\,
\int_{\mathcal D} 
\omega \,J\big( f_1 , \overline{f}_1 \big)
+ \Theta \Big( J\big( f_1,\overline{f}_2 \big) - J\big( \overline{f}_1,f_2 \big) \Big)
\,\mathrm{d}x \mathrm{d}y
\\& = -\,
\int_{\mathcal D} 
\omega \,J\big( f_1,\overline{f}_1 \big)
+  \Big( \Theta J\big( f_1,\overline{f}_2 \big) + \Theta J\big( f_2 , \overline{f}_1 \big) \Big)
\,\mathrm{d}x \mathrm{d}y
\\\hbox{By integrating by parts}\quad
& = 
\int_{\mathcal D} 
f_1\,\Big( J\big( \omega ,\overline{f}_1 \big) +  J\big( \Theta  ,\overline{f}_2 \big)\Big) 
- f_2J\big( \Theta , \overline{f}_1 \big) 
\,\mathrm{d}x \mathrm{d}y
\\
\hbox{Applying the SDP pairing}\quad
&= 
\Big\langle 
\Big( J\big( \omega,\overline{f}_1 \big) + J\big( \Theta ,\overline{f}_2\big)  \,;\, - J\big( \Theta,\overline{f}_1 \big)
\Big)
\,,\,
\big( f_1 ; f_2\big) 
\Big\rangle
\\
\hbox{By \eqref{SDP-LPB-EB}}\quad
&=: 
\Big\langle {\mathrm{ad}}^*_{\big(\overline{f}_1;\overline{f}_2\big)}\big(\omega ; \Theta\big) \,,\,
\big(f_1;f_2\big) 
\Big\rangle
\quad\hbox{with}\quad
\big(\overline{f}_1;\overline{f}_2\big)
=\Big(\frac{\delta \mathcal{H}}{\delta \omega};\frac{\delta \mathcal{H}}{\delta \Theta}\Big)
\,.
\end{split}
\label{SDP-adstar-action-fns}
\end{align}
\noindent
$\bullet$  The calculation in \eqref{SDP-adstar-action-fns} identifies the bracket in \eqref{EB-brkt3} as the Lie-Poisson bracket for functionals of $(\omega ; \Theta)$ defined on the dual $(f_1\circledS f_2)^*$ of semidirect-product Lie algebra $f_1\circledS f_2$, whose commutator is defined in equation  \eqref{SDP-ad-action-fns}. 
\end{proof}

\subsection{Casimir conservation laws for 2D EB flows}\label{subsec: Casimirs} \index{Casimirs}
We are now in a position to explain the conservation laws for the Hamilton matrix operator in \eqref{EB-LPHam}. Namely, those conservation laws comprise Casimir functions $\mc{C}_{\Phi,\Psi}(\omega,\Theta)$ whose variational derivatives are null eigenvectors of the Hamilton matrix operator in \eqref{EB-eqns-Ham} which defines the \emph{semidirect-product Lie-Poisson bracket} as
 \begin{align}
  \begin{split}
 \frac{\mathrm{d}}{\dt} \mc{F}(\omega,\Theta)
 &=
\int_{\mathcal D}
\begin{bmatrix}
{\delta \mc F}/{\delta \omega} 
\\ 
{\delta \mc F}/{\delta \Theta} 
\end{bmatrix}^T
\begin{bmatrix}
J(\omega,\,\cdot\,) & J(\Theta,\,\cdot\,) 
\\ 
J(\Theta,\,\cdot\,)  & 0
\end{bmatrix}
\begin{bmatrix}
{\delta \mathcal{H}}/{\delta \omega} 
\\ 
{\delta \mathcal{H}}/{\delta \Theta} 
\end{bmatrix}
\mathrm{d}x\mathrm{d}z
=:
\Big\{\mc F\,,\,{\mathcal{H}}\Big\}(\omega,\Theta)
\\&=-
\int_{\mathcal D} 
\omega \,J\bigg(\frac{\delta \mc F}{\delta \omega},\frac{\delta \mathcal{H}}{\delta \omega}\bigg)
+ \Theta \Bigg[J\bigg(\frac{\delta \mc F}{\delta \Theta},\frac{\delta \mathcal{H}}{\delta \omega}\bigg)
- J\bigg(\frac{\delta \mc F}{\delta \omega},\frac{\delta \mathcal{H}}{\delta \Theta}\bigg)\Bigg] \mathrm{d}x\mathrm{d}z
\,.
\end{split}
\label{EB-brkt3}
\end{align}
In particular, when $\mc{F}(\omega,\Theta)=\mc{C}_{\Phi,\Psi}(\omega,\Theta)= \int_{\mathcal D}\Phi(\Theta ) + \omega \Psi(\Theta )\,\mathrm{d}x\mathrm{d}z$ in \eqref{eq:casimirsEB} we have   
\begin{align}
  \begin{split}
 \frac{\mathrm{d}}{\dt} \mc{C}_{\Phi,\Psi}(\omega,\Theta)
 &=
\Big\{\mc{C}_{\Phi,\Psi}\,,\,{\mathcal{H}}\Big\}(\omega,\Theta)
\\&=-
\int_{\mathcal D} 
\omega \,J\bigg(\Psi(\Theta),\frac{\delta \mathcal{H}}{\delta \omega}\bigg)
+  \Bigg(\Theta J\bigg(\Phi'(\Theta)+\omega\Psi'(\Theta),\frac{\delta \mathcal{H}}{\delta \omega}\bigg)
- \Theta J\bigg(\Psi(\Theta),\frac{\delta \mathcal{H}}{\delta \Theta}\bigg)\Bigg) \mathrm{d}x\mathrm{d}z
\\&=-
\int_{\mathcal D} 
\frac{\delta \mathcal{H}}{\delta \omega}J\bigg(\omega ,\Psi(\Theta)\bigg)
+ \Bigg( \frac{\delta \mathcal{H}}{\delta \omega}J\bigg(\Theta,\Phi'(\Theta)+\omega\Psi'(\Theta)\bigg)
- \frac{\delta \mathcal{H}}{\delta \Theta}J\bigg(\Theta,\Psi(\Theta)\bigg)\Bigg) \mathrm{d}x \mathrm{d}y
\,.
\\&=-
\int_{\mathcal D} 
\frac{\delta \mathcal{H}}{\delta \omega}\Big(J\big(\omega ,\Theta\big)
+  J\big(\Theta,\omega\big)\Big)
\Psi'(\Theta) \mathrm{d}x \mathrm{d}y
\\& = 0 \quad\hbox{for all}\quad {\mathcal{H}}(\omega,\Theta)
\,.
\end{split}
\label{EB-CasimirBrkt3}
\end{align}

\begin{remark}[Casimir functions]$\,$\\
$\bullet$  The previous calculation proves the conservation of $\mc{C}_{\Phi,\Psi}(\omega,\Theta)= \int_{\mathcal D}\Phi(\Theta ) + \omega \Psi(\Theta )\,\mathrm{d}x\mathrm{d}z$ in \eqref{eq:casimirsEB} for an arbitrary Hamiltonian ${\mathcal{H}}(\omega,\Theta)$. These quantities are called \emph{Casimir functions} and as shown in \eqref{EB-CasimirBrkt3} their conservation follows because their variational derivatives comprise { null eigenvectors} of the semidirect-product Lie-Poisson bracket in \eqref{EB-CasimirBrkt3}. 
\smallskip

$\bullet$  Lie-Poisson brackets on dual Lie algebras generate coadjoint orbits on level sets of the bracket's Casimir functions. \\This feature limits the function space available to 2D EB solutions. In particular, the 2D EB solutions are restricted to stay on the same Casimir level set $\mc{C}_{\Phi,\Psi}(\omega ,\Theta)=const$ as that of their initial conditions.
\end{remark}

\begin{remark}
As we will discover next, the geometric mechanics properties such as semidirect-product action and its dual action in the Lie--Poisson operator in \eqref{SDP-LPB-eqns} for deterministic fluid equations derived from deterministic Hamilton's principles also apply for the stochastic cases. For the stochastic cases, we will also exemplify the geometric mechanics properties for fluid dynamics in the cases of 2D and 3D Euler Boussinesq (EB) equations.
\end{remark}


\newpage
\part{3D Euler Boussinesq (EB) SALT equations}\label{EB-SALT}

\section*{Key words}
\begin{center}
$\bullet\,$ Geometric mechanics $\qquad$ $\bullet\,$ Fluid dynamics $\qquad$ $\bullet\,$ Stochastic partial differential equations\\
$\bullet\,$ Lie groups $\qquad$ $\bullet\,$ Diffeomorphism group $\qquad$ $\bullet\,$ Sobolev spaces $\qquad$ $\bullet\,$ Momentum maps
\end{center}

\section{The origins of transport noise} \index{transport noise}

{\bf  History of transport noise in turbulence theory and stochastic mathematics.}
Transport noise in fluid dynamics has several forebears in turbulence theory, especially in the works of Kraichnan \cite{kraichnan1994anomalous} and Pope \cite{pope1994lagrangi}, who each sought in the early 1990's to explain the anomalous spectral scaling  in turbulent flows (which signals non-Gaussian intermittency) by computationally simulating stochastic transport of passive scalars (Kraichnan) and stochastic Lagrangian particle paths (Pope). For recent discussions of the Kraichnan model, see, e.g., \cite{flandoli2023stochastic,gess2025stabilization}. At about the same time, Brze\'zniak, Capi\'nski and Flandoli, \cite{brzezniak1991stochastic,brzezniak1992stochastic} introduced stochasticity into \emph{active} transport of vorticity in the 2D Navier-Stokes equation. See also Cipriano and Cruzeiro \cite{cipriano2007navier} who established a stochastic variational principle for the 2D Navier-Stokes equation. Also see Constantin and Iyer \cite{constantin2008stochastic} who derived a probabilistic representation of the deterministic 3D Navier-Stokes equations based on stochastic Lagrangian paths. Stochasticity eventually found its way into both the motion and transport equations for 3D fluid dynamics via Hamilton's variational principle in Holm \cite{Holm2015} which provided a unifying theory of stochastic transport which preserves Kelvin's circulation theorem for 3D ideal fluid dynamics of all types; including, e.g., compressibility, buoyancy, magnetic fields and other advected quantities. This variational stochastic model for fluid transport arising from Lie group symmetries of a stochastic Hamilton's principle is called Stochastic Advection by Lie Transport, whose acronym is SALT. For an alternative derivation of the SALT equations, see \cite{cotter2017stochastic}.
\index{SALT!history} \index{SALT!advected quantities}

SALT has provided one of the main assets in the STUOD program for modelling transport noise in fluid dynamics and evaluating methods of stochastic data assimilation with a variety of test models, in preparation for applying them to satellite observations of transport in upper ocean dynamics. For a description of the STUOD program's success in determining the analytical properties  of SALT, see \cite{crisan2022variational,crisan2022solution,CrFlHo2019}. For a summary of the development of the particle filtering method for applying data assimilation methods to computational simulations of SALT at various levels of fluid dynamics approximations, 
see \cite{cotter2018modelling,CCHOS18a,CCHOS18b,cotter2019particle,cotter2019numerically}.

{\bf  Hamiltonian Stochastic Differential Equations (SDE).}
Stochasticity of finite dimensional system of was made Hamiltonian by Bismut \cite{bismut1982mecanique} and stochastic Hamiltonian dynamical systems were reformulated in the Lie-Poisson framework by L\'azarro-Cam\'i and Ortega in \cite{lazaro2008stochastic}. 

{\bf  Hamilton's Principle for Finite-Dimensional Dynamics (SDE).}
It was natural to expect that stochastic Hamiltonian dynamical systems should also emerge from a phase-space version of Hamilton's principle, although a larger class of variations $(\delta q,\delta p)$ would be necessary in phase space. A variant of such a theorem was proved in L\'azarro-Cam\'i and Ortega in \cite{lazaro2008stochastic}. Another variant for Lagrangian systems was proved by Bou-Rabee and Owhadi \cite{bou2009stochastic} in the special case when the Hamiltonian in the Legendre transformation \index{Legendre transformation} is independent of the momentum. The approach of Holm and Tyranowski \cite{holm2018stochastic} presents a general methodology for deriving new structure-preserving numerical schemes and enables recasting a number of integrators previously studied in the literature so that: (i) the resulting integrators are symplectic; (ii) they preserve integrals of motion related to Lie group symmetries; and (iii) they include stochastic symplectic Runge--Kutta methods as a special case. 
These discoveries for finite-dimensional dynamics quickly led Cruziero, Holm and Ratiu \cite{cruzeiro2018momentum} to derive the entire spectrum of finite-dimensional geometric mechanics from reduction by symmetry of the Lagrangian to derive a stochastic version of the Euler-Poincar\'e variational principle arising from symmetry reduction in Hamilton's principle for Lie group invariant Lagrangians by following the methods developed in Holm, Marsden and Ratiu \cite{HMR1998}. 
\index{Geometric Mechanics!reduction by symmetry}

{\bf  Hamilton's Principle for Infinite-Dimensional Dynamics (SPDE).}
Following the recognition of Lie-group reduction of Hamilton principles for deriving finite-dimensional stochastic Euler-Poincar\'e (EP) variational principles in Cruzeiro, Holm and Ratiu \cite{cruzeiro2018momentum}, it was natural to attempt the same feat for the Eulerian representation of fluid dynamics. An EP variational principle was derived for ideal fluid systems in Holm \cite{Holm2015} and a variational derivation of the Navier-Stokes equation was given in Chen, Cruzeiro and Ratiu \cite{chen2023stochastic}. As it turned out, in 2D, these variational approaches recovered the planar stochastic Euler vorticity equations with transport noise studied earlier by Brze\'zniak, Capi\'nski and Flandoli, \cite{brzezniak1991stochastic,brzezniak1992stochastic}.

In fact, when all is said and done in following the EP approach, the stochasticity appears only in the transport velocity of the material loop in Kelvin's theorem for the conservation of circulation around material loops in 3D Euler fluid dynamics. 
Thus, the EP variational approach is natural for deriving stochastic fluid dynamics equations with transport noise.

{\bf  Stochastic analysis of the EP 3D Euler fluid equations with transport noise.}
The next natural step was then for Crisan, Flandoli and Holm \cite{CrFlHo2019} to examine the analytical properties of the newly derived EP Euler fluid equations with transport noise in three dimensions. Remarkably, because the noise involved Lie derivatives whose commutator admits a certain type of inequality estimate, the analytical properties of the 3D Euler fluid equations with Stochastic Advection by Lie Transport (SALT) turned out to be essentially the \emph{same} as for the deterministic Euler fluid case in three dimensions. This agreement in analytic regularity was not entirely a surprise, though, because Cotter, Gottwald and Holm \cite{cotter2017stochastic} had earlier derived the SALT model for the Euler fluid equations by applying the mathematical method of \emph{homogenisation} in time to the Euler fluid equations. The contrast between transport noise which conserves circulation of Euler fluids, but not their energy is contrasted with energy conserving stochasticity which does not conserve circulation in 
Drivas and Holm \cite{DrivasHolm2019,DrivasHolmLeahy2020}. 

{\bf  Other types of transport noise: Geometric Rough Paths.} The recent work in Crisan, Holm, Leahy and Nilsen \cite{crisan2022variational} entitled ``Variational Principles for Fluid Dynamics on Rough Paths" aims to transfer the fundamental properties of deterministic fluid dynamics derived by Hamilton's principle into their formulation on geometric rough paths. Another recent work \cite{crisan2022solution} by the same authors entitled ``Solution Properties of the Incompressible Euler System with Rough Path Advection" analyses the solution properties of Euler fluid dynamics on rough paths and demonstrates the efficacy of the approach in this work to produce previously unavailable results, such as the Beale-Kato-Majda blowup criterion for Euler fluid solutions on geometric rough paths.

{\bf  Use of transport noise in data assimilation for oceanography.}
Recent work on stochastic Geophysical Fluid Dynamics (GFD) has led to collaboration to investigate  Stochastic Dynamical Data Assimilation (SDDA) for quantifying and reducing uncertainty in numerical simulations of weather, climate and ocean circulation. Remarkably, the SDDA approach for GFD also has been found to work well for quantifying and reducing uncertainty in applications of stochastic shape analysis for image registration in Arnaudon, Holm and Sommer \cite{arnaudon2021stochastic}. 
In particular, foundational development of the new science of Stochastic Geometric Mechanics (SGM) for spatial smooth invertible maps with stochastic time dependence has produced fruitful applications to uncertainty quantification and reduction of uncertainty via data assimilation in fluid dynamics. A promising goal is to apply stochastic geometric mechanics for the derivation, analysis, numerical simulation and assimilation of computational data, e.g., satellite observations of various types of transport (heat, mass, colour, texture, and other observable order parameters) in upper ocean dynamics. Reaching this goal will require a variety of developments in modelling and data assimilation. The fundamental advances in the development of data assimilation methods based on transport noise include \cite{alonso2020modelling,de2020implications,cotter2019numerically,GH19,holm2020stochastic,holm2021astochastic}. 
\index{Geometric Mechanics!Stochastic} \index{Stochastic!Geometric Mechanics} 

\section{Stochastic geometric mechanics with diffeomorphisms}
\subsection*{Key points}
\begin{center}
    $\bullet\,$ Semimartingale $\qquad$ \index{stochastic!semimartingale}
    $\bullet\,$ Stratonovich integral $\qquad$ \index{stochastic!Stratonovich integral}
    $\bullet\,$ Stochastic Lie chain rule \index{stochastic!Lie chain rule} \index{Lie chain rule!stochastic} 
\end{center}
Geometric mechanics can be made stochastic in two distinct ways while preserving structure. The first option is called stochastic advection by Lie transport (SALT) and was introduced by \cite{Holm2015}. SALT leaves the geometric structure invariant, but at the cost of losing conservation of energy. It corresponds to replacing the deterministic reconstruction equation in \eqref{eq:reconstructiondeterministic} by the semimartingale \index{SALT!history}
\begin{equation}
{\rmd}\phi_t(x_0) = u(\phi_t(x_0),t)dt + \sum_{i=1}^M \xi_i(\phi_t(x_0))\circ dW_t^i,
\label{eq:reconstructionstochastic}
\end{equation}
where here the symbol $\circ$ means that the stochastic integral is taken in the Stratonovich sense. Note that in all other instances $\circ$ means composition of functions. The initial data is given by $\phi(x_0,0)=x_0$. The $W_t^i$ are independent, identically distributed Brownian motions, defined with respect to the standard stochastic basis $(\Omega,\mathcal{F},(\mathcal{F}_t)_{t\geq 0},\mathbb{P})$, see \cite{karatzas1998brownian}. Such a noise was shown to arise from a multi-time homogenisation argument in \cite{cotter2017stochastic}. This replacement \eqref{eq:reconstructionstochastic} need not be by a semimartingale, one can introduce geometric rough paths by the same approach, which was done in \cite{crisan2022variational}. The $\xi_i(\,\cdot\,)\in\mathfrak{X}^s$ are called data vector fields and are prescribed. These data vector fields represent the effects of unresolved degrees of freedom on the resolved scales of motion and are meant to account for unrepresented processes. 

The data vector fields $\xi_i$ can be determined by applying empirical orthogonal function analysis to appropriate numerical and/or observational data. For instance, for an application to the two dimensional Euler equations for an ideal fluid, see \cite{cotter2019numerically} and \cite{ephrati2023data}. An application of this framework to a two-layer quasi-geostrophic model can be found in \cite{cotter2018modelling}. Stochastic models enable the use of a variety of methods in data assimilation, which are discussed in \cite{cotter2019particle}. It is not difficult to make sense of \eqref{eq:reconstructiondeterministic}, but understanding \eqref{eq:reconstructionstochastic} is more complicated and requires stochastic analysis. Of particular importance is a stochastic chain rule, which is shown to exist in \cite{de2020implications}. This stochastic chain rule is called the \emph{Kunita-It\^o-Wentzell (KIW) formula} and helps interpret the semimartingale in \eqref{eq:reconstructionstochastic}. The KIW formula will also be used later to prove the stochastic Kelvin--Noether circulation theorem.  \index{Kunita-It\^o-Wentzell formula!KIW formula} 
\index{empirical orthogonal function} \index{empirical orthogonal function!data vector field} 
\index{stochastic!data vector field!empirical orthogonal function}

In deriving equations in continuum mechanics, one needs to keep track of the mass form as well. As discussed in the previous section, an appropriate mathematical setting for this is an \emph{outer semidirect product group}. This means that one constructs a new group from two given groups with a particular type of group operation. For continuum mechanics, the ingredients are $\mathfrak{D}^s$ and $V^*$, where $V^*$ is a vector space of tensor fields. The reason for starting with $V^*$ rather than just $V$ is historical and will become clear when we discuss the extension of the diagram in Figure \ref{fig:GM-Cube}. The vector space $V^*$ is the space of \emph{SALT!advected quantities} and it will always contain at least the mass form $\rho\,d\mu=(\phi^{-1})^*(d\mu)$.  \smallskip
\index{semidirect product group!outer}

\subsection{Euler-Poincar\'e theorem} \index{stochastic!Euler-Poincar\'e theorem} \index{Euler-Poincar\'e theorem!stochastic}
\subsection*{Key points}
\begin{center}
    $\bullet\,$ Variational derivative $\qquad$ \index{variational derivative!stochastic}
    $\bullet\,$ Symmetry-reduced variations $\qquad$ \index{variation!symmetry-reduced}
    $\bullet\,$ Stochastic Euler-Poincar\'e theorem $\qquad$ \\ \index{Euler-Poincar\'e theorem!stochastic} 
    \index{stochastic!Euler-Poincar\'e theorem}  \index{stochastic!Kelvin-Noether theorem}
    $\bullet\,$ Stochastic Kelvin-Noether theorem $\qquad$  \index{Kelvin-Noether theorem!stochastic}
    $\bullet\,$ Stochastic Euler-Boussinesq equations \index{Euler-Boussinesq equations!stochastic}
    \index{stochastic!Euler-Boussinesq equations}
\end{center}
In the situation where noise is present, that is, when the reconstruction equation is  \eqref{eq:reconstructionstochastic}, the Euler-Poincar\'e variations become stochastic. Consider $\phi:\mathbb{R}^2\to\mathfrak{D}^s$ with $\phi_{t,\epsilon}=\phi(t,\epsilon)$ to be a two parameter subgroup with smooth dependence on $\epsilon$, but only continuous dependence on $t$. Let us denote 
\[
{\rmd}\chi_{t,\epsilon}(X) = ({\rmd}\phi_{t,\epsilon}\circ \phi_{t,\epsilon})(X) = u_{t,\epsilon}(X)dt + \sum_{i=1}^N \xi_i(X)\circ dW_t^i
\]
 and  
 \[
 v_{t,\epsilon}(X) = (\frac{\partial}{\partial \epsilon}\phi_{t,\epsilon}\circ \phi_{t,\epsilon})(X)\,.
 \] 
As mentioned before, when a $\circ$ symbol is followed by $dW_t$ it means Stratonovich integration and in every other context the $\circ$ symbol is used to denote composition. Note that the data vector fields $\xi_i$ are prescribed and hence will not have a dependence on $\epsilon$. 
 
In order to compute with these stochastically parametrised subgroups and their associated vector fields, one needs a stochastic Lie chain rule. The Kunita-It\^o-Wentzell (KIW) formula is the stochastic generalisation of the Lie chain rule \eqref{LieChainRule}. A proof of the KIW formula is given  in \cite{de2020implications} for differential $k$-forms and vector fields. The proof includes the technical details on regularity that will be omitted here. In the KIW formula, the $k$-form is allowed to be a semimartingale itself. Let $K$ be a continuous adapted semimartingale that takes values in the $k$-forms and satisfies 
\begin{equation}
K_t = K_0 + \int_0^t G_s ds + \sum_{i=1}^N\int_0^t H_{i\,s}\circ dB_s^i,
\label{eq:kformsemimartingale}
\end{equation}
where the $B_t^i$ are independent, identically distributed Brownian motions. The drift of the semimartingale $K$ is determined by $G$ and the diffusion by $H_i$, both of which are $k$-form valued continuous adapted semimartingales with suitable regularity. Let $\phi_t$ satisfy \eqref{eq:reconstructionstochastic}, then \cite{de2020implications} shows that the following holds
\begin{equation}
{\rmd}(\phi_t^*K_t) = \phi_t^*\big({\rmd}K_t + \mathcal{L}_{u_t} K_t\,dt + \mathcal{L}_{\xi_i}K_t \circ dW_t^i\big).
\label{eq:kiwformula}
\end{equation}
Equation \eqref{eq:kformsemimartingale} helps to interpret the ${\rmd}K_t$ term in the KIW formula \eqref{eq:kiwformula}. This formula will be particularly useful in computing the variations of the variables in the Lagrangian. To compute these variations, one needs the variational derivative. 

\paragraph{The variational derivative.} The variational derivative of a functional $F:\mathcal{B}\to\mathbb{R}$, where $\mathcal{B}$ is a Banach space, is denoted $\delta F/\delta \rho$ with $\rho\in\mathcal{B}$. The variational derivative can be defined by the first variation of the functional
\begin{equation}
\delta F[\rho]:= \frac{d}{d\epsilon}\Big|_{\epsilon=0} F[\rho+\epsilon \delta\rho] = \int \frac{\delta F}{\delta \rho}(x)\delta\rho(x)\,dx = \left\langle\frac{\delta F}{\delta \rho},\delta \rho\right\rangle.
\end{equation}
In the definition above, $\epsilon\in\mathbb{R}$ is a parameter, $\delta\rho\in\mathcal{B}$ is an arbitrary function and the first variation can be understood as a Fr\'echet derivative. A precise and rigorous definition can be found in \cite{gelfand2000calculus}. With the definition of the functional derivative in place, the following lemma can be formulated.
\medskip

\begin{lemma}
With the notation as above, the variations of $u$ and any advected quantity $a$ are given by 
\begin{equation}
\delta u(t) = {\rmd}v(t) + [{\rmd}\chi_t,v(t)],\quad \delta a(t) = -\mathcal{L}_{v(t)}a(t),
\label{def:delta-var}
\end{equation}
where $v(t)\in\mathfrak{X}^s$ is arbitrary.
\index{stochastic!variation of velocity} \index{variation!stochastic}
\end{lemma}
\begin{proof}
The proof of the variation of $a(t)$ is a direct application of the Kunita-It\^o-Wentzell formula to $a(t,\epsilon)=g_{t,\epsilon*}a_0$. Note that the data vector fields $\xi_i$ are prescribed and do not depend on $\epsilon$. Denote by $x_{t,\epsilon} = g_{t,\epsilon}(X)$. Then one has
\begin{equation}
{\rmd}\phi_{t,\epsilon}(X) = {\rmd}x_{t,\epsilon} = u_{t,\epsilon}(x_{t,\epsilon})\,dt + \sum_{i=1}^N \xi_i(x_{t,\epsilon})\circ dW_t^i =: {\rmd}\chi_{t,\epsilon}(x_{t,\epsilon}).
\label{eq:twoparameterstochu}
\end{equation}
The vector field associated to the $\epsilon$-dependence of the two parameter subgroup is given by
\begin{equation}
\frac{\partial}{\partial \epsilon}\phi_{t,\epsilon} = \frac{\partial}{\partial \epsilon}x_{t,\epsilon} = v_{t,\epsilon}(x_{t,\epsilon}).
\label{eq:twoparameterstochv}
\end{equation}
Computing the derivative with respect to $\epsilon$ of \eqref{eq:twoparameterstochu} gives
\begin{equation}
\begin{aligned}
\frac{\partial}{\partial \epsilon}{\rmd}x_{t,\epsilon} &= \frac{\partial}{\partial \epsilon}\big({\rmd}\chi_{t,\epsilon}(x_{t,\epsilon})\big)\\
&= \left(\frac{\partial}{\partial \epsilon}u_{t,\epsilon} + v_{t,\epsilon}\cdot\frac{\partial}{\partial x_{t,\epsilon}}{\rmd}\chi_{t,\epsilon}\right)(x_{t,\epsilon}),
\end{aligned}
\end{equation}
where the independence of the data vector fields $\xi_i$ on $\epsilon$ was used. Taking the differential with respect to time of \eqref{eq:twoparameterstochv} gives
\begin{equation}
\begin{aligned}
{\rmd}\left(\frac{\partial}{\partial \epsilon} x_{t,\epsilon}\right) &= {\rmd}\big(v_{t,\epsilon}(x_{t,\epsilon})\big)\\
&=  \left( {\rmd}v_{t,\epsilon}(x_{t,\epsilon}) + {\rmd}\chi_{t,\epsilon}\cdot\frac{\partial}{\partial x_{t,\epsilon}}v_{t,\epsilon}\right)(x_{t,\epsilon}).
\end{aligned}
\end{equation}
One can then evaluate at $\epsilon=0$ and call upon equality of cross derivative-differentials to obtain the result by subtracting. Since $g_{t,\epsilon}$ depends on $t$ in a $C^0$ manner, the integral representation is required. The particle relabelling symmetry permits one to stop writing the explicit dependence on space,
\begin{equation}
\delta u(t)\,dt = {\rmd}v(t) + [{\rmd}\chi_t,v(t)].
\end{equation}
This completes the proof of formula \eqref{def:delta-var} for the variation of $u(t)$.
\end{proof}
The notation in \eqref{eq:twoparameterstochu} needs careful explanation, because it comprises both a stochastic differential equation and a definition. The symbol ${\rmd}\chi_{t,\epsilon}$ is used to define a vector field, whereas ${\rmd}x_{t,\epsilon}$ denotes a stochastic differential equation. This lemma makes the presentation of the stochastic Euler-Poincar\'e theorem particularly simple.
\index{stochastic!Euler-Poincar\'e theorem} \index{Euler-Poincar\'e theorem!stochastic} \index{stochastic!partial differential equation}
\medskip

\begin{theorem}[Euler-Poincar\'e theorem for the diffeomorphisms]\label{thm:SEP}
With the notation as above, the following are equivalent.
\begin{enumerate}[i)]
\item The constrained variational principle
\begin{equation}
\delta\int_{t_1}^{t_2}\ell(u,a)\,dt = 0
\end{equation}
holds on $\mathfrak{X}^s\times V^*$, using variations $\delta u$ and $\delta a$ of the form
\begin{equation}\label{eq:epconstraints}
\delta u = {\rmd}v + [{\rmd}\chi_t,v], \qquad \delta a = -\mathcal{L}_v a,
\end{equation}
where $v(t)\in \mathfrak{X}^s$ is arbitrary and vanishes at the endpoints in time for arbitrary times $t_1,t_2$.
\item The stochastic Euler-Poincar\'e equations hold on $\mathfrak{X}^s\times V^*$
\begin{equation}
{\rmd}\frac{\delta \ell}{\delta u} + \mathcal{L}_{{\rmd}\chi_t}\frac{\delta \ell}{\delta u} = \frac{\delta \ell}{\delta a}\diamond a\,dt,
\label{eq:stochep}
\end{equation}
and the advection equation
\begin{equation}
{\rmd}a + \mathcal{L}_{{\rmd}\chi_t}a = 0.
\label{eq:stochadv}
\end{equation}
\end{enumerate}
\end{theorem} 

\begin{proof}
Using integration by parts and the endpoint conditions $v(t_1)=0=v(t_2)$, the variation can be computed to be
\begin{equation}
\begin{aligned}
\delta\int_{t_1}^{t_2}\ell(u,a)\,dt 
&= 
\int_{t_1}^{t_2}\left\langle\frac{\delta\ell}{\delta u},\delta u\right\rangle + \left\langle\frac{\delta\ell}{\delta a},\delta a\right\rangle\,dt\\
&= \int_{t_1}^{t_2}\left\langle\frac{\delta\ell}{\delta u},{\rmd}v + [{\rmd}\chi_t,v]\right\rangle + \left\langle\frac{\delta\ell}{\delta a}\,dt,-\mathcal{L}_v a\right\rangle\\
&= \int_{t_1}^{t_2}\left\langle -{\rmd}\frac{\delta\ell}{\delta u} - \mathcal{L}_{{\rmd}\chi_t}\frac{\delta\ell}{\delta u} + \frac{\delta\ell}{\delta a}\diamond a\,dt,v\right\rangle\\
&= 0\,.
\end{aligned}
\end{equation}
Since the vector field $v$ is arbitrary, one obtains the stochastic Euler-Poincar\'e equations. Finally, the advection equation \eqref{eq:stochadv} follows by applying the KIW formula to $a(t)=\phi_{t*}a_0$. \index{SALT!Euler--Poincar'e formulation} \index{Euler--Poincar'e equations!SALT}
\end{proof}

\begin{remark}
The stochastic Euler-Poincar\'e theorem is equivalent to the version presented in \cite{Holm2015}, which uses stochastic Clebsch constraints. In \cite{Holm2015} one can also find an investigation of the It\^o formulation of the stochastic Euler-Poincar\'e equation. See also \cite{cruzeiro2018momentum}. The general form of the Euler-Poincar\'e theorem uses the ${\rm ad}$ operator in place of the commutator in \eqref{eq:epconstraints} and the derivative of the representation of the group action in place of the Lie derivative. The resulting Euler-Poincar\'e equations are then formulated using ${\rm ad}^*$ for the group that is involved. The  general form can be found in \cite{holm1998euler}. 
\end{remark}

\subsection{Lie-Poisson formulation.}
The stochastic Euler-Poincar\'e equations have an equivalent stochastic Lie-Poisson formulation. To obtain the Lie-Poisson formulation, one must Legendre transform the reduced Lagrangian. The Legendre transformation \index{Legendre transformation} in the presence of stochasticity becomes itself stochastic in the following way
\begin{equation}
m := \frac{\delta\ell}{\delta u}, \qquad \hslash(m,a)\,dt + \sum_{i=1}^N\langle m,\xi_i\rangle \circ dW_t^i = \langle m,{\rmd}\chi_t\rangle - \ell(u,a)\,dt.
\label{eq:reducedstochlegendre}
\end{equation}
The stochasticity enters the Legendre transformation \index{Legendre transformation} because the momentum map \index{momentum map} $m$ is coupled to the stochastic vector field ${\rmd}\chi_t$. The left hand side of the transformation determines the Hamiltonian, which is a semimartingale. The underlying semidirect product group structure has not changed, it is still the $H^s$ diffeomorphisms with a vector space, but the Hamiltonian has become a semimartingale. This implies that in the stochastic case the energy is not conserved, because Hamiltonian depends on time explicitly. Note that \eqref{eq:reducedstochlegendre} emphasises that the Lagrangian does not feature stochasticity in this framework. Instead, the Lagrangian represents the physics in the problem, which does not change. The stochasticity is supposed to account for the difference between observed data and deterministic modelling. The stochastic Lie-Poisson equations are given by
\begin{equation}
{\rmd}(m,a) = -{\rm ad}^*_{(\frac{\delta\hslash}{\delta m},\frac{\delta\hslash}{\delta a})}(m,a)\,dt - \sum_{i=1}^N{\rm ad}^*_{(\xi_i,0)}(m,a)\circ dW_t^i,
\label{eq:stochliepoisson}
\end{equation}
where ${\rm ad}^*$ is dual to ${\rm ad}$ under $L^2$ pairing. Since both the drift and the diffusion part use the same operator (the ${\rm ad}^*$ operator) in \eqref{eq:stochliepoisson}, the stochastic Lie-Poisson equations preserve the same family of Casimirs (or integral conserved quantities) as the deterministic Lie-Poisson equations. 
Hence, the stochastic Euler-Poincar\'e theorem yields a stochastic Kelvin-Noether circulation theorem as a corollary. 
\index{SALT!Lie-Poisson formulation} \index{Lie-Poisson!SALT} \index{stochastic!Euler-Poincar\'e theorem}
 \index{stochastic!Kelvin-Noether theorem}

\subsection{Kelvin-Noether theorem.} 
Let $\mathfrak{C}^s$ be the space of loops $\gamma:S^1\to\mathfrak{D}^s$, which is acted upon from the left by $\mathfrak{D}^s$. Given an element $m\in\mathfrak{X}^{s*}$, one obtains a covector-valued density whose density is constant in space by dividing the momentum $m$ by the density $\rho$. By considering only the covector-valued part of the momentum, the Kelvin-Noether theorem is as follows. The circulation map $\mathcal{K}:\mathfrak{C}^s\times V^*\to\mathfrak{X}^{s**}$ is defined by 
\begin{equation}
\langle \mathcal{K}(\gamma,a),m\rangle = \oint_\gamma\frac{m}{\rho}\,.
\end{equation}
Given a Lagrangian $\ell:\mathfrak{X}^s\times V^*\to \mathbb{R}$,  the \emph{Kelvin-Noether quantity} is defined by
\begin{equation}
I(\gamma,u,a) := \oint_\gamma\frac{1}{\rho}\frac{\delta\ell}{\delta u}\,.
\end{equation}
One can now formulate the following stochastic Kelvin-Noether circulation theorem. 
\medskip

\begin{theorem}[Stochastic Kelvin-Noether theorem]\label{thm:KelThm}
Let $u_t=u(t)$ satisfy the stochastic Euler-Poincar\'e equation \eqref{eq:stochep} and $a_t=a(t)$ the stochastic advection equation \eqref{eq:stochadv}. Let $\phi_t$ be the flow associated to the vector field ${\rmd}\chi_t$. That is, ${\rmd}\chi_t = {\rmd}\phi_t\circ \phi_t^{-1} = u_t\,dt + \sum_{i=1}^N \xi_i\circ dW_t^i$. Let $\gamma_0\in \mathfrak{C}^s$ be a loop. Denote by $\gamma_t = \phi_t\circ \gamma_0$ and define the Kelvin-Noether quantity $I(t):= I(\gamma_t,u_t,a_t)$. Then 
\index{SALT!Kelvin-Noether theorem}
\begin{equation}
{\rmd}I(t) = \oint_{\gamma_t}\frac{1}{\rho}\frac{\delta\ell}{\delta a}\diamond a\,dt\,.
\label{eqn:KelThm}
\end{equation}
\end{theorem}
\begin{proof}
The statement of the stochastic Kelvin-Noether circulation theorem involves a loop that is moving with the stochastic flow. One can transform to stationary coordinates by pulling back the flow to the initial condition. This pull-back yields
\begin{equation}
I(t) = \oint_{\gamma_t}\frac{1}{\rho}\frac{\delta\ell}{\delta u} = \oint_{\gamma_0}\phi_t^*\left(\frac{1}{\rho}\frac{\delta\ell}{\delta u}\right) = \oint_{\gamma_0}\frac{1}{\rho_0}\phi_t^*\left(\frac{\delta\ell}{\delta u}\right).
\end{equation}
An application of the Kunita-It\^o-Wentzell formula \eqref{eq:kiwformula} then leads to 
\begin{equation}
{\rmd}I(t) = \oint_{\gamma_0}\frac{1}{\rho_0}\phi_t^*\left({\rmd}\frac{\delta\ell}{\delta u} + \mathcal{L}_{{\rmd}\chi_t}\frac{\delta \ell}{\delta u}\right) = \oint_{\gamma_0}\frac{1}{\rho_0}\phi_t^*\left(\frac{\delta\ell}{\delta a}\diamond a\right)\,dt,
\end{equation}
since $\delta\ell/\delta u$ satisfies the stochastic Euler-Poincar\'e equation in \eqref {eq:stochep}. Transforming back to the moving coordinates by pushing forward with $\phi_t$ yields the Kelvin circulation equation \eqref{eqn:KelThm}.
\end{proof}

Thus, Theorem \ref{thm:KelThm} explains how particle relabelling symmetry gives rise to the Kelvin-Noether circulation theorem via Noether's theorem. When the only advected quantity present is the mass density, the loop integral of the diamond terms vanishes. This means that circulation is conserved according to Noether's theorem for an incompressible fluid, or for a barotropically compressible fluid. The presence of other advected quantities breaks the symmetry further and introduces the  \emph{diamond terms} as forces which generate circulation on the Kelvin-Noether circulation theorem in equation \eqref{eqn:KelThm}. Consequently, symmetry breaking due to additional order parameters can provide additional mechanisms for the generation of Kelvin-Noether circulation in ideal fluid dynamics.
\index{symmetry breaking!Kelvin-Noether circulation}

\subsection{Euler-Boussinesq equations.}
Let us now use the theory introduced so far to derive the Euler-Boussinesq equations in a domain $\Omega$, with a given bottom topography $\mathscr{B}(x,y)$, a free surface $\zeta(x,y,t)$ and non-penetration boundary conditions on the lateral walls.
where $\mathbf{u}_3=(\mathbf{u},w)$ is the three dimensional velocity field, $\mathbf{u}$ is the horizontal velocity field and $w$ is the vertical velocity. The $\boldsymbol\xi_{3i}(x,y,z)$ are data vector fields and they represent spatial velocity-velocity correlations. The $W_t^i$ are independent, identically distributed Brownian motions for each $i=1,\hdots,M$. In what follows, we will employ Einstein's convention of summing over repeated indices to shorten the notation. Before deriving the equations, let us set up the appropriate boundary conditions. When a surface in a moving fluid consists of the same particles for all time, then it is a bounding surface of the fluid. The converse is also true, every bounding surface is a material surface. Let
\begin{equation}
F(x,y,z,t) = 0
\end{equation}
be the equation for the material surface. For $F$ to be a material surface, it must be a Lagrangian invariant, i.e., the stochastic material derivative must vanish. This means that $F$ is required to satisfy
\begin{equation}
\frac{1}{|\nabla_3 F|}\big({\rmd}F + ({\rmd}\boldsymbol\chi_{3t}\cdot\nabla_3)F\big) = 0.
\end{equation}
On the free surface boundary, one has $F(x,y,z,t)=z-\zeta(x,y,t)$. Hence, the free surface boundary condition is given by
\begin{equation}\label{eq:freesurface}
w\,dt + \hat{\mathbf{z}}\cdot\boldsymbol \xi_{3i}\circ dW_t^i = {\rmd}\zeta + ({\rmd}\boldsymbol \chi_t\cdot\nabla)\zeta \quad \text{ at } z=\zeta(x,y,t).
\end{equation} 
The notation in \eqref{eq:freesurface} uses horizontal and vertical components of ${\rmd}\boldsymbol\chi_{3t}=({\rmd}\boldsymbol \chi_t,w\,dt + \hat{\mathbf{z}}\cdot\boldsymbol \xi_{3i}\circ dW_t^i)$. For the bottom topography, i.e., the bottom boundary condition, we set $F(x,y,z,t) = z + \mathscr{B}(x,y)$, which yields
\begin{equation}\label{eq:bathymetry}
w\,dt + \hat{\mathbf{z}}\cdot\boldsymbol \xi_{3i}\circ dW_t^i = -({\rmd}\boldsymbol \chi_t\cdot\nabla)h \quad \text{ at } z=-h(x,y).
\end{equation}
The non-penetration boundary condition on the lateral boundaries is given by
\begin{equation}
{\rmd}\boldsymbol \chi_t\cdot\hat{\mathbf{n}} = 0, \quad \text{ on any vertical lateral boundary},
\end{equation}
where $\hat{\mathbf{n}}$ is the outward pointing unit vector normal to the lateral boundaries. This condition can be obtained from the incompressiblity condition by an application of the divergence theorem to
\begin{equation}
\nabla_3\cdot{\rmd}\boldsymbol \chi_{3t} = 0, i.e., \quad \nabla_3\cdot\mathbf{u}_3 = 0 \text{ and } \nabla_3\cdot\boldsymbol \xi_{3i}=0  \hbox{ for all}\, i.
\end{equation}
For incompressible flows, the pressure is determined by the divergence-free condition, rather than by an equation of state. Since the fluid velocity field is semimartingale, the pressure must also be a semimartingale. We use the notation
\index{semimartingale!velocity} \index{semimartingale!pressure}
\begin{equation}
p\circ d\mathbf{S}_t = p_d\,dt + p_i\circ dW_t^i,
\end{equation}
where the notation $\circ d\mathbf{S}_t=(dt,\circ dW_t^1,\hdots,\circ dW_t^M)$ is borrowed from \cite{street2021semi} for semimartingale driven variational principles. Also, here $p_d$ corresponds to the drift part and $p_i$ for $i=1,\hdots,M$ comprise the pressures associated with the diffusion part. We do not include surface tension in this model. The dynamic boundary condition for the pressure is then given by
\begin{equation}
p\circ d\mathbf{S}_t = 0 \text{ at } z=\zeta(x,y,t) \quad \text{ or } \quad p\circ d\mathbf{S}_t = {\rmd}\zeta(x,y,t) \text{ at } z=0.
\end{equation}
The final boundary condition is on the buoyancy $\mathscr{B}$, which is given by
\begin{equation}
\hat{\mathbf{n}}_3\times \nabla_3 \mathscr{B} = 0 \text{ on } \partial\Omega
\,.\end{equation}
This boundary condition implies that the boundary $\partial\Omega$ of the domain $\Omega$ is a level set of buoyancy. We can now derive the equations for an Euler-Boussinesq fluid flow in the domain $\Omega$. The starting point is the Lagrangian 
\begin{equation}
\ell_{EB} = \int_{\Omega}\left(\frac{1}{2}|\mathbf{u}|^2 + \frac{\sigma^2}{2}w^2 + \frac{1}{{\rm Ro}}\mathbf{u}\cdot\mathbf{R} - \frac{1}{{\rm Fr}^2}(1+\mathfrak{s}b)z\right)D\,dx\,dy\,dz.
\end{equation}
Here $D$ denotes the dimensionless density of the fluid. $\mathbf{R}$ is the vector potential for the Coriolis parameter $f(x,y)$, that is, $\nabla_3\times\mathbf{R} = f(x,y)\hat{\mathbf{z}}$. Furthermore, $\sigma$ denotes the aspect ratio, ${\rm Ro}$ is the Rossby number, ${\rm Fr}$ is the Froude number and $\mathfrak{s}$ is the stratification parameter. To obtain the equations of motion, we need a constrained variational principle due to the incompressibility. This means that we formulate a constrained action principle where we set the density equal to unity by a Lagrange multiplier. This Lagrange multiplier can be identified as the pressure. The dimensionless action for the Euler-Boussinesq model is given by
\begin{equation}
S_{EB} = \int_{t_1}^{t_2}\ell_{EB}\,dt - \left\langle\frac{1}{{\rm Fr}^2} p, D-1\right\rangle\circ d\mathbf{S}_t =: \int_{t_1}^{t_2} c\ell_{EB}\circ d\mathbf{S}_t,
\end{equation}
We can now compute the variational derivatives of the constrained Lagrangian $c\ell_{EB}$, which are given by 
\begin{equation}
\begin{aligned}
\frac{\delta c\ell_{EB}}{\delta \mathbf{u}} &= D\left(\mathbf{u}+\frac{1}{{\rm Ro}}\mathbf{R}\right)\\
\frac{\delta c\ell_{EB}}{\delta w} &= \sigma^2 Dw\\
\frac{\delta c\ell_{EB}}{\delta b} &= -\frac{1}{{\rm Fr}^2}Dz,\\
\frac{\delta c\ell_{EB}}{\delta D} &= \left(\frac{1}{2}|\mathbf{u}|^2 + \frac{\sigma^2}{2}w^2 + \frac{1}{{\rm Ro}}\mathbf{u}\cdot\mathbf{R} - \frac{1}{{\rm Fr}^2}(1+\mathfrak{s}b)z\right) - \frac{1}{{\rm Fr}^2} p \circ d\mathbf{S}_t,\\
\frac{\delta c\ell_{EB}}{\delta p} &= \frac{1}{{\rm Fr}^2}(D-1)\circ d\mathbf{S}_t.
\end{aligned}
\end{equation}
Inserting the variational derivatives into the stochastic Euler-Poincar\'e theorem \ref{thm:SEP} yields the stochastic Euler-Boussinesq equations
\index{SALT!stochastic Euler-Boussinesq equations}
\begin{equation}
\begin{aligned}
{\rmd}\mathbf{u} + ({\rmd}\boldsymbol\chi_{3t}\cdot\nabla_3)\mathbf{u} + (\nabla\boldsymbol\xi_{3i})\cdot\mathbf{u}_3\circ dW_t^i &= -\frac{1}{{\rm Fr}^2}\nabla (p_d \,dt + p_i\circ dW_t^i) - \frac{1}{{\rm Ro}}f\hat{\mathbf{z}}\times{\rmd}\boldsymbol \chi_t - \frac{1}{{\rm Ro}}\nabla(\boldsymbol \xi_i\cdot\mathbf{R})\circ dW_t^i,\\
\sigma^2\left({\rmd}w +({\rmd}\boldsymbol\chi_{3t}\cdot\nabla_3)w + \Big(\frac{\partial}{\partial z}\boldsymbol\xi_{3i}\Big)\cdot\mathbf{u}_3\circ dW_t^i \right) &= -\frac{1}{{\rm Fr}^2}\frac{\partial}{\partial z}(p_d \,dt + p_i\circ dW_t^i) + \frac{1}{{\rm Fr}^2}(1+\mathfrak{s}b)dt,\\
{\rmd}b + ({\rmd}\boldsymbol \chi_{3t}\cdot\nabla_3)b &= 0,\\
\nabla_3\cdot{\rmd}\boldsymbol\chi_{3t} &= 0.
\end{aligned}
\end{equation}
The Kelvin circulation theorem is given by
\begin{equation}
\begin{aligned}
{\rmd}\oint_{c({\rmd}\boldsymbol \chi_{3t})}\left(\mathbf{u}_3 + \frac{1}{{\rm Ro}}(\mathbf{R},0)\right)\cdot d\mathbf{x}_3 &= -\frac{\mathfrak{s}}{{\rm Fr}^2}\oint_{c({\rmd}\boldsymbol \chi_{3t})} b\,dz\,dt\\
&= -\frac{\mathfrak{s}}{{\rm Fr}^2}\int\!\!\int_{\partial S = c({\rmd}\boldsymbol\chi_{3t})}\hat{\mathbf{z}}\times\nabla_3 b\cdot d\mathbf{S}_3\,dt.
\end{aligned}
\end{equation}
The Kelvin circulation theorem shows that for the Euler-Boussinesq model circulation is being generated whenever the buoyancy gradient does not align with the vertical unit vector. The stochastic Euler--Poincar\'e equation for the Euler-Boussinesq model also preserves the geometric structure of the determinant equations, including the Lagrangian invariant known as the potential vorticity and an infinite family of conserved integral quantities known as Casimirs, for which we refer to \cite{holm2021stochastic, luesink2021stochastic}. \index{Casimirs}


\newpage
\part{Lagrangian Averaged Stochastic Lie Transport}

\section{Introduction to LA-SALT}
The two stages of our stochastic approach are called stochastic advection by Lie transport
(SALT) \cite{Holm2015} and Lagrangian averaged SALT, written as LA-SALT, \cite{DrivasHolm2019,DrivasHolmLeahy2020}. These two stages represent
two different viewpoints or modelling philosophies depending on the time scales of the intended application. 
For SALT, atmospheric `weather' produces uncertainty in advection arising
from motion on unresolved time scales. In LA-SALT, atmospheric `climate' is taken as the
expectation, and the atmospheric `weather' is treated as a field of pathwise fluctuations, as discussed in Ed Lorenz's famous lecture \cite{Lorenz1995}. 
\index{LA-SALT!history} \index{LA-SALT!Hasselmann's paradigm}

The LA-SALT approach also brings us back to Hasselmann's paradigm, which decomposes a general climate model into deterministic 
and stochastic parts \cite{hasselmann1976stochastic}. Namely, the LA-SALT approach
results in deterministic linear fluctuation equations that govern the nonlinear dynamics of the climate
statistics themselves, including variance, covariance and higher statistical moments. In 
the LA-SALT framework, these higher order statistical moments are governed by linear equations. This
convenient result offers potential computational advantages and opens new perspectives for the theoretical
analysis of these moments.

The prediction of climate change and its impact on extreme weather events is one of the great societal and intellectual challenges of our time. 
The climate-change problem requires a model that:  (i) distinguishes between weather and climate; (ii) governs the fluctuation dynamics 
of the physical variables; and (iii) predicts how the variances of the fluctuations are affected by statistical correlations in the full fluctuating nonlinear dynamics. 

This part shows that the LA-SALT framework can meet these three challenges of climate dynamics in the example of   
the 2D Euler--Boussinesq (EB) equations for an incompressible stratified fluid flowing under gravity in a vertical plane with no other external forcing. 
All three parts of the problem are solved for this case. In fact, for this problem, the LA-SALT framework also delivers global 
well-posedness of the dynamics 
of the physical variables and closed dynamical equations for the moments of their fluctuations. 
Thus, in a well-posed mathematical setting, the framework developed in this paper shows that the mean field dynamics combines 
with an intricate array of correlations in the fluctuation dynamics to drive the evolution of the mean statistics. 
The results of the framework exemplified here for the 2D EB model define its climate, as well as supporting 
the concepts of climate change, weather dynamics, and change of weather statistics, 
all in the context of a model system of SPDEs with unique global strong solutions. 
\index{climate variances!prediction} 
\index{climate!definition!expectation} \index{weather v climate!expectation v pathwise fluctuations}

\subsection{Background} To meet the challenge of climate change prediction in practice, one must predict the coarse-grained dynamic changes of an extremely complex atmosphere/ocean system which is only partially observed by using a suite of imperfect theoretical and computational simulation models. This means that predictions of quantities of climate interest may be strongly affected by uncertainty arising from unknown model errors and incomplete knowledge of state variables. In addition, one must assess the impacts of climate change over a wide range of significant temporal and spatial scales. For example, one must predict and understand the seasonal, yearly, decadal, and centennial impacts of climate change for issues ranging from extreme weather events, to sea level rise, and the dynamic distributions of deserts and forests.

\paragraph{\bf Previous approaches.} Deterministic physics characterises the climate change problem as a high-dimensional complex dynamical system with sensitivity to initial conditions on essentially all spatial and temporal scales. To estimate the level of difficulty of the climate change problem, one notes that the turbulence problem falls into the same class of problems. There is an additional difficulty in climate science, though. The governing Navier-Stokes equations for turbulence are known; but  the dynamical equations for the actual climate are unknown. In fact, even the definition of climate is still under discussion in the literature \cite{Bothe2018}.  

Climate science would face an extensive closure problem, if it were assumed that the weather and the climate obey the same equations. 
Following the Hasselmann paradigm \cite{hasselmann1976stochastic} and echoing the celebrated distinction between weather and climate in \cite{Lorenz1995}, the following proposition then arises: ``Suppose the climate were defined as simply `what you expect' as a statistical property of a stochastic dynamical system." 

This proposition parallels recent computational approaches in climate/weather numerical simulations. The modern computational approach typically involves the introduction of \textit{stochastic parameterisation}, in which mean quantities of interest do have a precise sense of `expectation' and the remainder at a given instant has a sense of `fluctuation'. For recent reviews of this approach, see, e.g. \cite{BYP2012, Berner-etal-2017,GCF2015}.   In the approach to stochastic parameterisation of weather prediction, the summary conclusion of \cite{BYP2012} is that stochasticity must be incorporated at a very basic level within the design of physical process parameterisations and improvements to the dynamical core.

The SALT (stochastic advection by Lie transport) approach introduced in \cite{Holm2015} discussed in Part I combines stochasticity at the `basic level' of Kelvin's circulation theorem. A protocol for applying the SALT approach in data assimilation based on comparing fine scale and coarse scale computational simulations has recently been developed in \cite{CCHOS18a,CCHOS18b}. This part will concentrate on developing a Lagrangian-averaged (LA) version of SALT which was recently proposed in \cite{DrivasHolm2019} and developed further in \cite{DrivasHolmLeahy2020} for potential use in climate change prediction.

\subsection{Aims}
This section aims to describe results of a stochastic version of the two-dimensional Euler-Boussinesq fluid system which is non-local in \emph{probability space}, rather than in physical space, in the sense that the expected velocity is assumed to replace the drift velocity in the transport operator for the stochastic fluid flow. The LA-SALT stochastic fluid model is derived by exploiting a novel idea introduced in \cite{DrivasHolm2019}, of applying Lagrangian-averaging (LA) in \emph{probability space} to the fluid equations governed by stochastic advection by Lie transport (SALT) which were introduced in \cite{Holm2015}.  \index{LA-SALT!climate modelling}

The LA-SALT approach yields three results of interest in climate modelling based on the Kelvin circulation theorem for stochastic transport of the Kelvin loop. 
\begin{itemize}
\item
First, it answers Lorenz's question in \cite{Lorenz1995} about determinism of the climate in the affirmative. Namely, by replacing the drift velocity of the stochastic vector field by its expected value, one finds that the expected fluid motion becomes deterministic. This first result implies the second one.
\item
As a second result, the LA-SALT fluctuation dynamics reduces to a \emph{linear} stochastic transport problem with a deterministic drift velocity.  Such problems tend to be well-posed. 
\item
The third result involves the dynamics of the variances of the fluctuations. Namely, the variances and higher moments of the fluctuation statistics are found to evolve deterministically, driven by a certain set of correlations of the fluctuations among themselves. 
Having identified these driving correlations could be a fundamental asset in identifying conditions probabilistically ripe for extreme events.
\end{itemize}
In summary, the first result makes the distinction between climate and weather for the case at hand. Namely, the LA-SALT fluid equations may be regarded as a dissipative system akin to the Navier--Stokes equations for the expected motion (climate) which is embedded into a larger conservative system which includes the statistics of the fluctuation dynamics (weather). The second result provides a set of linear stochastic transport equations for predicting the fluctuations (weather) of the physical variables, as they are driven by the deterministic expected motion. The third result produces closed deterministic evolutionary equations for the evolution of the variances and covariances of the stochastic fluctuations, as well as their $p$-th order central moments in certain cases.

The next section applies the LA-SALT approach to the 2D Euler--Boussinesq (EB) system on a vertical plane and reveals the dynamics of its statistical properties. Specifically, the analysis in the next section defines climate, as well as climate change, weather, and change of weather statistics for 2D EB LA-SALT, all in the context of a model system of SPDEs with unique global strong solutions \cite{alonso2020modelling}.

\section{The Euler-Boussinesq (EB) fluid system in a vertical plane}\label{sec:EBsystem}

In concert with the idea that the climate should be computed with the same fundamental equations as the weather, this section addresses a representative model of stratified incompressible flow which may be a component of any climate model. Namely, it addresses the familiar 2D Euler-Boussinesq (EB) fluid system in a vertical plane. The issue of global existence of regular solutions of the deterministic Boussinesq model still remains an outstanding open problem. Its SALT version inherits most of the properties of its deterministic counterpart and its local well-posedness has been recently established in \cite{DieAytJNLS}.  

We first recall the introduction into 2D EB of Stochastic Advection by Lie Transport (SALT) as discussed in that work.  
\index{SALT!2D EB equations}  \index{LA-SALT!2D EB equations}

We then apply the Lagrangian averaging (LA) concept in probability space to derive and analyse the LA-SALT version of the 2D EB equations.  We establish global well-posedness of the LA-SALT EB system and investigate the solution behaviour of this stochastic PDE system, following the work of \cite{DieAytJNLS}. 

Section \ref{sec:EBsystem} introduces the 2D EB LA-SALT system and computes the dynamics of the expectation and fluctuation components of its solutions, as well as their variances. 

Section \ref{sec:LASDPsystems} computes expectation and fluctuation dynamics for LA-SALT equations, as well as their variances, covariances and $p$-th central moments, in a general setting. In general, the dynamics of these statistics for LA-SALT does not close. However, the fluctuation statistics for the 2D EB LA-SALT system in fact does close. The properties resulting from this closure are discussed in Example \ref{2DEB-example} of Section \ref{sec:LASDPsystems}. 

We begin with the following question. What is Kelvin's circulation theorem for the 2D EB climate/weather system?

\subsection{Kelvin circulation theorem for the 2D EB climate/weather system}

The Kelvin circulation theorem is a statement of Newton's Force Law for the motion of distributions of mass on closed material loops $c(\mb{u}_t^L)$, where the subscript $t$ denotes explicit time dependence. By definition, such material loops move with the transport velocity $\mb{u}_t^L$ of the fluid flow. Newton's Force Law states that the time rate of change of the momentum $\mb{P}$ of such a loop of a given mass distribution is equal to the force $\mb{F}$ applied to it. For the fluid situation, this is written as
\begin{align}
\frac{\diff \mb{P}}{\diff t} := \frac{\diff}{\diff t}\oint_{c(\mb{u}_t^L)} \mb{u}_t(\mb x)\cdot {\rmd} \mb x
= \oint_{c(\mbs{u}_t^L)} \mb{f}(\mb x)\cdot {\rmd} \mb x =: \mb{F}
\,.
\label{Kel-forcelaw}
\end{align}
The Kelvin-Newton relation in \eqref{Kel-forcelaw} for loop momentum dynamics apparently involves two kinds of velocity. The first velocity is $\mb{u}_t^L$, which is the velocity of the material masses distributed in the line elements along the moving loop. Since it refers to the fluid parcel transport, the velocity $\mb{u}_t^L$ is a Lagrangian quantity. A second quantity with dimensions of velocity $(\mb{u}_t)$ appears in the integrand of the Kelvin circulation. This quantity is physically the momentum per unit mass, defined in the fixed inertial frame which is required for Newton's force law \eqref{Kel-forcelaw} to be valid. This means that $\mb{u}_t$ is an Eulerian quantity, defined in the fixed frame through which the Lagrangian parcels move at velocity $\mb{u}_t^L$. Mathematically, the momentum per unit mass $(\mb{u}_t)$ is the product of the  inverse of the mass density (which itself is a subset of the advected quantities, $D\subset a$) times the variational derivative at fixed spatial coordinate of the Lagrangian $\ell(\mb{u}_t^L,a)$ in Hamilton's principle with respect to the velocity, $\mb{u}_t^L$. In Euler--Poincar\'e form, this is the Kelvin--Noether theorem of \cite{HMR1998}. Namely, \index{Kelvin--Noether!theorem} 
\begin{align}
\frac{\diff \mb{P}}{\diff t} 
:= \frac{\diff}{\diff t}\oint_{c(\mb{u}_t^L)}
\frac{1}{D}\frac{\delta \ell(\mbs{u}_t^L,a)}{\delta \mb{u}_t^L}\cdot {\rmd} \mb x
= \oint_{c(\mbs{u}_t^L)} \frac{1}{D} \frac{\delta \ell}{\delta a}\diamond a \cdot {\rmd} \mb x =: \mb{F}
\,,
\label{KelNoether-forcelaw}
\end{align}
where the diamond operation $(\diamond)$ is defined in \cite{HMR1998} and is discussed further in the present context below. 

Note, in the discussion below, when the Lagrangian velocity happens to be equal to the momentum per unit mass, then $\mb{u}_t^L\to \mb{u}_t$ and we may drop the superscript $L$, although the distinction in their definitions still remains. This slight abuse of notation should cause no confusion, because the transport velocity is a vector field which acts on the momentum per unit mass which, in turn, is the 1-form appearing in the integrand of the Kelvin circulation integral. 
\bigskip

The modelling approach of Stochastic Advection by Lie Transport (SALT) modifies the Kelvin theorem in \eqref{Kel-forcelaw} for deterministic fluids by replacing the transport velocity of the loop $\mb{u}_t^L$ in the deterministic Kelvin theorem by a Stratonovich stochastic vector field ${\rm d} x_t$ whose drift velocity is the same as the Eulerian velocity in the \emph{integrand} of the deterministic Kelvin theorem  \cite{Holm2015},
\begin{align}\label{SALT-Kel}
\oint_{c(\mb u^L_t)} \mb u_t\cdot {\rmd} \mb x
\quad\to\quad 
\oint_{c(\diff \chi_t)}
\mb u_t\cdot {\rmd} \mb x\,,
\end{align}
where $\diff \chi_t$ denotes the following stochastic process,
\begin{align}\label{dx-form}
\diff {\chi_t} := \u^L_t (x_t)\diff t+ \displaystyle\sum_k \mb{\xi}_k (x_t)\circ \diff W_t
\,.
\end{align}

The vector fields $\mb{\xi}_k$ are to be determined from data analysis as in  \cite{CCHOS18b,CCHOS18a}. In this section we will work formally, by simply assuming that these vector fields are already known from appropriate data analysis for a given application.

\begin{remark}[Notation temporal ($\diff$\,) vs spatial (${\rmd}$)]
In the literature, the letter $d$ is typically used to denote either (1) stochastic time evolution, or (2) exterior derivative/spatial differential.
To avoid confusion, here we will use the roman font $``\diff"$ to denote the former and the sans serif $``\sf{d}"$ to denote the latter.
\end{remark}

The same stochastic transport velocity ${\diff } x_t$ advects the Lagrangian parcels, which may carry advected quantities $(a)$, such as heat, mass and magnetic field lines, by Lie transport along with the flow, as
 \cite{HMR1998}
 \begin{align}\label{advec-qty}
\diff  {a} + \mathcal{L}_{\diff  {\chi_t} }a = 0\,.
\end{align}

In this section, we apply the LA-SALT (Lagrangian-averaged SALT) approach proposed in \cite{DrivasHolm2019} and developed in \cite{DrivasHolmLeahy2020}. The LA-SALT approach modifies the SALT Kelvin circulation in \eqref{SALT-Kel} by replacing the drift velocity in the stochastic transport loop velocity in \eqref{dx-form} by its expectation, plus the same noise as in SALT. Namely, cf. equation \eqref{dx-form},
\begin{align}\label{KelThm-form}
\oint_{c({\diff} \chi_t)} \mb{u}_t  \cdot {\rmd} \mb x
\quad\to \quad
\oint_{c({\diff} X_t )} \mb{u}_t \cdot {\rmd} \mb x\,,
\end{align}
where 
\begin{align}\label{dX-form}
 \diff  X_t := \E{\u^L_t}(X_t) \diff t+ \displaystyle\sum_k {\mb \xi_k (X_t)} \circ \diff W_t
\,.
\end{align}
Since the expectation in \eqref{dX-form} refers to the transport velocity $u^L_t$ of Lagrangian loop in Kelvin's theorem, we refer to this process as probabilistic Lagrangian Average (denoted as LA), reminiscent of the time average at fixed Lagrangian coordinate in the LANS-alpha turbulence model,\cite{Chen-etal1998,Chen-etal1999,Chen-etal1999+,Foiasetal01,Foiasetal02}.
For example, in the Euler fluid case the modified Kelvin theorem reads,
\begin{align}\label{KelThm-Eul}
\diff  \oint_{c\big({\diff} X_t \big)} \mb{u}_t  \cdot {\rmd} \mb x
 =  
\oint_{c\big({\diff} X_t \big)}
\big[ \diff{\mb{u}_t} \cdot {\rmd} \mb x+  \mathcal{L}_{\diff X_t} u_t \big]
= 0 \,,
\end{align}
where $ \mathcal{L}_{ {\diff}X_t}u_t $ denotes the Lie derivative of the one-form 
$u_t = \mb{u}_t \cdot \sf d \mb{x}$ with respect to the vector field $ \diff X_t$ given in equation \eqref{dX-form}. The LA-SALT motion equation leading to the modified Kelvin theorem in \eqref{KelThm-Eul} was first stated along with additional noisy and viscous terms in Lemma 3 of \cite{DrivasHolm2019}. 



The work in \cite{DieAytJNLS} extended the work in \cite{DrivasHolm2019} by following the LA-SALT (Lagrangian Averaged SALT) approach along the same lines as \cite{DrivasHolmLeahy2020} in applying expectations of the variations with respect to advected variables in combination with the known semidirect-product structure of the Lie--Poisson Hamiltonian formulation of ideal fluid dynamics. The semidirect-product structure of ideal fluid dynamics is reviewed for example in \cite{MR2013, HSS2009}. 

To express the LA-SALT equations discussed in \cite{DrivasHolmLeahy2020}, one may act with the semidirect-product (SDP) Lie--Poisson Hamiltonian matrix operator on the expected values of the variational derivatives of the Hamiltonian. In the absence of advected fluid quantities, the corresponding expected-quantity equations produce a Lie-Laplacian version of the Navier-Stokes equation, which reduces to the Navier--Stokes equation in a special choice of the functions $ \xi^{(k)}=\{(1,0,0)^T,(0,1,0)^T,(0,0,1)^T\}$ for $k=1,2,3$, as discussed in \cite{Hoch2018}. After writing the expectation equations with advected quantities in the SDP Hamiltonian matrix form, one observes that the fluctuation equations comprise a linear transport system which is slaved to the expectation equations whose solutions are deterministic and can be obtained for all time for a certain class of Hamiltonians. This slaving relation enables one to calculate the evolution equations for the local and spatially integrated variances of the fluctuations. We will describe this process here for the LA-SALT modification of the two-dimensional Euler--Boussinesq equations for a stratified incompressible fluid in a vertical plane.


\bigskip

\subsection{The LA-SALT 2D Euler--Boussinesq equations in a vertical plane}
\subsection*{The deterministic case}
The deterministic Euler--Boussinesq (EB) equations for an incompressible, inviscid 2D fluid flow in a vertical plane under gravity are given by
\begin{equation}\label{Inviscid_Boussinesq}
\left\{
\begin{array}{rl}
\partial_t\u+\left(\u\cdot\nabla\right)\u &=-\nabla p+ g\theta \hat{\y}, \qquad (\x,t)\in \mathbb{T}^2\times\RR^{+},\\
\partial_t\theta +\u\cdot\nabla \theta &= 0, \\
\nabla\cdot\u&=0,
\end{array}
\right.
\end{equation}
where $\u=(u_{1},u_{2})$ is the incompressible vector velocity field, $p$ is the scalar pressure, $g$ is the acceleration due to gravity, $\theta$  corresponds to the temperature, or buoyancy, which is transported by the fluid, and $\hat{\y}$ is the unit vector in the vertical direction. 

The 2D EB equations \eqref{Inviscid_Boussinesq} are regarded as a fundamental model of large scale atmospheric and oceanic flows, \cite{Ped87,Ric07}.  
From a mathematical point of view, the 2D EB equations retain some key features of the 2D Euler equations, including a vortex stretching mechanism for $\nabla\theta\times\hat{\y}\ne0$. The problem has attracted considerable attention in the PDE community, and local existence results and regularity criteria, as well as numerical experiments, are available, \cite{CanBen,HouLi1,Chae,ElgJeo}. The fundamental issue of whether classical solutions of the 2D incompressible Boussinesq equations can develop finite time singularities remains an outstanding open problem which still seems to be out of reach. For reviews of the history and recent analytical results for this class of problems, see, e.g., \cite{wang5645122inviscid}.

\subsection*{Deterministic computational simulations of 2D EB Rayleigh-Taylor instability} The inviscid non-diffusive 2D EB equations \eqref{Inviscid_Boussinesq} may be applied in most oceanic flows. This is because diffusive processes on oceanic flows have little effect because of their immense inertia. In particular, non-diffusive 2D EB equations have been used to model the process of oceanic deep water formation driven by Rayleigh-Taylor instability, which produces rapidly down-welling water columns called \emph{chimneys}. This process occurs in regions where warm ocean currents run into cold Arctic seas,  such as the region where the Gulf Stream runs into the Greenland Sea. 
\smallskip

Animations of computational simulations from \cite{HolmPan2022} of Rayleigh-Taylor instability in initial value problems for the deterministic EB equations appear online at the following websites:\\
Example 1: \url{https://youtu.be/zorhwJ0pmUI}\\
Example 2: \url{https://youtu.be/pXU5mJqQjuA}\\
Example 3: \url{https://youtu.be/FFdxxyyRVk8}

The stark contrasts among the simulated dynamical results of the three different simulations of Rayleigh-Taylor instability for 2D EB show a strong sensitivity to initial conditions and no particular tendency toward equipartition of energy.

\subsection*{The LA-SALT model of EB in a vertical plane} 

The LA-SALT model of EB in a vertical plane may be written in horizontal $x$ and vertical $y$ Cartesian coordinates in the Euclidean velocity representation 

\begin{equation}\label{LA:SALT:Lie:Bou}
\left\{
\begin{array}{rl}
{\diff\,}{u}\ +& \mathcal L_{\mathbb E[u]} {u} \diff t + \displaystyle \displaystyle\sum_k \mathcal L_{\xi_k} {u} \circ \diff W_t^k 
\\=& - \,{\sf{d}} \E{ p -  |\u|^2/2 }\, \diff t + g \E{\theta} \mb{\hat{y}} \diff t
\,- gy {\sf{d}} (\theta - \E{\theta}) \diff t
\,,\\ \\
\diff \theta\ +& \mathcal L_{\mathbb E[u]} {\theta} \diff t + \displaystyle \displaystyle\sum_k \mathcal L_{\xi_k} {\theta} \circ \diff W_t^k = 0
\,,\hfill \nabla \cdot \mathbb{E} \left[ \u \right]=0.
\end{array}
\right.
\end{equation}
\begin{remark}[Divergence-free condition on the expectation of the velocity] We note that although the more restrictive divergence-free condition $\nabla \cdot \u = 0$ might seem more natural to consider at first sight than our current condition $\nabla \cdot \mathbb{E} \left[ \u \right]=0$, it would make equations \eqref{LA:SALT:Lie:Bou} ill-posed. This is due to the presence of the term $\nabla \E{ p -  |\u|^2/2 },$ which imposes the pressure to be deterministic. Further insight into this will be provided once we present our approach for solving equations \eqref{LA:SALT:Lie:Bou}. Here, we simply note that if the expectation in the term $\nabla \E{ p -  |\u|^2/2 }$ is removed, the condition $\nabla \cdot \u = 0$ could be considered.

\end{remark}
In the equations above, we have employed the notation $\mathcal{L}_{\xi_k}$ to indicate Lie derivative along a vector field $\xi_k$. As stressed in Subsection \ref{2-1}, the Lie derivative on one-forms 
\[
\mathcal{L}_{\xi} u = (\mb{\xi} \cdot \nabla) \u + \displaystyle\sum_{j} \u^j \nabla {\mb \xi}^j
\]
is different from the Lie derivative applied to scalar fields $\mathcal{L}_{\xi} \theta = \mb{\xi} \cdot \nabla \theta.$ 

As explained below in Example \ref{EB-Ham-example2.1} the system \eqref{LA:SALT:Lie:Bou} 
can be rewritten in Hamiltonian operator form as 
\begin{align}
\diff 
\begin{bmatrix}
\mu \\ \\  \theta \\ \\ \rho 
\end{bmatrix}
= -
\begin{bmatrix}
 \mathcal{L}_{\Box} \mu &  - \Box  (\nabla \theta)    &  \rho \nabla {\Box}
\\ \\
\Box \cdot (\nabla \theta) & 0 & 0 \\ \\
 \nabla \cdot  (\rho \Box) & 0 & 0
\end{bmatrix}
\begin{bmatrix}
 \E{u}\diff t + \sum_k {\mb \xi_k}\circ \diff W^{(k)}_t
\\ \\ -g y \E{\rho} \diff t \\ \\ \E{ p  - \frac{ |\mb{u}|^2}{2} }\diff t  - g\E{\theta} y \diff t
\end{bmatrix},
\label{Ham-matrix-Boussinesq2-1}
\end{align}
which yields equations \eqref{LA:SALT:Lie:Bou}. 
Upon passing to the It\^o formulation, the LA-SALT EB system \eqref{LA:SALT:Lie:Bou} transforms into 
\begin{equation}\label{LA:SALT:Ito:Bou}
\left\{
\begin{array}{rl}
\diff {u} + \mathcal L_{\mathbb E[u]} {u} \diff t 
+  \displaystyle\sum_k \mathcal L_{\xi_k} {u} \, \diff W_t^k 
&= - {\sf{d}} \E{p -  |\mb{u}|^2/2} \, \diff t + g \E{\theta} \mb{\hat{y}} \diff t 
\\ & \quad
- gy {\sf{d}} (\theta - \E{\theta}) \diff t 
+ \displaystyle \frac12 \sum_k \mathcal L_{\xi_k}^2 {u} \diff t, 
\\
\diff \theta + \mathcal L_{\mathbb E[u]} {\theta} \diff t + \displaystyle \displaystyle\sum_k \mathcal L_{\xi_k} {\theta}\diff W_t^k &=  \displaystyle \displaystyle \frac12 \sum_k \mathcal L_{\xi_k}^2 {\theta} \diff t\,,
\end{array}
\right.
\end{equation}
where we denote the composition of Lie derivatives as, for example,  $\mathcal L_{\xi_k}(\mathcal L_{\xi_k} {\theta}) =: \mathcal L_{\xi_k}^2 {\theta}$.

Next, taking expectation at both sides of the equations above yields a deterministic equation for the evolution of the expectations given by

\begin{equation}\label{LA:SALT:Exp:Bou}
\left\{
\begin{array}{rl}
\partial_t \mathbb E[u] + \mathcal L_{\mathbb E[u]} \mathbb E[u] &= -{\sf{d}} \left(\mathbb E[p] - \mathbb E \left[|\u|^2/2\right]\right) + g \mathbb E[\theta] \mb{\hat{y}} +  \displaystyle \displaystyle \frac12 \sum_k \mathcal L_{\xi_k}^2 \mathbb E[u], \\
\partial_t \mathbb E[\theta] + \mathcal L_{\mathbb E[u]} \mathbb E[\theta] &= \displaystyle \displaystyle \frac12 \sum_k \mathcal L_{\xi_k}^2 \mathbb E[\theta].
\end{array}
\right.
\end{equation} 
It is straightforward to check that in vorticity form where $\omega=\nabla^{\perp}\cdot \mb u = \hat{\y}\cdot{\rm curl} \mb u$, we have that 
\begin{equation}\label{LA:SALT:Stra:Bou}
\left\{
\begin{array}{rl}
\diff \omega + \mathcal L_{\mathbb E[u]} \omega \diff t + \displaystyle \displaystyle\sum_k \mathcal L_{\xi_k} \omega \,\circ \diff W_t^k &= g \p_{x}\theta \diff t, \\
\diff \theta + \mathcal L_{\mathbb E[u]} {\theta} \diff t + \displaystyle \displaystyle\sum_k \mathcal L_{\xi_k} {\theta} \circ \diff W_t^k &=0.
\end{array}
\right.
\end{equation}
We stress here again that since $\omega$ is a scalar quantity for incompressible planar flow, its Lie derivative is to be understood as $\mathcal{L}_{\xi} \omega = \xi \cdot \nabla \omega.$ The corresponding equation for the expectation is given by
\begin{equation}\label{LA:Vor:Exp:Bou}
\left\{
\begin{array}{rl}
\partial_t \mathbb E[\omega] + \mathcal L_{\mathbb E[u]} \mathbb E[\omega] &= g \mathbb \p_{x} \E\theta+  \displaystyle \frac12 \displaystyle\sum_k \mathcal L_{\xi_k}^2 \mathbb E[\omega], \\
\partial_t \mathbb E[\theta] + \mathcal L_{\mathbb E[u]} \mathbb E[\theta] &=  \displaystyle \displaystyle \frac{1}{2} \sum_k \mathcal L_{\xi_k}^2 \mathbb E[\theta].
\end{array}
\right.
\end{equation}

\section{Lagrangian-averaged (LA) semidirect product systems with transport noise}\label{sec:LASDPsystems} 
In subsequent discussions, we will employ the following notations:
\begin{itemize}
\item $M$ is a smooth, orientable manifold,
\item $\rm{Diff}(M)$ denotes the group of diffeomorphisms on $M$,
\item $\mathfrak X(M)$ denotes the set of smooth vector fields on $M$,
\item $\Omega^1(M)$ denotes the set of differential one-forms on $M$,
\item $\rm{Den}(M)$ denotes the set of volume forms (densities) on $M$,
\item $V$ is any tensor field such that $\rm{Diff}(M)$ acts on it from the right (e.g. $V = C^\infty(M,\mathbb R)$ and $\rm{Diff}(M)$ acts on $V$ by composition from the right).
\end{itemize}

\subsection{Lie--Poisson structure of fluid equations with advected quantities.} \label{2-1}
\index{Lie--Poisson equations!advected quantities}

We have introduced a class of stochastic partial differential equations (SPDE) for continuum dynamics. 
This class of equations is Hamiltonian with a Lie--Poisson bracket given by the $L^2$ pairing between $\mathfrak{X}(M) \circledS V$ and its dual \cite{HMR1998} 
\index{Lie--Poisson!bracket} 
\begin{align}
\frac{\diff F}{\diff t}=\{ F, H\} 
= 
-\,\SCP{(\mu,a)}{\left[ \frac{\delta F}{\delta (\mu,a)}\,,\, \frac{\delta H}{\delta (\mu,a)} \right]}
\label{LP-Brkt}
\end{align}
where $F,H\in C^\infty(\mathfrak X^*(M) \times V^* \to \mathbb R)$, $\mu \in \mathfrak X^*(M) \cong \Omega^1(M) \otimes \rm{Den}(M)$, $a\in V^*$, ${\delta F}/{\delta (\mu,a)}\in \mathfrak{X}(M) \circledS V$ is the variational derivative (see \cite{marsden1983coadjoint}), and $\mathfrak{X}(M) \circledS V$ denotes the semidirect product Lie algebra of vector fields on $M$ acting on the vector space $V$. The square brackets $[\,\cdot\,,\,\cdot\,]$ denote the adjoint action of the semidirect product Lie algebra $\mathfrak{X}(M) \circledS V$ on itself.
\index{Lie algebra!adjoint action}

Upon integration by parts, the Lie--Poisson bracket in \eqref{LP-Brkt} may be expressed in terms of a Hamiltonian operator as \index{Lie--Poisson!bracket} 
\begin{align}
\frac{\diff F}{\diff t}=\{ F, H\}
=
-\int_M 
\begin{bmatrix}
{\delta F}/{\delta \mu} \\  {\delta F}/{ \delta a} 
\end{bmatrix}^T
\begin{bmatrix}
{\rm ad}^*_{\Box}\mu & \Box \diamond a
\\ 
\mathcal{L}_{\Box}a & 0
\end{bmatrix}
\begin{bmatrix}
{\delta H}/{\delta \mu} \\  {\delta H}/{ \delta a}  
\end{bmatrix}
{\rmd}V
\label{Ham-matrix-det}
\end{align}
where ${\rm ad^*} : \mathfrak X(M) \times \mathfrak X^*(M) \rightarrow \mathfrak X^*(M)$ is the coadjoint action, $\mathcal L_u \alpha$ is the Lie derivative of a tensor field $\alpha$ with respect to a vector field $u$, and the diamond operation $\diamond: V\times V^*\to \mathfrak{X}^*(M)$  is defined in terms of the Lie derivative as, 
\index{vector field!$\mathrm{ad}$ and $\mathrm{ad}^*$ actions} \index{vector field!Lie derivative} \index{diamond operation $(\diamond)$} \index{Lie derivative!tensor field} \index{Lie algebra!coadjoint action}
\begin{align}
\Scp{b\diamond a}{v}_{\mathfrak{X}(M)} := \Scp{b }{- \mathcal{L}_v a}_V 
\,,\label{diamond-def}
\end{align}
where $a\in V^*$ and $b\in V$. 
The definition \eqref{diamond-def} makes the Lie--Poisson bracket skew-symmetric in $L^2$ under integration by parts.
\index{Lie--Poisson!bracket} \index{Lie derivative!local expressions}

We note that the Lie derivative $\mathcal L$ has different local expressions depending on which type of tensor field it acts on, which we will list below. Let $u \in \mathfrak X(\mathbb R^n)$ for all examples below.
\begin{itemize} 
    \item (Scalar functions) Given a scalar field $f$, we have
    $$\mathcal L_u f = \mb{u} \cdot \nabla f.$$
    \item (Vector fields) If $v \in \mathfrak X(\mathbb R^n)$ is another vector field, then
    $$\mathcal L_u v = \left(\mb{u}\cdot \nabla \mb{v} - \mb{v}\cdot \nabla \mb{u}\right)\cdot \nabla = [u,v] = -{\rm ad}_u v.$$
    \item (One-forms) Given a one-form $\alpha \in \Omega^1(\mathbb R^n)$, the corresponding Lie derivative reads
    $$\mathcal L_{u}\alpha = \left(\mb{u} \cdot \nabla \mb{\alpha} + \sum_{j=1}^n \alpha_j \nabla u^j\right) \cdot {\rmd} \x.$$
    \item (Densities) Given a density $D = \rho {\rmd}^nx \in \Omega^n(\RR^n)$, we have
    $$\mathcal L_{u} D = \div(\rho \u) {\rmd}^nx.$$
    \item (One-form densities) Given a one-form density $\mu = \alpha \otimes \rho\,{\rmd} ^nx,$ where $\alpha \in \Omega^1(\mathbb R^n)$ and $\rho\,{\rmd}^nx \in \Omega^n(\RR^n)$, its Lie derivative is given by
    $$\mathcal{L}_u (\alpha \otimes \rho\,{\rmd} ^nx) = (\rho\, \mathcal{L}_u \alpha + \div(\rho \u) \alpha) \otimes {\rmd} ^nx. $$
    It is well-known that for one-form densities (which are dual under $L^2$ pairing to the Lie algebra of vector fields), the coadjoint representation of the Lie algebra is equivalent to the Lie derivative, i.e., ${\rm ad}^*_u (\alpha \otimes \rho\,{\rmd} ^nx) \equiv \mathcal{L}_u (\alpha \otimes \rho\,{\rmd} ^nx)$, a fact we will use throughout this section.
\end{itemize}
We refer the readers to \cite{HMR1998} for further examples of Lie derivatives arising in continuum dynamics and the corresponding expressions for the diamond operator. We also remark that all the previous definitions take the same form on the torus $\mathbb{T}^2.$
\begin{example}[The deterministic 2D Euler-Boussinesq equations] \label{EB-Ham-example2.1}
We recall that the Boussinesq system is given by
\begin{equation}\label{Inviscid_Boussinesq00}
\left\{
\begin{array}{rl}
\partial_t\u+\left(\u\cdot\nabla\right)\u &=-\nabla p + g\theta \hat{\y}, \qquad (\x,t)\in \mathbb{T}^2\times\RR^{+},\\
\partial_t\theta +\u\cdot\nabla \theta &= 0, \\
\nabla\cdot\u&=0,
\end{array}
\right.
\end{equation}
where $\u=(u_{1},u_{2})$ is the incompressible vector velocity field, $p$ is the scalar pressure, $g$ is the acceleration due to gravity, and $\theta$  corresponds to the temperature, which is transported by the fluid. In Lie--Poisson form with $(\mu,\theta,D)$ denoting momentum one-form density, potential temperature, and density respectively, where $\mu (x,t) :=  \mb u \cdot {\rmd} \x \otimes \rho \, {\rmd} ^2x,$ $D := \rho \, {\rmd} ^2x,$ and the advected potential temperature $\theta = \theta(x,t)$ is understood as a scalar quantity. In the semidirect product formalism presented in \eqref{Ham-matrix-det}, this can be expressed as
\begin{align}
\diff{F} =\{ F, h\}
=
-\hbox{\Large$\int$}_{\mathbb{T}^2}
\begin{bmatrix}
{\delta F}/{\delta \mu} \\  {\delta F}/{ \delta \theta}  \\  {\delta F}/{ \delta D}
\end{bmatrix}^T
\begin{bmatrix}
{\rm ad}^*_{\Box} \mu & \Box \diamond \theta & \Box \diamond D
\\ 
\mathcal{L}_{\Box} \theta & 0 & 0 \\
\mathcal{L}_{\Box} D & 0 & 0
\end{bmatrix}
\begin{bmatrix}
{\delta H } / {\delta \mu}
\\  {\delta H}/{ \delta \theta} \\
 {\delta H}/{ \delta D} 
\end{bmatrix}
{\rmd} ^2x,
\label{Ham-matrix-Boussinesq}
\end{align}
for Boussinesq Hamiltonian $h$ given in terms of $(\mu,\theta,D)$ by the sum of the kinetic and potential energies, plus a constraint applied by the Lagrange multiplier $p$ (the pressure) which enforces incompressibility
\begin{align}
\begin{split}
h(\mu,\theta,\rho) &=  \int_{\mathbb{T}^2} \left(\frac{1}{2\rho} |\mu|^2  
-  g\rho  \theta y + p (\rho-1) \right) \,{\rmd} ^2x 
\\&= \int_{\mathbb{T}^2} \left \langle \mu, u \right \rangle 
- \int_{\mathbb{T}^2} \left(\frac{\rho}{2} |\mb{u}|^2 + g\rho \theta y - p (\rho-1) \right) \,{\rmd} ^2x
\,,
\end{split}
\label{Boussinesq-Ham}
\end{align}
so that
\begin{align}
\frac{\delta h}{\delta \mu} = u := \mb u\cdot\nabla
\,,\quad
\frac{\delta h}{ \delta u} = \mu - \rho u = 0
\,,\quad
\frac{\delta h}{ \delta \theta} = -g\rho y
\,,\quad
\frac{\delta h}{\delta \rho} = p - \frac{|\mu|^2}{2\rho^2}- g\theta y.
\label{EBHam-var}
\end{align}
We note that the constraint coming from the Lagrangian multiplier $p$ yielding $\rho = 1$ is only to be imposed once the variations are taken and the final equations derived.
The definitions for the Lie-derivative, diamond, and coadjoint operator ${\rm ad}^*$ have been specified above. We note that these depend on the type of object they are being applied to (i.e. $\mu$ is a one-form density, whereas $\theta$ a scalar, and $D$ a volume form). Upon applying these definitions, we can rewrite \eqref{Ham-matrix-Boussinesq} as 
\begin{align}
\partial_t
\begin{bmatrix}
\mu \\  \theta \\ \rho 
\end{bmatrix}
= -
\begin{bmatrix}
 \mathcal{L}_{\Box} \mu &  - \Box  (\nabla \theta)    &  \rho \nabla {\Box}
\\ 
\Box \cdot (\nabla \theta) & 0 & 0 \\
 \nabla \cdot  (\rho \Box) & 0 & 0
\end{bmatrix}
\begin{bmatrix}
 u
\\  -g\rho y \\ p - |\mb{u}|^2/2 - g\theta y
\end{bmatrix},
\label{Ham-matrix-Boussinesq2}
\end{align}
which yields equations \eqref{Inviscid_Boussinesq00}. 
\end{example}

\subsection{SALT equations.}

The class of Hamiltonian SPDE treated here may be obtained by extending the Hamiltonian function to make it stochastic by adding the $L^2$ pairing of the momentum density $\mu$ with a Stratonovich stochastic process (denoted with the symbol $\circ \diff W_t$) whose spatial correlations are specified by a set of smooth vector fields, $\mb{\xi}_k(\x)$, $k=1,\dots,N$, as in 
\cite{Holm2015}, as
\begin{align}
H(\mu,a) \to \diff {h(\mu,a;\xi_k) }
:= H(\mu,a)\diff t + \displaystyle\sum_k\scp{\mu}{\xi_k}\circ \diff W_t^k
\,.\label{Stoch-Ham}
\end{align}
The Lie--Poisson bracket then yields \index{Lie--Poisson!bracket} 
\begin{align}
\diff  {F} =\{ F, \diff {h}\}
=
-\int_M 
\begin{bmatrix}
{\delta F}/{\delta \mu} \\  {\delta F}/{ \delta a} 
\end{bmatrix}^T
\begin{bmatrix}
{\rm ad}^*_{\Box}\mu & \Box \diamond a
\\ 
\mathcal{L}_{\Box}a & 0
\end{bmatrix}
\begin{bmatrix}
({\delta H } / {\delta \mu})\diff t + \displaystyle\sum_k\mb{\xi}_k(\x)\circ \diff W_t^k
\\  ({\delta H}/{ \delta a} ) \diff t 
\end{bmatrix}
{\rmd}V.
\label{Ham-matrix-SALT}
\end{align}
These equations describe stochastic advection by Lie transport (SALT)
\cite{Holm2015} and they comprise the basis for a new approach for data analysis, uncertainty quantification and uncertainty reduction by data assimilation using particle filtering \cite{CCHOS18b,CCHOS18a}. By defining the stochastic vector field 
\begin{align}
\diff{x_t} := 
({\delta H } / {\delta \mu})\diff t + \displaystyle\sum_k \mb{\xi}_k(\x)\circ \diff W_t^k
\label{Stoch-Lag-traj}
\end{align}
and recalling that ${\rm ad}^*_{{\rm d}x_t}\mu
= \mathcal{L}_{{\rm d}x_t}\mu,$ 
the SALT equations \eqref{Ham-matrix-SALT} may be rewritten in a compact form as 
\begin{align}
\begin{split}
\diff{\mu} + \mathcal{L}_{\diff{x_t}}\mu 
&= - \frac{\delta H }{\delta a}\diamond a \,\diff t
\,,\\
\diff{a} + \mathcal{L}_{\diff{x_t}}a &= 0
\,.
\end{split}
\label{SALT-adv-form}
\end{align}
The SALT equations in this form have been studied extensively, for example, in 
wave-current interactions \cite{Holm_WCI_JNLS2019}, uncertainty prediction \cite{GH19}, solution properties of stochastic fluid dynamics  \cite{CrFlHo2019, AloBetTak}, and turbulent cascades \cite{HolmTurbulent}, even when the spatial correlations are nonstationary \cite{GH18a,GH19}. 

\begin{example}[SALT 2D Euler-Boussinesq system]
The 2D SALT Boussinesq equations are given by
\begin{align}\label{Ham-matrix-Boussinesq-SALT}
\diff  {F} =\{ F, h\}
=
-\hbox{\Large$\int$}_{\mathbb{T}^2}
\begin{bmatrix}
{\delta F}/{\delta \mu} \\  {\delta F}/{ \delta \theta}  \\  {\delta F}/{ \delta D}
\end{bmatrix}^T
\begin{bmatrix}
{\rm ad}^*_{\Box} \mu & \Box \diamond \theta & \Box \diamond D
\\ 
\mathcal{L}_{\Box} \theta & 0 & 0 \\
\mathcal{L}_{\Box} D & 0 & 0
\end{bmatrix}
\begin{bmatrix}
{\delta h } / {\delta \mu}
\\  {\delta h}/{ \delta \theta} \\
 {\delta h}/{ \delta D} 
\end{bmatrix}
{\rmd} ^2x,
\end{align}
\end{example}
with
\begin{align}\label{Boussinesq-Ham-SALT}
h(\mu,\theta,D) = \int_{0}^t \int_{\mathbb{T}^2} \left(\frac{1}{2\rho} |\mu|^2  +  g \rho \theta y + p (\rho-1) \right) \,{\rmd} ^2x \diff s +  \displaystyle\sum_k\int_{0}^t \int_{\mathbb{T}^2} \left \langle \mu (x,t), \xi_k \right \rangle \, \circ \diff W_s^k
\,,
\end{align}
where $\mu = \rho \mb{u}  \cdot \sf{d} \mb{x} \otimes$ $\diff^2 x$ and $D = \rho \diff^2 x$ giving rise to the SALT 2D Euler--Boussinesq (EB) system
\begin{equation}\label{Inviscid_Boussinesq00-SALT}
\left\{
\begin{array}{rl}   
\diff \u + \u \cdot \nabla \u \diff t + \displaystyle\sum_k \mb{\xi_k} \cdot \nabla u \circ \diff W_t^k + \displaystyle\sum_k u^j \nabla \xi_k^j \circ \diff W_t^k 
&= -{\sf{d}}(p- |\u|^2/2)  + g\theta \hat{\y} \diff t,\\
\diff \theta + \u \cdot \nabla \theta \diff t + \displaystyle\sum_k \mb{\xi_k} \cdot \nabla \theta \circ \diff W_t^k &= 0, \\
\nabla\cdot\u&=0.
\end{array}
\right.
\end{equation}

We note that the well-posedness of this equation and a blow-up criterion for it were derived in \cite{DieAytJNLS}. In this section, by considering the Lagrangian-averaged version of \eqref{Inviscid_Boussinesq00-SALT}, we construct the LA-SALT 2D EB model, which will turn out to be \emph{globally} well-posed.

\subsection{Lagrangian-averaged (LA) SALT equations.}

LA-SALT is a modification of the SALT introduced in \cite{DrivasHolm2019} and analysed in \cite{DrivasHolmLeahy2020} for 3D stochastic fluid motion. This modification preserves the Lie--Poisson bracket structure of the SALT equations, while replacing the variational derivatives of the Hamiltonian by their expected values, denoted $\mathbb{E}[\,\cdot\,]$, as follows. First, the Lagrangian trajectory equation \eqref{Stoch-Lag-traj} is modified by taking the expectation of the drift velocity, as \index{Lie--Poisson!bracket} 
\begin{align}
\diff{X_t} := 
\E{\frac{\delta H }{\delta \mu}} \diff t + \displaystyle\sum_k\xi_k(x)\circ \diff W_t^k,
\label{Exp-Lag-traj}
\end{align}
where $H$ is the same Hamiltonian as in the SALT equations. We also take the expectation of the variational derivatives with respect to advected quantities $\E{{\delta H }/{\delta a}}$.

The Poisson operator then yields 
\begin{align}
\diff  {F} =\{ F, \diff {h}\}
=
-\int_M 
\begin{bmatrix}
{\delta F}/{\delta \mu} \\  {\delta F}/{ \delta a} 
\end{bmatrix}^T
\begin{bmatrix}
{\rm ad}^*_{\Box}\mu & \Box \diamond a
\\ 
\mathcal{L}_{\Box}a & 0
\end{bmatrix}
\begin{bmatrix}
\mathbb{E}\left[\frac{\delta H }{\delta \mu}\right] \diff t + \displaystyle\sum_k\xi_k(x)\circ \diff W_t^k
\\  \mathbb{E}\left[\frac{\delta H }{\delta a} \right] \diff t 
\end{bmatrix}
{\rmd}V \,.
\label{Ham-matrix-LASALT}
\end{align}
These equations describe \emph{Lagrangian-averaged} stochastic advection by Lie transport (LA-SALT). That is, the Lagrangian path ${\rm d}X_t$ in equation \eqref{Exp-Lag-traj} has been acquired by taking the expectation (averaging in probability space) of the drift velocity of the SALT Lagrangian path \eqref{Stoch-Lag-traj} at \emph{fixed Lagrangian label}.
The SALT equations in advective form \eqref{SALT-adv-form} now become the LA-SALT equations, given by
\begin{equation}\label{LASALT-adv-form}
\left\{
\begin{array}{rl}
{\rm d}{\mu} + \mathcal{L}_{\E{\frac{\delta H }{\delta \mu}}} {\mu} \,\diff t
+ \displaystyle\sum_k \mathcal{L}_{\xi_k} {\mu} \circ \diff W_t^k
&= - \,\mathbb{E}\Big[ \frac{\delta H}{\delta a}\Big] \diamond { a }\,\diff t
\,,\\
 {\rm d} {a} + \mathcal{L}_{\E{\frac{\delta H }{\delta \mu}}} {a} \,\diff t
 + \displaystyle\sum_k \mathcal{L}_{\xi_k} {a} \circ \diff W_t^k  &= 0
\,,
\end{array}
\right.
\end{equation}
with ${\rm d}X_t$ defined in equation \eqref{Exp-Lag-traj}. If there are several advected quantities, one sums over all of them in the diamond term in \eqref{LASALT-adv-form}. Notice that the LA-SALT equations in \eqref{Ham-matrix-LASALT} have the same Poisson matrix operator as for the SALT equations in \eqref{Ham-matrix-SALT} and therefore many key features of the Lie-Poisson system are preserved, such as the conservation of Casimirs and Kelvin's circulation theorem (see Remark \ref{kelvin} below).
Thus, between equations \eqref{Ham-matrix-SALT} and \eqref{Ham-matrix-LASALT}, only the variational derivatives of the deterministic parts of the Hamiltonian have been changed to accommodate the differences between Lagrangian trajectories for SALT and LA-SALT in equations \eqref{Stoch-Lag-traj} and \eqref{Exp-Lag-traj}. \index{Casimirs}

\begin{remark}[Comparing SALT and LA-SALT]\rm
The LA-SALT approach applies to the same physical class of equations as for SALT. Following the deterministic route set in \cite{HMR1998}, the class of SALT fluid equations was first derived in \cite{Holm2015} from the symmetry-reduced Lagrangians $\ell(u,a)$ for the Euler--Poincar\'e Hamilton's principle with $\mu=\delta \ell/\delta u$, whose variations were constrained to respect stochastic advection laws in \eqref{SALT-adv-form}. The LA-SALT approach modifies the stochastic process $ \diff{x_t}$ for the transport vector field in \eqref{Stoch-Lag-traj} which defines the stochastic Lagrangian trajectory in SALT to become $ \diff{X_t}$ as in \eqref{Exp-Lag-traj}. The Euler--Poincar\'e version of the Lie--Poisson expression of the motion equation in \eqref{LASALT-adv-form} is, 
\begin{align}\label{LASALT-EP}
 \diff{\frac{\delta \ell}{\delta u}} + \mathcal{L}_{ \diff{X_t}} \frac{\delta \ell}{\delta u}  
 = \mathbb{E}\Big[ \frac{\delta \ell}{\delta a}\Big] \diamond a\,\diff t
 \quad\hbox{and}\quad
  {\rm d}a + \mathcal{L}_{ {\rm d}X_t} a = 0
\,.
\end{align}
The comparisons between them can be derived from the relations ${\delta \ell}/{\delta u}=\mu$ and ${\delta \ell}/{\delta a}=-{\delta h}/{\delta a}$ which are obtained from the deterministic Legendre transform from the reduced Lagrangian to the reduced Hamiltonian,
\begin{align}
\diff h(\mu,a) = \scp{\mu}{u} - \ell(u,a)
\,,
\label{Legendre-xform}
\end{align}
and the assumption that the reduced Lagrangian is hyperregular, which almost always holds in continuum mechanics. \hfill $\square$
\end{remark}

\begin{remark}[The Kelvin circulation theorem for LA-SALT]\label{kelvin} 
In fluid dynamics, the mass density $D {\rmd}^3x$ is always an advected quantity, satisfying 
the continuity equation, which in this case is expressed as,
\index{LA-SALT!Kelvin circulation theorem}
\begin{align}
\diff {(D {\rmd}^3x)} + \mathcal{L}_{\diff{X_t}}(D {\rmd}^3x)  
= \big(\diff{D} + {\rm div}(\diff{X_t} D)\big) {\rmd}^3x
= 0
\,.
\label{Contin-eqn}
\end{align}
Consequently, if we define the circulation one-form $v= \mb{v}\cdot {\rmd} \mb{x}$ by 
\begin{align}
\mu = \mb{m}\cdot {\rmd} \mb{x} \otimes {\rmd}^3x = \mb{v}\cdot {\rmd} \mb{x} \otimes D {\rmd}^3x
= {v} \otimes D {\rmd}^3x
\,,\label{Circ-1form}
\end{align}
and use the continuity equation \eqref{Contin-eqn},
and then the advective form of the motion equation in \eqref{LASALT-adv-form}, we can write the Kelvin circulation theorem for LA-SALT as
\begin{align}
\rmd \oint_{c(\diff{X_t})} \!\!\!{\mb v}\cdot {\rmd} {\mb x}  
= \oint_{c(\diff{X_t})} \!\!\!
\big(\diff +  \mathcal{L}_{\diff{X_t}}\big)({\mb v}\cdot {\rmd}{\mb x})  
= - \oint_{c(\diff{X_t})} \frac{1}{D} 
\mathbb{E}\Big[\frac{\delta H }{\delta a}\Big]\diamond a
\,.\label{Kel-thm-LASALT}
\end{align}
This relation may be proved, for example, by following the corresponding proof of the stochastic Kelvin calculation for SALT in  \cite{de2019implications}. Thus, because the LA-SALT modification in \eqref{Exp-Lag-traj} of the SALT transport vector field in \eqref{Stoch-Lag-traj} preserves the Lie--Poisson Hamiltonian structure of SALT, one also acquires the Kelvin circulation theorem for LA-SALT in \eqref{Kel-thm-LASALT}. Note that for compressible fluids, the right-hand side of the relation in \eqref{Kel-thm-LASALT} can be nonlinear in the stochastic variables. \hfill $\square$
\end{remark}

\color{black}

\subsection{Evolution of the covariance tensor}
We have seen that the expectation of the variables in the LA-SALT equation form a closed system. Could we say the same about the covariance? For general semi-direct product LA-SALT systems \eqref{LASALT-adv-form}, the answer is no. However, the covariance for the advected quantities {\em does} always form a closed system. The proof of this covariance property in \cite{DieAytJNLS} is shown below. \index{climate!covariance tensor!evolution} \index{climate!covariance tensor!LA-SALT}
\index{weather v climate!expectation v pathwise fluctuations}  






\begin{proposition} \label{covariance-eq}
Let $a_t$ be any tensor field that satisfies the linear stochastic advection equation
\begin{align}\label{a-trans-eq}
    \diff a + \mathcal L_{\E{\frac{\delta H}{\delta \mu}}} a \diff t + \displaystyle\sum_k \mathcal L_{\xi_k} a \circ \diff W_t^k = 0,
\end{align}
and let $A^{(2)} := \mathbb E\left[(a - \mathbb E[a])^2\right]$ be the covariance tensor for the tensor field $a$, where $(\cdot)^2$ here means taking the tensor product with itself.
Then $A^{(2)}$ satisfies the following PDE:
\begin{align}\label{covar-eq}
\partial_t A^{(2)} + \mathcal L_{\E{\frac{\delta H}{\delta \mu}}}A^{(2)} = \sum_k \left(\frac12 \mathcal L_{\xi_k}^2 A^{(2)} + \left(\mathcal L_{\xi_k} \E{a}\right)^2 \right).
\end{align}
This is closed since $\E{a}$ and $\E{\frac{\delta H}{\delta \mu}}$ are determined by the closed system \eqref{LASALT-adv-form}. 
\end{proposition}
\begin{proof}[Proof of Proposition \ref{covariance-eq}]
Let $a' := a - \E{a}$ be the fluctuation about the mean, which can be shown using \eqref{a-trans-eq}-\eqref{covar-eq} to satisfy
\begin{align} \label{a-fluctuation-eq}
    \diff a' + \mathcal L_{\E{\frac{\delta H}{\delta \mu}}} a' \diff t + \displaystyle\sum_k \mathcal L_{\xi_k} a \circ \diff W_t^k = -\frac12 \displaystyle\sum_k \mathcal{L}_{\xi_k}^2 \E{a} \diff t.
\end{align}
Then by It\^o's product rule, we have
\begin{align*}
    \diff \, (a')^2 &= (\circ \diff a') \otimes a' + a' \otimes (\circ \diff a') \\
    &= -\mathcal L_{\E{\frac{\delta H}{\delta \mu}}} a' \otimes a' \diff t - \displaystyle\sum_k \mathcal L_{\xi_k} a \otimes a' \circ \diff W_t^k - \frac12 \displaystyle\sum_k \mathcal L_{\xi_k}^2 \E{a} \otimes a' \diff t \\
    &\quad\, - a' \otimes \mathcal L_{\E{\frac{\delta H}{\delta \mu}}} a'  \diff t - \displaystyle\sum_k  a' \otimes \mathcal L_{\xi_k} a \circ \diff W_t^k - \frac12 \displaystyle\sum_k a' \otimes \mathcal L_{\xi_k}^2 \E{a} \diff t,
\end{align*}
and using the Leibniz property of the Lie derivative, i.e., $\mathcal{L}(S \otimes T) = \mathcal{L} S \otimes T + S \otimes \mathcal{L} T$, for any tensors $S$ and $T$, we have
\begin{align}
    &\diff \,(a')^2 + \mathcal L_{\E{\frac{\delta H}{\delta \mu}}} (a')^2 \diff t + \frac12 \displaystyle\sum_k \left(\mathcal L_{\xi_k}^2 \E{a} \otimes a' + a' \otimes \mathcal L_{\xi_k}^2 \E{a}\right) \diff t \nonumber\\
    &= -\displaystyle\sum_k \left(\mathcal L_{\xi_k} a \otimes a' + a' \otimes \mathcal L_{\xi_k} a\right) \circ \diff W_t^k \nonumber\\
    &= \frac12 \displaystyle\sum_k \left(\mathcal L_{\xi_k}^2 a \otimes a' + 2\left(\mathcal L_{\xi_k} a\right)^2 + a' \otimes \mathcal L_{\xi_k}^2 a\right) \diff t - \displaystyle\sum_k \left(\mathcal L_{\xi_k} a \otimes a' + a' \otimes \mathcal L_{\xi_k} a\right) \diff W_t^k \nonumber\\
    &= \frac12 \displaystyle\sum_k \underbrace{\left(\mathcal L_{\xi_k}^2 a' \otimes a' + 2\left(\mathcal L_{\xi_k} a'\right)^2 + a' \otimes \mathcal L_{\xi_k}^2 a'\right)}_{\hbox{$= \mathcal{L}_{\xi_k}(\mathcal L_{\xi_k} (a')^2)$}} \diff t  + \sum_k\left(\mathcal L_{\xi_k} \E{a}\right)^2 \diff t -\displaystyle\sum_k \left(\mathcal L_{\xi_k} a \otimes a' + a' \otimes \mathcal L_{\xi_k} a\right) \diff W_t^k \nonumber\\
    &\quad  +\frac12 \displaystyle\sum_k \left(\mathcal L_{\xi_k}^2 \E{a} \otimes a' + 2\left(\mathcal L_{\xi_k} \E{a} \otimes \mathcal L_{\xi_k} a' + \mathcal L_{\xi_k} a'  \otimes \mathcal L_{\xi_k} \E{a} \right) + a' \otimes \mathcal L_{\xi_k}^2 \E{a}\right) \diff t,
\label{a'-squared-eq}
\end{align}
where in the second equality we converted from Stratonovich to It\^o integral (see appendix \ref{app-LASALT}) and in the last equality, we expanded the Stratonovich-to-It\^o correction term using $a = a' + \E{a}$ and the linearity of Lie derivatives. Taking expectations on both sides of \eqref{a'-squared-eq} and noting that (1) the expectation of the It\^o integral vanishes by the martingale property, and (2) $\E{a'} = 0$ by definition, we obtain
\begin{align*}
    \partial_{t} A^{(2)} + \mathcal L_{\E{\frac{\delta H}{\delta \mu}}}A^{(2)} = \displaystyle\sum_k \left(\frac12 \mathcal L_{\xi_k}^2 A^{(2)} + \left(\mathcal L_{\xi_k} \E{a}\right)^2 \right),
\end{align*}
as expected, where $A^{(2)} = \E{(a')^2}$.
\end{proof}

The covariance for the $\mu$ variable in \eqref{LASALT-adv-form} is unlikely to form a closed equation in general due to presence of the coupling term $\E{\frac{\delta H}{\delta a}} \diamond a$, however in the special example of the 2D Boussinesq equation, this is indeed possible as we will illustrate in the next example.

\begin{example}[Covariance of 2D LA-SALT Boussinesq]\label{2DEB-example}
Let us consider the special case of 2D LA-SALT Boussinesq system \eqref{LA:SALT:Ito:Bou}. Letting $u' := u - \E{u}$ and $\theta' := \theta - \E{\theta}$, we have the following equations for the fluctuations
\begin{equation}\label{theta-prime-eq}
\left\{
\begin{array}{rl}
\diff u' + \mathcal L_{\E{u}} u' \diff t + \displaystyle\sum_{k}\mathcal L_{\xi_k} u \circ \diff W_t^k &= -\displaystyle\sum_{k}\frac12 \mathcal L_{\xi_k}^2 \E{u} \diff t - g y {\rmd} \theta' \diff t, \\
    \diff \theta' + \mathcal L_{\E{u}} \theta' \diff t + \displaystyle\sum_{k}\mathcal L_{\xi_k} \theta \circ \diff W_t^k &= -\displaystyle\sum_{k}\frac12 \mathcal L_{\xi_k}^2 \E{\theta} \diff t. 
\end{array}
\right.
\end{equation}

Then by similar arguments as in the proof of Proposition \ref{covariance-eq}, we can show that the covariance fields satisfy the following PDEs
\begin{equation}\label{theta-var-eq}
\left\{
\begin{array}{rl}
    \partial_{t} U^{(2)}+ \mathcal L_{\E{u}}U^{(2)} &= \displaystyle\sum_k \left(\frac12 \mathcal L_{\xi_k}^2 U^{(2)} + \left(\mathcal L_{\xi_k} \E{u}\right)^2\right) - gy \E{u' \otimes {\rmd} \theta' + {\rmd} \theta' \otimes u'},  \\
    \partial_{t}\Theta^{(2)}+ \mathcal L_{\E{u}}\Theta^{(2)} &= \displaystyle\sum_k \left(\frac12 \mathcal L_{\xi_k}^2\Theta^{(2)} + \left(\mathcal L_{\xi_k} \E{\theta}\right)^2\right),
\end{array}
\right.
\end{equation}
where $U^{(2)} := \E{(u')^2}$ and $\Theta^{(2)} := \E{(\theta')^2}$. Clearly, this system is not closed due to the presence of the term $\E{u' \otimes {\rmd} \theta' + {\rmd} \theta' \otimes u'}$ in the $U^{(2)}$ equation. However, applying the exterior derivative $\sf d$ on both sides of the $\theta$-equation and its corresponding fluctuation \eqref{theta-prime-eq}, and noting that the exterior derivative and the Lie derivative commute as a consequence of Cartan's formula, we obtain the following system for $\sf{d}\theta$ and $\sf{d}\theta'$:
\begin{equation}
    \left\{
    \begin{array}{rl}
&\diff \,({\rmd} \theta) + \mathcal L_{\E{u}} {\rmd} \theta \diff t + \displaystyle\sum_k\mathcal L_{\xi_k} {\rmd} \theta \circ \diff W_t^k = 0, \\
&\partial_{t}\E{{\rmd} \theta} + \mathcal L_{\E{u}} {\rmd} \E{ \theta} = \frac12 \displaystyle\sum_k \mathcal L_{\xi_k}^2 \E{{\rmd}\theta},\\
&\diff\, ({\rmd} \theta') + \mathcal L_{\E{u}} {\rmd} \theta' \diff t + \displaystyle\sum_k\mathcal L_{\xi_k} {\rmd} \theta \circ \diff W_t^k = -\frac12\displaystyle\sum_k \mathcal L_{\xi_k}^2 \E{{\rmd}\theta} \diff t.
\end{array}
\right.
\end{equation}

By Proposition \ref{covariance-eq}, the covariance for $\sf{d}\theta$ evolves as
\begin{align}
    \partial_{t}({\sf{d}}\Theta^{(2)})+ \mathcal L_{\E{u}}{\sf{d}}\Theta^{(2)} = \displaystyle\sum_k \left(\frac12 \mathcal L_{\xi_k}^2 {\sf{d}}\Theta^{(2)} + \left(\mathcal L_{\xi_k} \E{{\sf{d}}\theta}\right)^2\right),\label{d-theta-var-eq}
\end{align}
where ${\sf{d}}\Theta^{(2)} := \E{(\sf{d}\theta')^2}$. We show that obtaining an equation for $\E{u' \otimes {\rmd} \theta' + {\rmd} \theta' \otimes u'}$ closes the system \eqref{theta-var-eq}.

By the stochastic product rule, we have
\begin{align}
    &\diff \,(u' \otimes {\rmd}\theta') = u' \otimes (\circ \diff \,({\rmd} \theta')) + (\circ \diff u') \otimes {\rmd} \theta'\nonumber \\
    &=-\mathcal L_{\E{u}}(u' \otimes {\rmd}\theta') \diff t - \left(u'\otimes \mathcal L_{\xi_k}({\rmd}\theta) + \mathcal L_{\xi_k} u \otimes {\rmd}\theta'\right)\circ \diff W_t^k \nonumber\\
    &\qquad - \frac12 \left(u' \otimes \mathcal L_{\xi_k}^2 \E{{\rmd}\theta} + \mathcal L_{\xi_k}^2 \E{u} \otimes {\rmd}\theta'\right) - gy ({\sf{d}}\theta')^2\diff t \nonumber\\
    &=-\mathcal L_{\E{u}}(u' \otimes {\rmd}\theta') \diff t - \left(u'\otimes \mathcal L_{\xi_k}({\rmd}\theta) + \mathcal L_{\xi_k} u \otimes {\rmd}\theta'\right) \diff W_t^k \nonumber\\
    &\qquad - \frac12 \left(u' \otimes \mathcal L_{\xi_k}^2 \E{{\rmd}\theta} + \mathcal L_{\xi_k}^2 \E{u} \otimes {\rmd}\theta'\right) - gy ({\sf{d}}\theta')^2\diff t \nonumber \\
    &\qquad \qquad + \frac12 \left(u' \otimes \mathcal L_{\xi_k}^2 ({\rmd} \theta) + 2(\mathcal L_{\xi_k} u) \otimes \left(\mathcal L_{\xi_k} ({\rmd} \theta)\right) + \mathcal L_{\xi_k}^2 u \otimes {\rmd} \theta'\right) \diff t, \label{omega-theta-prime}
\end{align}
By the Leibniz property of Lie derivatives, we have
\begin{align*}
    &\mathcal L_{\xi_k} \mathcal L_{\xi_k} \left(u' \otimes {\sf{d}}\theta'\right) = \mathcal L_{\xi_k} \left(u' \otimes \mathcal L_{\xi_k}({\sf{d}}\theta') +\mathcal L_{\xi_k}(\omega') \otimes {\sf{d}}\theta'\right) \\
    &= u' \otimes \mathcal L_{\xi_k}^2 ({\rmd} \theta') + 2(\mathcal L_{\xi_k} u') \otimes \left(\mathcal L_{\xi_k} ({\rmd} \theta')\right) + \mathcal L_{\xi_k}^2 \omega' \otimes {\rmd} \theta'.
\end{align*}

Now using the above expression and taking expectations on both sides of \eqref{omega-theta-prime} give us the PDE:
\begin{align*}
    \begin{split}
   \partial_{t} \E{u' \otimes \sf d\theta'} + \mathcal L_{\E{u}}\E{u' \otimes {\rmd}\theta'}
    &= \frac12 \mathcal L_{\xi_k}^2 \E{u' \otimes  {\sf{d}}\theta'} \\
    &\qquad + \left(\mathcal L_{\xi_k} \E{u}\right) \otimes \left(\mathcal L_{\xi_k} \E{{\rmd} \theta}\right) - gy\E{({\sf{d}}\theta')^2}.
    \end{split}
\end{align*}
Similarly, we get an equation for $\E{{\sf{d}}\theta' \otimes u'}$ and combining them gives us an equation for $\E{u' \otimes {\rmd} \theta' + {\rmd} \theta' \otimes u'}$, which reads
\begin{align}\label{cross-variance-eq}
    \begin{split}
    &\partial_{t} \E{u' \otimes \sf d\theta' + {\rmd} \theta' \otimes u'} + \mathcal L_{\E{u}}\E{u' \otimes {\rmd}\theta' + {\rmd} \theta' \otimes u'}
    = \frac12 \mathcal L_{\xi_k}^2 \E{u' \otimes  {\sf{d}}\theta' + {\rmd} \theta' \otimes u'} \\
    &\qquad + \left(\mathcal L_{\xi_k} \E{u}\right) \otimes \left(\mathcal L_{\xi_k} \E{{\rmd} \theta}\right) + \left(\mathcal L_{\xi_k} \E{{\rmd} \theta}\right) \otimes \left(\mathcal L_{\xi_k} \E{u}\right) - 2gy\E{({\sf{d}}\theta')^2}.
    \end{split}
\end{align}

Since the last term  $\E{({\sf{d}}\theta')^2}$ is just the covariance tensor ${\sf{d}}\Theta^{(2)}$, which we can solve for, we conclude that equations \eqref{theta-var-eq},\eqref{d-theta-var-eq} and \eqref{cross-variance-eq} form together a closed system for the covariance of the fields $(u, \theta)$ in the 2D LA-SALT Boussinesq system. For more details, see \cite{DieAytJNLS}.
\end{example}
\color{black}

\begin{remark}
By having a closed system of PDEs for the evolution of the covariance, we may deduce for instance its growth behaviour through the application of standard PDE methods. For instance if we consider the equation for the evolution of $\Theta^{(2)}$ \eqref{LASALT-adv-form}, where we assume incompressibility $\div{(\E{\u})} = 0,$ and choose $\xi^{(1)} = \hat{\mb{x}},$ $\xi^{(2)}=\hat{\mb{y}}$, then we can check directly that its $L^2$-norm satisfies
\begin{align*}
    \|\Theta^{(2)}_t\|_{L^2}^2 + \frac12 \int^t_0 \|\nabla \Theta^{(2)}_s\|_{L^2}^2 \diff s = \int^t_0 \|\nabla \E{\theta_s}\|_{L^2}^2 \diff s,
\end{align*}
where we have taken into account that $\Theta^{(2)}(0) = 0$. Since by the parabolicity of the expectation equation \eqref{LASALT-EPX}, we have the estimate
\begin{align*}
    \int^T_0 \|\nabla \E{\theta_t}\|_{L^2}^2 \diff t \leq C \|\theta_0\|_{L^2}^2 \, e^T,
\end{align*}
where $C > 0$ is some constant, we can deduce that the space-averaged covariance $\|\Theta^{(2)}\|_{L^2}^2$ evolves at most exponentially fast.
\end{remark}

\begin{remark}[Extension to $p$-th central moments]\label{p-th-moment-remark}
One may also ask if closed equations for the higher moments of the advected tensor field $a$ can be derived, thus providing a generalisation of Proposition \ref{covariance-eq}, which may help us to understand for instance the non-Gaussianity of the system. In the case where $a$ is a scalar field, the $p$-th central moment $A^{(p)} := \E{(a-\E{a})^p}$ indeed satisfies a closed, iterated system:
\begin{align} \label{pth-moment}
    \begin{split}
   \partial_{t} A^{(p)} + \mathcal L_{\E{\frac{\partial H}{\partial \mu}}}A^{(p)} &= \displaystyle\sum_k \left(\frac12 \mathcal L_{\xi_k}^2 A^{(p)} + p\left(\mathcal L_{\xi_k} A^{(p-1)}\right)\left(\mathcal L_{\xi_k} \E{a}\right) + \frac{p(p-1)}{2}A^{(p-2)} \left(\mathcal L_{\xi_k} \E{a}\right)^2\right),
    \end{split}
\end{align}
which recovers \eqref{covar-eq} in the case $p=2$ (see \cite{DieAytJNLS} for the proof). However, when $a$ is a general tensor field, we have not been able to obtain a closed system for its $p$-th central moment due to the non-commutativity of the tensor product (which is commutative only in the scalar field case).
\end{remark}


\part{SALT / LA-SALT Ocean--Atmosphere Models} 

As we have been discussing in the previous parts, our modelling strategy formulates two complementary stochastic idealised climate models called SALT and LA-SALT. The SALT climate model couples a stochastic PDE for the atmospheric circulation to a deterministic PDE for the circulation of the ocean. The stochasticity is incorporated by assuming that Lagrangian particles in the atmosphere follow a stochastic path along a Stratonovich process which guides the motion of the material loop in Kelvin's circulation theorem. The stochastic Lagrangian path of the material loop is a semimartingale stochastic process in the SALT approach and is a McKean-Vlasov  process in the LA SALT approach \cite{McKean}. The paper \cite{crisan2023implementation} applies the SALT and LA-SALT approaches to an established class of idealised climate models and proves a local well-posedness theorem for both sets of of augmented nonlinear stochastic partial differential equations (SPDE). The two stochastic models also both possess a Kelvin circulation theorem. 

Recent numerical simulations have been performed for both the SALT and LA-SALT stochastic models \cite{sharma2026structure}. A key element of interpreting such a numerical experiment is the sensible specification of the stochastic process. For the purpose of the discussion here, one may assume that the stochastic process is of Stratonovich \emph{type}. In contrast, numerical experiments require one to choose a specific Stratonovich process by incorporating externally obtained information either from observations or from high-resolution simulations \cite{sharma2026structure}. 

\color{black}

\bigskip

 In the remainder of the introduction we detail our modelling approach for the deterministic model in Section \ref{subsect_DetModel} and for the stochastic 
 model in Section \ref{subsect_StochModel}.   


Contents of \cite{crisan2023implementation}:
 \begin{enumerate}
  \item
  Adaptation of the deterministic Gill-Matsuno \cite{gill1980some,matsuno1966quasi} class of ocean-atmosphere climate model (OACM) to the geometric variational framework. This adaptation produces a Kelvin circulation theorem which retains the transformation properties which are the basis for the remainder of the paper. These transformation properties are inherited from the variational framework. They enable the formulation of the deterministic and stochastic models in terms of the same type of Kelvin circulation theorem. 
  \item
  Derivations of the SALT and LA-SALT stochastic versions of the OACM, whose flows all possess the same geometric transformation properties. This shared geometric structure enables the analysis to develop sequentially from deterministic to stochastic models.
  \item
  Mathematical analysis for the deterministic, SALT and LA-SALT versions of the OACM. Specifically we prove existence and uniqueness of local solutions for the deterministic  OACM, the existence of a martingale solution 
  for the SALT version of the OACM and existence and uniqueness of local solution for the LA-SALT version. 
  
  \item
  Outlook -- open problems, including further study of the dynamical equations derived here for the dynamics of OACM statistics. 
  \end{enumerate}
 
\color{black} 
  \bigskip

 \subsection{The Deterministic Climate Model}\label{subsect_DetModel}
The model of the atmospheric component of our idealized climate model consists of  the compressible 2D Navier-Stokes equation coupled to an advection-diffusion equation for  temperature $\theta^a$. The atmospheric velocity field $\U^a$ transports the temperature that provides the gradient term of the velocity equation. The ocean component of the coupled system consists of a 2D incompressible Navier-Stokes equations and an equation for the oceanic temperature variable $\theta^o$ that is passively advected by the ocean velocity field $\U^o$. Here the pressure acts here as a Lagrange multiplier to impose incompressibility. More specifically, the deterministic coupled PDE's for the ocean and the atmosphere are given by
the derivations of the deterministic and stochastic models using Hamilton's variational principle. 

\begin{align}
\text{Atmosphere: }\quad&\frac{\partial \U^a}{\partial t}
+(\U^a\cdot\nabla)\U^a
+\frac{1}{Ro^a}\U^{a\bot}
+\frac{1}{Ro^a}\nabla \theta^a
=
\frac{1}{Re^a}\triangle\U^a,\label{COUPLED_SWE_VELOC_A}
\\
& \frac{\partial \theta^a}{\partial t}
+ (\U^a\cdot\nabla)\theta^a
= 
 \gamma(\theta^a - \theta^o)
+\frac{1}{Pe^a}\triangle \theta^a.\label{COUPLED_SWE_T_A}\\
\text{Ocean: }\quad&\frac{\partial \U^o}{\partial t}
+(\U^o\cdot\nabla)\U^o
+\frac{1}{Ro^o}\U^{o\bot}
+\frac{1}{Ro^o}\nabla (p^o+q^a)
=
\sigma(\U^o-\bar{\U}^a_{sol})
+\frac{1}{Re^o}\triangle\U^o,\label{COUPLED_SWE_VELOC_O}\\
&\frac{\partial \theta^o}{\partial t}
+(\U^o\cdot\nabla)\theta^o
=\frac{1}{Pe^o}\triangle\theta^o,\label{COUPLED_SWE_T_O}\\
&\Div(\U^o)=0,\label{COUPLED_SWE_INCOMPRESS_O}\\
\text{with }& \text{ initial conditions }\nonumber\\
&\U^a(t_0)=\U^a_0,\ \theta^a(t_0)=\theta^a_0, \ \U^o(t_0)=\U^o_0,\ \theta^o(t_0)=\theta^o_0\nonumber.
\end{align}
In these equations, the ocean velocity $\U^o$ is coupled to the atmospheric velocity $\U^a$ and the atmospheric temperature $\theta^a$ is coupled to the oceanic temperature $\theta^o$. The {\it coupling constants} $\gamma,\sigma<0$ regulate the strength of the interaction between the two components. 

The velocity coupling between the compressible atmosphere and the incompressible ocean model deserves some consideration. To preserve the incompressibility of the oceanic velocity field during the coupling we apply the Leray-Helmholtz Theorem to decompose the atmospheric velocity $\U^a$into a solenoidal component $\U^a_{sol}$ and a gradient term $q^a$ such that $\U^a=\U^a_{sol}+\nabla q^a$. The gradient part is combined with the oceanic pressure. In a second step we remove the space average via 
$\bar{\U}:=\U-\frac{1}{|\Omega|}\int_\Omega \U dx$ such that the oceanic velocity fields remains in the space of periodic flows with vanishing average. This property allows to determine the oceanic pressure.
Physically, this step removes the rapid mean velocity of the atmosphere relative to the slower ocean velocity in the frame of motion of the Earth's rotation. This means the ocean momentum responds to the shear force, which is proportional to the difference between the local ocean velocity at a given time and the local deviation of the atmospheric velocity away from its mean velocity.   


The model above belongs to the class of  {\it intermediate coupled models}. These models are much simpler than the coupled general circulation models of the atmosphere-ocean system that are used for climate research. Intermediate coupled models allow to study fundamental aspects of the atmosphere-ocean interaction. The most prominent example is {\it El Ni\~{n}o-Southern Oscillation (ENSO)} in the tropical Pacific. As originally hypothesized by Bjerknes in 1969 \cite{bjerknes1969atmospheric} this climate phenomenon crucially depends on the coupled interaction of both ocean and atmosphere. 
According to Bjerknes stronger trade winds increase the upwelling in the east Pacific, thereby creating a temperature gradient in the sea-surface temperature
that amplifies the trade winds. This interaction between the trade winds and sea surface temperature in the tropical Pacific generates a quasi-periodic oscillation between the three ENSO-phases: the neutral phase, El Ni\~{n}o and La Ni\~{n}a. Intermediate coupled models have been used successfully to shed light on the fundamental principle of ENSO, thereby confirming Bjerknes hypothesis. 

The story of intermediate coupled models began with (uncoupled) models to study equatorial waves and their response to external forcing.
Matsuno \cite{matsuno1966quasi} investigated an (uncoupled) divergent barotropic model (single layer of incompressible fluid of homogeneous density, 
with a free surface, on the beta plane), given by
\begin{equation}\begin{split}\label{MATSUNO}
&\frac{\partial \U}{\partial t}
+\frac{1}{Ro^a}\U^{a\bot}
+\frac{1}{Ro^a}\nabla \theta
=0,\\
&\frac{\partial \theta}{\partial t}
+H\Div(\U)= Q.
\end{split}\end{equation}
Matsuno \cite{matsuno1966quasi} refers to $\theta$ as {\it surface elevation} above a mean depth $H$, and in this context $Q$ appears as a source/sink of mass. Gill \cite{gill1980some} studied the steady response to heating anomalies of a tropical atmosphere, as described by the Matsuno model. Systems of equations in the following class are often called  {\it Gill models}
\begin{equation}\begin{split}\label{GILL}
&\frac{\partial \U}{\partial t}
+\frac{1}{Ro^a}\U^{a\bot}
+\frac{1}{Ro^a}\nabla \theta + a\U
=0,\\
&\frac{\partial \theta}{\partial t}
+H\Div(\U) +b\theta= Q,
\end{split}\end{equation}
where $a,b$ are Raleigh friction and Newtonian cooling and where $Q$ is a heating term. In Gill's work $\theta$ is proportional to the {\it surface pressure}.
Since surface hydrostatic pressure is proportional to surface height, this identification is consistent with Matsuno's interpretation.

Atmospheric models of Gill-Matsuno type are often used to understand the atmospheric response during an El Ni\~{n}o to observed 
sea surface temperature anomalies. Zebiak \cite{zebiak1982simple} parametrized the heat flux from the ocean to the atmosphere in terms of the ocean sea surface temperature . This relation can be motivated by a linearization of Clausius-Clapeyron relation. See, e.g., \cite{zebiak1986atmospheric}. 
The famous {\it Cane-Zebiak model} \cite{cane1985theory} applied a steady state atmosphere following Gill (\ref{GILL}) and a two-layer ocean model, with two equations for layer thickness and two equations for temperature. This model produced the first ENO forecast \cite{zebiak1987model}.
For further discussion of the historical development of this class of models, see \cite{crisan2023implementation}.

In \cite{crisan2023implementation}, we have modified the equations in the models above by including nonlinear terms 
in the velocity and temperature equations, derived from a deterministic Hamilton's principle. 
We also introduce SALT and LA-SALT as new stochastic transport approaches for modelling ENSO's  well-known irregularity.

 \subsection{The Stochastic Atmospheric Climate Model}\label{subsect_StochModel}



The fundamental principle in modelling stochastic fluid advection is the Kelvin circulation theorem. As we shall see, each component of the deterministic atmosphere-ocean model in equations \eqref{COUPLED_SWE_VELOC_A} - \eqref{COUPLED_SWE_INCOMPRESS_O} above possesses its own Kelvin theorem, and the two components are coupled together by their relative velocity. The model \eqref{COUPLED_SWE_VELOC_A} - \eqref{COUPLED_SWE_INCOMPRESS_O} describes their interaction as the exchange of circulation between the atmosphere and ocean. Later we treat the atmospheric component of the model as being stochastic either in the sense of weather (SALT) or in the sense of climate (LA-SALT). In either case, the stochastic modification of the atmospheric dynamics will retain a Kelvin circulation theorems.

\begin{theorem}[Kelvin theorem for the deterministic atmospheric model in \eqref{COUPLED_SWE_VELOC_A} - \eqref{COUPLED_SWE_INCOMPRESS_O}]$
\label{KelvinThm_atmosphere_determnistic}\,$\\
The deterministic model for atmospheric dynamics satisfies the following Kelvin theorem for circulation around a loop $c(u^a)$ moving with the flow of the atmospheric velocity $\bs{u}^a$. Namely,
\begin{align*}
\frac{d}{dt}\oint_{c(u^a)} (\bs{u}^a + \frac{1}{Ro^a}\bs{R}(\bs{x}))\cdot d\bs{x}
= \frac{1}{Re^a}\oint_{c(u^a)} \triangle\U^a \cdot d\bs{x}
\,,
\end{align*}
where ${\rm curl}\bs{R} = 2\bs{\hat{z}}\Omega(\bs{x})$ is the Coriolis parameter in nondimensional units.
\end{theorem}

\begin{proof}
By direct calculation, one shows that the deterministic atmospheric dynamics in the model above satisfies the relation in the Kelvin circulation theorem,
\begin{align*}
\frac{d}{dt}\oint_{c(u^a)} (\bs{u}^a + \frac{1}{Ro^a}\bs{R}(\bs{x}))\cdot d\bs{x}
&=
\oint_{c(u^a)} (\partial_t + \mathcal{L}_{u^a}) \big((\bs{u}^a + \frac{1}{Ro^a}\bs{R}(\bs{x}))\cdot d\bs{x}\big)
\\&= \oint_{c(u^a)} \Big( 
\partial_t \bs{u}^a + (\bs{u}^a\cdot\nabla)\bs{u}^a + u_j^a\nabla {u^a}^j
\\& \qquad- \bs{u}^a\times {\rm curl}\frac{1}{Ro^a}\bs{R}(\bs{x}) 
+ \nabla (\bs{u}^a\cdot\frac{1}{Ro^a}\bs{R})
\Big)\cdot d\bs{x} 
\\\hbox{By the model}\quad& =  \oint_{c(u^a)}\Big(- \frac{1}{Ro^a}\nabla\theta^a
+ \frac12\nabla |\bs{u}^a|^2 + \nabla (\bs{u}^a\cdot\frac{1}{Ro^a}\bs{R}) +  \triangle\U^a 
\Big)\cdot d\bs{x} 
\\& = \frac{1}{Re^a}\oint_{c(u^a)} \triangle\U^a \cdot d\bs{x}
\,.
\end{align*}
\end{proof}
\begin{remark}
In the proof above, the symbol $\mathcal{L}_u$ denotes the Lie derivative with respect to the vector field $u^a=\bs{u}^a \cdot\nabla$ with components $\bs{u}^a(\bs{x},t)$ and $c(u^a)$ denotes a material loop moving with the atmospheric Lagrangian transport velocity $\bs{u}^a(\bs{x},t)$.
Consequently, in the absence of viscosity, atmospheric circulation is conserved by the deterministic model because the viscous term is absent then and the loop integrals of gradients such as $u_j\nabla u^j=\frac12\nabla |\bs{u}|^2$ vanish on the right-hand side of the equation in the proof. \\
\end{remark}

Likewise, the dynamics of the ocean component of the model above satisfies the following Kelvin circulation theorem. 
\begin{theorem}[Kelvin theorem for the deterministic oceanic model in \eqref{COUPLED_SWE_VELOC_A} - \eqref{COUPLED_SWE_INCOMPRESS_O}]$\,$\\
The circulation dynamics around a loop $c(u^o)$ moving with the flow of the oceanic velocity $\bs{u}^o$ is given by
\begin{align*}
\frac{d}{dt}\oint_{c(u^o)} (\bs{u}^o + \frac{1}{Ro^o}\bs{R}(\bs{x}))\cdot d\bs{x}
= \oint_{c(u^o)} \Big( \sigma(\U^o-\bar{\U}^a)
+ \frac{1}{Re^o} \triangle\U^o 
\Big)\cdot d\bs{x} 
\,.
\end{align*}
\end{theorem}

\begin{proof}
The proof follows analogously to the proof of Theorem \ref{KelvinThm_atmosphere_determnistic}. 
\\
\end{proof}

\subsection*{Stochastic Advection by Lie Transport (SALT) atmospheric model.} 

Let $(\Xi ,\mathcal{F},(\mathcal{F}_{t})_{t},\mathbb{P})$ be a filtered probability space on which we have defined a sequence of independent Brownian motions  $(W^{i})_{i}$. Let $(\xi_{i})_{i}$ be a given sequence of sufficiently smooth vector 
fields that satisfies a certain condition discussed in \cite{crisan2023implementation}. In this work we assume the vector fields $(\xi_{i})_{i}$
to be given. For numerical simulations one defines these vector fields by extracting information from observational data. For examples we refer to \cite{CCHOS18a,CCHOS18b}.
The derivation of the SALT atmospheric model introduces the stochastic Lagrangian path 
\begin{align}
{\rmd x_t} := \U^a(\bs{x},t)dt-\sum_i\xi_i^a(\bs{x})\circ dW_i(t)
\,.\label{Atmos-SALT-Lag-Path}
\end{align}

Following \cite{Holm2015}, Section \ref{Section-SALT} discusses the introduction of the stochastic Lagrangian paths in \eqref{Atmos-SALT-Lag-Path} into Hamilton's variational principle for the atmospheric model equations. This step  leads to the SALT version of the idealized deterministic climate model comprising equations \eqref{COUPLED_SWE_VELOC_A} - \eqref{COUPLED_SWE_INCOMPRESS_O}. Namely, the SALT model is specified by the system of stochastic differential equations below:\\[3mm]

\noindent
Atmosphere:\footnote{%
As in the deterministic case we will write the Coriolis parameter as $\mathrm{curl}\,\mathbf{R}(%
\mathbf{x})=2\Omega \mathbf{(x}).$}
\begin{align}
& d\mathbf{u}^{a}+(d\mathbf{x}_{t}^{a}\cdot \nabla )\mathbf{u}^{a}+\frac{1}{%
Ro^{a}}d\mathbf{x}_{t}^{a\bot }+{\sum_{i}\Big(u_{j}^{a}\nabla \xi
_{i}^{j}+\frac{1}{Ro^{a}}\nabla \Big(R_{j}\mathbf{(x})\xi _{i}^{j}\Big)\Big)%
\circ dW_{t}^{i}}  \notag \\
& \hspace{3cm}+\frac{1}{Ro^{a}}\nabla \theta ^{a}=\frac{1}{Re^{a}}\triangle 
\mathbf{u}^{a},  \label{COUPLED_SWE_VELOC_A_STOCH} \\
& d\theta ^{a}+d\mathbf{x}_{t}^{a}\cdot \nabla \theta ^{a}=-\gamma (\theta
^{o}-\theta ^{a})+\frac{1}{Pe^{a}}\triangle \theta ^{a},
\label{COUPLED_SWE_T_A_STOCH} \\
& d\mathbf{x}_{t}^{a}=\mathbf{u}^{a}dt+{\sum_{i}\xi _{i}\circ
dW_{l}^{i}}  \label{stochasticpath}
\end{align}

\noindent Ocean: 
\begin{align}
\frac{\partial \mathbf{u}^{o}}{\partial t}+(\mathbf{u}^{o}\cdot \nabla )%
\mathbf{u}^{o}&+\frac{1}{Ro^{o}}\mathbf{u}^{o\bot }+\frac{1}{Ro^{o}}\nabla
p^{o}  \notag \\
& =\sigma (\mathbf{u}^{o}-\mathbb{E}\bar{\mathbf{u}^{a}})+\frac{1}{Re^{o}}%
\triangle \mathbf{u}^{o},  \label{COUPLED_SWE_VELOC_O_STOCH} \\
\frac{\partial \theta ^{o}}{\partial t}+(\mathbf{u}^{o}\cdot \nabla
)\theta ^{o}&=\frac{1}{Pe^{o}}\triangle \theta ^{o},
\label{COUPLED_SWE_T_O_STOCH} \\
{\ div}(\mathbf{u}^{o})&=0,  \label{COUPLED_SWE_INCOMPRESS_O_STOCH}
\end{align}%

\begin{theorem}Kelvin theorem for the SALT version of the atmospheric model in equations \eqref{COUPLED_SWE_VELOC_A_STOCH} - \eqref{stochasticpath}
\begin{align}
{\rm d}\oint_{c({\rmd x_t})} (\bs{u}^a + \frac{1}{Ro^a}\bs{R}(\bs{x}))\cdot d\bs{x}
=  \frac{1}{Re}\oint_{c({\rmd x_t})} \triangle\U^a \,dt \cdot d\bs{x}
\,,\label{Atmos-SALT-Kel}
\end{align}
where $c({\rmd x_t})$ denotes any closed material loop whose line elements follow stochastic  Lagrangian paths 
as in \eqref{Atmos-SALT-Lag-Path}.
\end{theorem}
\begin{proof}
Upon suppressing the superscript $a$ in the velocity $\bs{u}^a$ for brevity of notation, we calculate
\begin{align*}
{\rm d}\oint_{c({\rmd x_t})} (\bs{u}+ \bs{R}(\bs{x}))\cdot d\bs{x}
&=
\oint_{c({\rmd x_t})} ({\rm d} + \mathcal{L}_{{\rmd x_t}}) \big((\bs{u} + \bs{R}(\bs{x}))\cdot d\bs{x}\big)
\\&= \oint_{c({\rmd x_t})} \Big( 
{\rm d} \bs{u} + (\rmd\bs{x}_t\cdot\nabla)\bs{u} + u_j \nabla {\rm d}x_t^j
\\& \qquad - {\rmd x_t}\times {\rm curl}\bs{R}(\bs{x}) 
+ \nabla ({\rmd x_t}\cdot\bs{R})
\Big)\cdot d\bs{x} 
\\\hbox{[By motion equation \eqref{COUPLED_SWE_VELOC_A_STOCH}]}\quad& =  \oint_{c({\rmd x_t})}\Big(- \nabla\theta dt
+ \frac12\nabla |\bs{u}|^2dt - {\rmd x_t}\times {\rm curl}\bs{R}(\bs{x})  + \nabla (\bs{u}\cdot\bs{R})dt
\\& \qquad  + { u_j  \nabla \sum \xi^j\circ dW(t)
+ \sum\nabla \big(\bs{\xi}\circ dW(t)\cdot\bs{R}\big)}
\Big)\cdot d\bs{x}
\\&=
\frac{1}{Re}\oint_{c({\rmd x_t})} \triangle\U\,dt \cdot d\bs{x}
\,.
\end{align*}

\end{proof}

\begin{remark} The stochastic equation for the potential temperature $\theta^a$ in the atmospheric model inherits the stochasticity of the Lagrangian trajectories ${\rmd x_t}$ in \eqref{Atmos-SALT-Lag-Path}, as a scalar tracer transport equation,
\begin{align}
{\rm d}\theta^a 
+ (\rmd\bs{x}_t\cdot\nabla)\theta^a
= 
 \big[ \gamma(\theta^a - \theta^o)
+\frac{1}{Pe^a}\triangle \theta^a\big]dt.
\label{Stoch_THETA_A}
\end{align}
\end{remark}



\subsection{Lagrangian-Averaged Stochastic Advection by Lie Transport (LA-SALT) atmospheric  model.} 

\color{black}
We next modify the SALT approach to the two-dimensional atmospheric component of the climate system in the previous section to make it non-local in probability space, in the sense that the expected velocity will replace the drift velocity in the semimartingale for the SALT transport velocity of the stochastic fluid flow. This stochastic fluid model is derived by exploiting a novel idea introduced in \cite{DrivasHolm2019} and developed further in \cite{alonso2020modelling,DrivasHolmLeahy2020}, of applying Lagrangian-averaging (LA) in probability space to the fluid equations governing stochastic advection by Lie transport (SALT) which were introduced in \cite{Holm2015}.

The LA-SALT approach achieves three results of potential interest in climate modelling. These results address three different components of the climate change problem. 
\begin{itemize}
    \item First, the LA-SALT approach introduces a sense of determinism into climate science, by replacing the drift velocity of the stochastic vector field for material transport by its expected value in equation \eqref{Atmos-SALT-Lag-Path}. In this step, the expected fluid velocity becomes deterministic. 
    \item Second, the LA-SALT approach reduces the dynamical equations for the fluctuations to a \emph{linear} stochastic transport problem with a deterministic drift velocity. Such problems are well-posed. We prove here that the LA-SALT version of the SALT climate model a possesses local weak solutions. 
    \item Third, the LA-SALT approach addresses the dynamics of the variances of the fluctuations. In particular, the third result enables the variances and higher moments of the fluctuation statistics to be found deterministically, as they are driven by a certain set of correlations of the fluctuations among themselves. 
\end{itemize}
In summary, the first LA-SALT result makes the distinction between climate and weather for the case at hand. Namely, the LA-SALT fluid equations for the 2D atmosphere-ocean climate model system may be regarded as a dissipative system akin to the Navier-Stokes equations for the expected motion (climate) which is embedded into a larger conservative system which includes the statistics of the fluctuation dynamics (weather). The second result provides a set of linear stochastic transport equations for predicting the fluctuations (weather) of the physical variables, as they are driven by the deterministic expected motion. The third result produces closed deterministic evolutionary equations for the dynamics of the variances and covariances of the stochastic fluctuations.
Thus, the LA-SALT approach to investigating the 2D atmosphere-ocean climate model system treated here reveals that its statistical properties are fundamentally dynamical. 
Specifically, the LA-SALT analysis of the 2D atmosphere-ocean model presented here defines its climate, climate change, weather, and change of weather statistics, in the context of a hierarchical systems of PDEs and SPDEs with unique local weak solutions.
\color{black}

\paragraph{LA-SALT}
The expectation terms in LA-SALT induce another modification of the model which preserves the Kelvin circulation theorem, whose expectation yields a deterministic equation,
\begin{theorem}[Kelvin theorem for the LA-SALT atmospheric model]
\begin{align}
{\rmd}\oint_{c({\rmd x_t})} \big(\bs{u}+ \bs{R}(\bs{x})\big)\cdot d\bs{x}
&=
\frac{1}{Re}\oint_{c({\rmd x_t})} \triangle\U \,dt\cdot d\bs{x}
\label{KelThm-LASALT}
\end{align}
where 
\[{\rmd x_t}^a := \mathbb{E}[\U^a](\bs{x},t)dt+\sum_i\xi_i^a\circ dW_i(t).\]
\end{theorem}
\noindent 
{\bf Proof.} The proof of the Kelvin theorem for LA-SALT follows the same lines as for SALT. $\Box$

\paragraph{LA-SALT atmospheric equations in Stratonovich form.}

The \emph{Stratonovich} LA-SALT equations are given 
in the standard notation for stochastic fluid dynamics by expanding out Kelvin's theorem in \eqref{KelThm-LASALT} to find
\begin{align}
{\rmd}\mathbf{u}^{a}+({\rmd\mathbf{X}_{t}}^{a}\cdot \nabla )\mathbf{u}%
^{a}+\frac{1}{Ro^{a}}{\rmd\mathbf{X}_{t}}^{a\bot }& 
+{ \sum_{i}\Big(u_{j}^{a}\nabla \xi _{i}^{j}+\frac{1}{Ro^{a}}\nabla \Big(R_{j}%
\mathbf{(x})\xi _{i}^{j}\Big)\Big)\circ dW_{t}^{i}} 
 \notag\\
& \hspace{-22mm}{+\,u_{j}^{a}\nabla \mathbb{E}[{u^{a}}^{j}]dt+%
\frac{1}{Ro^{a}}\nabla (\mathbb{E}[{\mathbf{u}}^{a}]\cdot \mathbf{R})dt}+%
\frac{1}{Ro^{a}}\nabla \theta ^{a}\,dt=\frac{1}{Re^{a}}\triangle \mathbf{u}^{a}\,dt\,,
 \notag\\
& d\theta ^{a}+{\rmd\mathbf{X}_{t}}^{a}\cdot \nabla \theta
^{a}=-\gamma (\theta ^{o}-\theta ^{a})\,dt + \frac{1}{Pe^{a}}\triangle \theta ^{a}\,dt\,.
\label{COUPLED_SWE_T_A_STOCH_LA-Strat-Intro}
\end{align}
where the \emph{Stratonovich} stochastic Lagrangian trajectory for LA-SALT is given by
\begin{equation}
{{\rmd}\mathbf{X}_{t}^{a}} 
:= {\mathbb{E}[\mathbf{u}^{a}]}(\bs{x},t)dt + \sum_{i}\xi _{i}^{a}(\bs{x})\circ dW_{i}(t).
\label{Lag-path-Strat}
\end{equation}

\paragraph{LA-SALT atmospheric equations in It\^o form.}
Likewise, the \emph{It\^o} LA-SALT equations are
given in the standard notation for stochastic fluid dynamics by
\begin{align}
& \rmd\mathbf{u}^{a}+({\rmd\mathbf{\wh{X}}_{t}}^{a}\cdot \nabla )\mathbf{u}%
^{a}+\frac{1}{Ro^{a}}{\rmd\mathbf{\wh{X}}_{t}}^{a\bot }
 + {
\sum_{i}\Big(u_{j}^{a}\nabla \xi _{i}^{j}+\frac{1}{Ro^{a}}\nabla \Big(R_{j}%
\mathbf{(x})\xi _{i}^{j}\Big)\Big) dW_{t}^{i}}  \nonumber \\
&+\frac12 \bigg[ \mathbf{\hat{z}}\times \xi \Big( {\rm div}\Big(\xi\,\big(\,\mathbf{\hat{z}}\cdot{\rm curl}\,(\,\mathbb{E}[{\mathbf{u}}^{a}] + \frac{1}{Ro^{a}}\mathbf{R}(\mathbf{x}) \big)\Big) \,\,\Big)  
- \nabla \bigg( \xi\cdot\nabla\Big(\xi \cdot\big(\mathbb{E}[{\mathbf{u}}^{a}] + \frac{1}{Ro^{a}}\mathbf{R}(\mathbf{x})\big) \Big)
\bigg)\bigg]dt
\nonumber \\& \hspace{22mm}{+\,u_{j}^{a}\nabla \mathbb{E}[{u^{a}}^{j}]dt+%
\frac{1}{Ro^{a}}\nabla (\mathbb{E}[{\mathbf{u}}^{a}]\cdot \mathbf{R})dt}+%
\frac{1}{Ro^{a}}\nabla \theta ^{a}\,dt = \frac{1}{Re^{a}}\triangle \mathbf{u}^{a}\,dt\,,
\label{COUPLED_SWE_VELOC_A_STOCH_LA-Ito-Intro} \\
& d\theta ^{a}+{\rmd\mathbf{\wh{X}}_{t}}^{a}\cdot \nabla \theta
^{a} - \frac12 \Big(\xi\cdot\nabla(\xi\cdot\nabla \theta^a)  \Big)dt 
=-\gamma (\theta ^{o}-\theta ^{a})\,dt +  \frac{1}{Pe^{a}}\triangle \theta ^{a}\,dt
\,,
\label{COUPLED_SWE_T_A_STOCH_LA-Ito-Intro}
\end{align}
where the \emph{It\^o} stochastic Lagrangian trajectory for LA-SALT is given by
\begin{equation}
{{\rmd}\mathbf{\wh{X}}_{t}^{a}} 
:= {\mathbb{E}[\mathbf{u}^{a}]}(\bs{x},t)dt + \sum_{i}\xi _{i}^{a}(\bs{x})dW_{i}(t).
\label{Lag-path-Ito}
\end{equation}

\begin{remark}[Expected LA-SALT atmospheric equations.]
Taking the expectation of equations \eqref{COUPLED_SWE_VELOC_A_STOCH_LA-Ito-Intro} and  \eqref{COUPLED_SWE_T_A_STOCH_LA-Ito-Intro} yields a closed set of deterministic PDE for the expectations 
$\mathbb{E}[{\mathbf{u}}^{a}]$ and $\mathbb{E}[\theta ^{a}]$. Subtracting the expectations from equations \eqref{COUPLED_SWE_VELOC_A_STOCH_LA-Ito-Intro} and  \eqref{COUPLED_SWE_T_A_STOCH_LA-Ito-Intro} yields \emph{linear equations} for the differences,   
\begin{align}
{\mathbf{u}}^{a'}:= {\mathbf{u}}^{a} - \mathbb{E}[{\mathbf{u}}^{a}]
\quad\hbox{and}\quad
\theta ^{a'} := \theta ^{a} - \mathbb{E}[\theta ^{a}]
\,.
\label{fluctuations_u_theta}
\end{align}
Since ${\mathbf{u}}^{a'}$ and $\theta ^{a'}$ satisfy $\mathbb{E}[{\mathbf{u}}^{a'}]=0$ and $\mathbb{E}[\theta ^{a'}]=0$, one may regard these difference variables as fluctuations of ${\mathbf{u}}^{a}$ and $\theta ^{a}$ away from their expected values. From here one can calculate the dynamical equations for the statistics of the atmospheric model, e.g., its variances and its other tensor moments, as detailed in \cite{DrivasHolm2019,DrivasHolmLeahy2020,alonso2020modelling}. Further details of these equations can be found in  \cite{crisan2023implementation,sharma2026structure}
\end{remark}

\paragraph{Oceanic part of the LA-SALT model.}
The oceanic part of the LA-SALT model \emph{coincides} with the oceanic part of the SALT model (\ref{COUPLED_SWE_VELOC_O_STOCH})-(\ref{COUPLED_SWE_INCOMPRESS_O_STOCH}).

\section{Summary conclusion and outlook from \cite{crisan2023implementation}}\label{Sec4}

\begin{enumerate}
    \item \cite{crisan2023implementation} showed that the Ocean-Atmosphere Climate Model (OACM) in equation set \eqref{COUPLED_SWE_VELOC_A} - \eqref{COUPLED_SWE_INCOMPRESS_O} analysed there for the well-known Gill-Matsuno class of models is simple enough to successfully admit the mathematical analysis required to prove the local well-posedness of these models. The physics underlying these models can be improved, of course. For example, one could naturally include heating by the Greenhouse Effect, and this heating would drive the statistical properties of the climate model. 

    \item In addition to proving well-posedness for both deterministic and stochastic OACM, \cite{crisan2023implementation} developed a new tool for climate science for predicting the \emph{evolution of climate statistics} such as the variance. Indeed, the application of LA-SALT to the OACM there established a method for also predicting the evolution of climate statistics such as the expectation of tensor moments of the fluctuations. The new tools introduced there in the context of LA-SALT show considerable promise for future applications.
    
    \item \cite{crisan2023implementation} also mentioned that the shared geometric structure of these OACM may facilitate the parallel development of their numerical simulations; for example, in stochastic ENSO simulations. The capability for these stochastic simulations has recently been developed in \cite{sharma2026structure} and they may be excellent candidates for applying the Data Analysis, Uncertainty Quantification and Particle Filtering methods for Data Assimilation of LA SALT solutions which have already been developed for SALT. See, e.g.,  \cite{CCHOS18a,CCHOS18b}. 
    
In particular, our collaborators have created a new \emph{Discrete Exterior Calculus} (DEC) framework for coupled atmosphere-ocean computations. The DEC computations in \cite{sharma2026structure} using the LA-SALT OACM model from \cite{crisan2023implementation} verify the predicted red-shift in the oceanic energy spectrum caused by stochastic atmospheric fluctuations, in agreement with the Hasselmann paradigm \cite{hasselmann1976stochastic}. This verification of the red-shift in the oceanic energy spectrum is important, because it is largely what makes computational oceanography possible at all.

\end{enumerate}

\subsection*{Acknowledgments}

I am enormously grateful to many friends and colleagues who have helped me to understand the recent development of stochastic geometric mechanics and its applications in oceanography. I particularly thank my collaborators, postdoctoral fellows and PhD students: C.J. Cotter, D. Crisan, A.B. Cruzeiro, A. Bethencourt de Le\'on, T. Drivas, S.R. Ephrati, A.D. Franken, F. Gay-Balmaz, B.J. Geurts, G. Gottwald, R. Hu, P. Korn, J.-M. Leahy, E. Luesink, J.P. Ortega, W. Pan, T.S. Ratiu, O.D. Street, S. Takao, and T. Tyranowski. The work presented here was partially supported by European Research Council (ERC) Synergy grant entitled ``Stochastic Transport in Upper Ocean Dynamics'', STUOD - DLV-856408.

\appendix

\section{Variational derivations of stochastic atmospheric models (SAM)} 
\label{sec: SAM1}

\paragraph{Summary.}
In this appendix we explain the Euler-Poincar\'e variational principle and use it to rederive the equations 
of the standard ideal 2D compressible atmospheric and incompressible oceanic flows. 
These ideal models separately conserve their corresponding energy and potential vorticity.
Next, we couple these models using the reduced Lagrange-d'Alembert method. Finally, we introduce stochasticity
 into the Lagrangian for the atmospheric flow and derive the full stochastic SALT and LA-SALT climate models that 
 were discussed in the text of \cite{crisan2023implementation} in both their Stratonovich and It\^o forms.\\

\subsection{Stochastic Euler-Poincar\'e variational principles for fluid dynamics}\label{Section-GeoMech}

\begin{definition}[Fluid trajectory]
A fluid trajectory starting from $X\in M$ in the flow domain manifold $M$ at time $t=0$ is given by $x(t)=g_t(X)=g(X,t)$, with $g:M\times\mathbb{R}^+\to M$ being a smooth one-parameter submanifold (i.e., a curve parameterised by time $t$) in the manifold of diffeomorphisms acting on $M$, denoted ${\rm Diff}(M)$. In the deterministic case, computing the time derivative, i.e., the tangent to the curve with initial data $g(X,0)=X$ along $g_t(X)$, leads to the following \emph{reconstruction equation}, given by
\begin{equation}
{\partial_t}g_t(X) = u(g_t(X),t),
\label{eq:reconstructiondeterministic}
\end{equation}
where $u_t(\,\cdot\,)=u(\,\cdot\,,t)$ is a time-dependent vector field whose flow, $g_t(\,\cdot\,)=g(\,\cdot\,,t)$, is defined by the characteristic curves of the vector field $u_t(\,\cdot\,)\in \mathfrak{X}(M)$. The vector fields in $\mathfrak{X}(M)$ comprise the Lie algebra associated to the class of time-dependent maps $g_t\in {\rm Diff}(M)$.  
\end{definition}

\begin{definition}[Advected quantities and Lie derivatives]
A fluid variable $a\in V^*$ defined in a vector space $V^*$ is said to be \emph{advected}, if it keeps its value $a_t=a_0$ along the fluid trajectories. Advected quantities are sometimes called \emph{tracers}, because the histories of scalar advected quantities with different initial values (labels) trace out the Lagrangian trajectories of each label, or initial value, via the \emph{push-forward} by the flow group, i.e., $a_t=g_{t\,*}a_0= a_0g_t^{-1}$, where $g_t\in {\rm Diff}(M)$ is the time-dependent curve on the manifold of diffeomorphisms whose action represents the evolution of the fluid trajectory by push-forward. An advected quantity $a_t$ satisfies an evolutionary partial differential equation (PDE) obtained from the time derivative of the pull-back relation $a_0 = g_t^*a_t$, as follows
\[
0 = {\partial_t}a_0 = {\partial_t}(g_t^*a_t )= g_t^*\big({\partial_t}a_t + \mathcal{L}_u a_t\big)
\quad\Longrightarrow\quad {\partial_t}a_t + \mathcal{L}_u a_t = 0 \,,
\]
where $\mathcal{L}_{u_t} (\,\cdot\,)$ denotes \emph{Lie derivative} by the vector field $u_t$ whose characteristic curves comprise the fluid trajectories. 
\end{definition}

\begin{definition}[Stochastic advection by Lie transport (SALT)\cite{Holm2015}]
In the setting of stochastic advection by Lie transport (SALT) the deterministic reconstruction equation in \eqref{eq:reconstructiondeterministic} is replaced by the semimartingale
\begin{equation}
{\rmd}g(X,t) = u(g_t(X),t)dt + \sum_{i=1}^M \xi_i(g_t(X))\circ dW_t^i,
\label{eq:reconstructionstochastic}
\end{equation}
where the symbol $\circ$ means that the stochastic integral is taken in the Stratonovich sense. The initial data is given by $g(X,0)=X$. The $W_t^i$ are independent, identically distributed Brownian motions, defined with respect to the standard stochastic basis $(\Omega,\mathcal{F},(\mathcal{F}_t)_{t\geq 0},\mathbb{P})$. The $\xi_i(\,\cdot\,)\in\mathfrak{X}$ are prescribed vector fields which are meant to represent uncertainty due to effects on advection of unknown rapid time dependence.

A stochastically advected quantity $a_t$ satisfies an evolutionary stochastic partial differential equation (SPDE) obtained as a  semimartingale relation via the pull-back relation $a_0 = g_t^*a_t$, as follows
\begin{equation}
0 = {\rmd} a_0 = {\rmd}(g_t^*a_t )= g_t^*\big({\rmd} a_t + \mathcal{L}_{{\rmd}x_t} a_t\big)
\quad\Longrightarrow\quad {\rmd} a_t + \mathcal{L}_{{\rmd}x_t} a_t = 0 \,,
\label{eq:KIWformula}
\end{equation}
where $\mathcal{L}_{u_t} (\,\cdot\,)$ denotes \emph{Lie derivative} by the vector field ${\rmd}x_t$ whose characteristic curves comprise the stochastic fluid trajectories in equation \eqref{eq:reconstructionstochastic}.  
\end{definition}
\index{stochastic!advection} \index{stochastic!advection}

\begin{remark}[Kunita-It\^o-Wentzell (KIW) formula.]
Equation \eqref{eq:KIWformula} is called the \emph{KIW formula}, after its discovery by Kunita as an extension of the It\^o-Wentzell formula to define a stochastic Lie derivative for tensors and differential $k$-forms. For references and a discussion of its recent role in stochastic advection for fluid dynamics, see \cite{de2024geometric}.
\end{remark}

\begin{remark}[The deterministic limit.]
In what follows, any of the stochastic fluid equations derived from the Euler-Poincar\'e variational approach will reduce to the corresponding deterministic fluid equations by simply setting $\xi_i\to0$ in the reconstruction equation for the stochastic fluid trajectory in \eqref{eq:reconstructionstochastic}. 
\end{remark}

\begin{definition}[The diamond operator]
The \emph{diamond operator} is defined for $a\in V^*$, $u\in\mathfrak{X}$ and fixed $b\in V$ as
\begin{equation}
\langle b\diamond a, u\rangle_{\mathfrak{X}^{*}\times\mathfrak{X}} := -\langle b,\mathcal{L}_u a\rangle_{V^*\times V}.
\end{equation}
Here, $\langle\,\cdot\,,\,\cdot\,\rangle_{\mathfrak{X}^{*}\times\mathfrak{X}}$ and  $\langle\,\cdot\,,\,\cdot\,\rangle_{\mathfrak{X}^{*}\times\mathfrak{X}}$ denote the real-valued, non-degenerate, symmetric pairings between corresponding dual spaces, which can be defined on a case-by-case basis. 
The diamond operator provides a map dual to the Lie derivative, as $\mathcal{L}_{(\,\cdot\,)}b:\mathfrak{X}\to V$ and $b\diamond(\,\cdot\,):V^*\to\mathfrak{X}^{*}$. This duality is crucial in defining the  Euler-Poincar\'e variational principle.
\end{definition}

\begin{definition}[The variational derivative] The variational derivative of a functional $F:\mathcal{B}\to\mathbb{R}$, where $\mathcal{B}$ is a Banach space, is denoted $\delta F/\delta \rho$ with $\rho\in\mathcal{B}$. The variational derivative $\delta F/\delta \rho$ can be defined via the linearisation of the functional $F$ with respect to the following infinitesimal deformation 
\begin{equation}
\delta F[\rho]:= \frac{d}{d\epsilon}\Big|_{\epsilon=0} F[\rho+\epsilon \delta\rho] = \int \frac{\delta F}{\delta \rho}(x)\delta\rho(x)\,dx =: \left\langle\frac{\delta F}{\delta \rho},\delta \rho\right\rangle.
\end{equation}
In the definition above, $\epsilon\ll1\in\mathbb{R}$ is a parameter, $\delta\rho\in\mathcal{B}$ is an arbitrary function and the  variation can be understood as a Fr\'echet derivative. With the definition of the functional derivative in place, the following lemma can be formulated.
\end{definition}

\begin{theorem}[Stochastic Euler-Poincar\'e theorem]\label{thm:SEP}
With the notation as above, the following statements are equivalent.
\begin{enumerate}[i)]
\item The constrained variational principle
\begin{equation}
\delta\int_{t_1}^{t_2}\ell(u,a)\,dt = 0
\end{equation}
holds on $\mathfrak{X}\times V^*$, using variations $\delta u$ and $\delta a$ of the form
\begin{equation}
\delta u = {\rmd}w - [{\rmd}x_t,w], \qquad \delta a = -\mathcal{L}_w a,
\label{EPstoch-var}
\end{equation}
where $w(t)\in \mathfrak{X}$ is arbitrary and vanishes at the endpoints in time for arbitrary times $t_1,t_2$.
\item The stochastic Euler-Poincar\'e equations
\begin{equation}
{\rmd}\frac{\delta \ell}{\delta u} + \mathcal{L}_{{\rmd}x_t}\frac{\delta \ell}{\delta u} 
= \frac{\delta \ell}{\delta a}\diamond a\,dt\,,
\label{eq:stochep}
\end{equation}
hold on $\mathfrak{X}^*$ and the stochastic advection equations
\begin{equation}
{\rmd}a + \mathcal{L}_{{\rmd}x_t}a = 0\,,
\label{eq:stochadv}
\end{equation}
hold on $\times V^*$.
\end{enumerate}
\end{theorem} 

\begin{proof}
Using integration by parts and the endpoint conditions $w(t_1)=0=w(t_2)$, the variation can be computed to be
\begin{equation}
\begin{aligned}
\delta\int_{t_1}^{t_2}\ell(u,a)\,dt 
&= 
\int_{t_1}^{t_2}\left\langle\frac{\delta\ell}{\delta u},\delta u\right\rangle + \left\langle\frac{\delta\ell}{\delta a},\delta a\right\rangle\,dt\\
&= \int_{t_1}^{t_2}\left\langle\frac{\delta\ell}{\delta u},{\rmd}w-[{\rmd}x_t,w]\right\rangle + \left\langle\frac{\delta\ell}{\delta a}\,dt,-\mathcal{L}_w a\right\rangle\\
&= \int_{t_1}^{t_2}\left\langle -{\rmd}\frac{\delta\ell}{\delta u} - \mathcal{L}_{{\rmd}x_t}\frac{\delta\ell}{\delta u} + \frac{\delta\ell}{\delta a}\diamond a\,dt,w\right\rangle\\
&= 0\,.
\end{aligned}
\label{eq:StochEPeqns}
\end{equation}
Since the vector field $w$ is arbitrary, one obtains the stochastic Euler-Poincar\'e equations. Finally, the advection equation \eqref{eq:stochadv} follows by applying the KIW formula to $a(t)=g_{t*}a_0$.
\end{proof}

\begin{remark}
This version of the stochastic Euler-Poincar\'e theorem is equivalent to the version presented in \cite{Holm2015}, which uses stochastic Clebsch constraints. In \cite{Holm2015} one can also find an investigation of the It\^o formulation of the stochastic Euler-Poincar\'e equation. 
\end{remark}

\begin{theorem}[Stochastic Kelvin-Noether Theorem] \label{thm:Kelvin}
Let $c$ denote a compact embedded one-dimensional smooth submanifold of $M$ and denote $c_t=g_t(c)$ for all $t\in [0,T]$. If the mass density $D_0$ (a top-form) is initially non-vanishing, then
\[
{\rmd}\oint_{c_t}\frac{1}{D_t}\frac{{\delta} \ell}{{\delta} u}(u_t,a_t)
= \oint_{c_t}\frac{1}{D_t}\frac{{\delta} \ell}{{\delta} a}(u_t,a_t)\diamond a_t\,.
\]
The integrated form of this relation is 
\[
\oint_{c_t}\frac{1}{D_t}\frac{{\delta} \ell}{{\delta} u}(u_t,a_t)
=\oint_{c_0}\frac{1}{D_0}\frac{{\delta} \ell}{{\delta} u}(u_0,a_0)
+ \int_0^t\oint_{c_s}\frac{1}{D_s}\frac{{\delta} \ell}{{\delta} a}(u_s,a_s)\diamond a_s \rmd s.
\]
\end{theorem}

\begin{proof}
The KIW formula \eqref{eq:KIWformula} -- also known as the \emph{Lie chain rule} -- implies that 
\begin{align*}
{\rmd}\oint_{c_t}\frac{1}{D_t}\frac{{\delta} \ell}{{\delta} u}(u_t,a_t)
&=
\oint_{c_0}
{\rmd}
g_t^*\Big(\frac{1}{D_t} \frac{\delta\ell}{\delta u}\Big)\bigg)
=
\oint_{c_0}
g_t^*\bigg(\big({\rmd} + \mathcal{L}_{{\rmd}x_t} \big)\Big(\frac{1}{D_t} \frac{\delta\ell}{\delta u}\Big)\bigg)
\\&=
\oint_{c_t}
\big({\rmd} + \mathcal{L}_{{\rmd}x_t} \big)\Big(\frac{1}{D_t} \frac{\delta\ell}{\delta u}\Big)
= \oint_{c_t}\frac{1}{D_t}\frac{{\delta} \ell}{{\delta} a}(u_t,a_t)\diamond a_t\,,
\end{align*}
where the first step also uses the stochastic advection equation for $D_t$ in \eqref{eq:stochep} and 
the final step follows by substituting the stochastic Euler-Poincar\'e equations in \eqref{eq:stochadv}.

\end{proof}

\subsection{2D SALT Stochastic Ocean-Atmospheric Model (SO-AM) }\label{Section-SALT}\label{sec-SAM}

\paragraph{Summary.}
This part of the appendix derives the SALT ocean-atmospheric model (SO-AM) for the isothermal ideal gas in equations \eqref{Atmos-SALT-Kel} and \eqref{Stoch_THETA_A}. \\ 	\index{SALT!ocean-atmospheric model (SO-AM)}

In the present notation, the Lagrangian for the deterministic 2D Compressible Atmospheric
Model (CAM) in Eulerian $(x,y)$ coordinates is,
\begin{align}
\ell\big[\bs{u}, D,\theta \big] = 
\int_{\Omega} \frac{D}{2} |\bs{u}|^2 
+ D \bs{u}\cdot\bs{R}(\bs{x})
- c_vD\theta\Pi  \,dx\,dy,
\label{CAM-Lag}
\end{align}
where $\bs{u}$ denotes 2D fluid velocity, $D$ is mass density, $\theta$ denotes the potential temperature, 
the function $\bs{R}(\bs{x})$ with ${\rm curl} \bs{R}=2\bs{\Omega}$ denotes the vector potential for the Coriolis parameter, 
$c_v$ is specific heat at constant volume, and $\Pi$ is the well-known Exner function, given by
\[
\Pi = \left(
\frac{p}{p_0}
\right)^{R/c_p},
\]
in which  $p_0$ is a reference pressure level, $c_p$ is specific heat at constant pressure and $R=c_p-c_v$ is the gas
constant. In these variables, the equation of state for an ideal gas in 2D with $n$ degrees of freedom is expressed as
\[
\Pi = \left(\frac{RD\theta}{p_0}\right)^{R/c_p} = \left(\frac{RD\theta}{p_0}\right)^{1-\gamma^{-1}}
= \left(\frac{RD\theta}{p_0}\right)^{2/(n+2)}
\,,
\]
since the specific heat ratio $\gamma=c_p/c_v = 1+2/n$ for ideal gases whose molecules possess $n$ degrees of freedom,   
comprising spatial translations, rotations and oscillations. Ideal gases of diatomic molecules in 3D have three translations, 
plus rotations and oscillations, so $n=5$ and $\gamma=7/5$ in that case.

Accordingly, the Lagrangian in \eqref{CAM-Lag} specialises for a ideal gas in 2D to 
\begin{align}
\ell\big[\bs{u}, D,\theta \big] = 
\int_{\Omega} \frac{D}{2} |\bs{u}|^2 
+ D \bs{u}\cdot\bs{R}(\bs{x})
- \kappa (D\theta)^\alpha  \,dx\,dy,
\label{2DAM-Lag}
\end{align}
where the constants $(\kappa,\alpha)$ take the values, 
\[
\kappa=c_v (R/p_0)^{2/(n+2)} 
\quad\hbox{and}\quad
\alpha = \frac{n+4}{n+2} = 1 + \frac{2}{n+2} = 2 -\gamma^{-1}
\,.
\]

We obtain the following variational derivatives of the Lagrangian in \eqref{CAM-Lag},
\begin{align}
\begin{split}
\frac1D\dede{\ell}{\bs{u}}  &=  \bs{u} + \bs{R}(\bs{x})
\,, \\
\dede{\ell}{D} & =  \frac{1}{2} |\bs{u}|^2 +  \bs{u}\cdot\bs{R}(\bs{x}) - \kappa\alpha(D\theta)^{\alpha-1}\theta \,, \\
\dede{\ell}{\theta}  &= -\, \kappa\alpha(D\theta)^{\alpha-1}D
\,.
\end{split}
\label{CAM-Lag-vars}
\end{align}

Substitution of the variational derivatives (\ref{CAM-Lag-vars}) of
the Lagrangian (\ref{CAM-Lag}) into the stochastic Euler-Poincar\'e equations
with SALT in \eqref{eq:stochep} gives the SAM system
 \begin{align}
\begin{split}
\left({\rmd} + \mathcal{L}_{{\rmd}x_t}\right) \left( \big( \bs{u} + \bs{R}(\bs{x}) \big) \cdot d\bs{x} \right)
 &=  \,{d\,}\left(\frac{1}{2} |\bs{u}|^2 +  \bs{u}\cdot\bs{R}(\bs{x}) -(\kappa/\gamma) (D\theta)^{\alpha} \right) dt
\,,\\
({\rmd} + \mathcal{L}_{{\rmd}x_t}) (D\theta \,dxdy) & =  0
\,.
\end{split}
\label{EPSD-geom-CAM}
\end{align}
Consequently, we recover the Kelvin circulation conservation law for the SAM in the compact form
\begin{align}
{\rmd}\oint_{c_t}\hspace{-2mm} 
 \big( \bs{u} + \bs{R}(\bs{x}) \big) \cdot d \bs{x} 
=
0 
\,,
\label{SAM-circons}
\end{align}
where $c_t=g_t(c_0)$ for all $t\in [0,T]$ denotes the push-forward by the SAM flow of the initial $c_0$,
a compact embedded one-dimensional smooth submanifold of $M$.

\begin{corollary}\label{CAM-PV}
  The system of SAM equations in \eqref{EPSD-geom-CAM} implies that
  potential vorticity $q:= \omega / (D\theta)$ is conserved along flow lines of the stochastic fluid
  trajectory ${\rmd}x_t$,
\begin{align}
{\rmd} q +{\rmd}\bs{x}_t \cdot\nabla q = 0
\quad\hbox{with potential vorticity }q:= \omega / (D\theta) 
\quad\hbox{and} \quad
\omega := \bs{\widehat{z}}\cdot{\rm curl}\big( \bs{u} + \bs{R}(\bs{x}) \big) 
\,.
\label{SAM-vorticitythm}
\end{align}
In turn, this formula implies that the following infinite family of integral quantities is conserved
\begin{align}
C_\Phi = \int_\Omega (D\theta)\Phi(q)\,dxdy
\,,
\label{SAM-enstrophy-thm}
\end{align}
for any differentiable function $\Phi$.
\end{corollary}
\begin{corollary}
The \emph{Deterministic} AM equations in (\ref{EPSD-geom-CAM}) with $\bs{\xi}_i\to0$ are Hamiltonian, with conserved energy\footnote{The energy $E$ in \eqref{SAM-erg} is not conserved for $\bs{\xi}_i\ne0$, though, because $\bs{\xi}_i\ne0$ injects the explicit time dependence of stochastic Lagrangian trajectories into the Euler-Poincar\'e variations \eqref{EPstoch-var} in Hamilton's principle.}
\begin{align}
E = \int_{\Omega} \frac{D}{2} |\bs{u}|^2 + \kappa (D\theta)^\alpha \,dxdy.
\label{SAM-erg}
\end{align}
The Lie-Poisson Hamiltonian structure of deterministic fluid equations is discussed in \cite{HMR1998} from the viewpoint of the Euler-Poincar\'e of Hamilton's principle for fluid dynamics. 
\end{corollary}
\begin{remark}
The system of SAM equations in \eqref{EPSD-geom-CAM} may also be written equivalently in \emph{standard} fluid dynamics notation as
\begin{align}
\begin{split}
{\rmd} \bs{u} + {\rmd}\bs{x}_t \cdot\nabla \bs{u} - {\rmd}\bs{x}_t  \times 2 \bs{\Omega}
 + \sum_{i=1}^M \Big( u_j  \nabla  \xi_i^j  +  \nabla \big(\bs{\xi}_i  \cdot \bs{R}\big)\Big) \circ dW_t^i 
&=  -\,(\kappa/\gamma) \nabla (D\theta)^\alpha  dt
\,, \\
{\rmd} (D\theta) + \nabla  \cdot \big( D\theta\, {\rmd}\bs{x}_t  \big) & =  0
\,.\end{split}
\label{Eady-EPSDeqns2}
\end{align}
\end{remark}
If one assumes low Mach number, so that $D\approx1$, and then adds viscosity and diffusion of heat, this set of equations will reproduce the SALT atmospheric model in equations \eqref{Atmos-SALT-Kel} and \eqref{Stoch_THETA_A} when one also sets $\alpha=1=\gamma$, which holds for the isothermal case of the ideal gas. 

\subsection{2D LA-SALT Stochastic Atmospheric Model (LA-SAM) }\label{app-LASALT}

\paragraph{Summary.}
Here we provide a geometric derivation of the Lagrangian Averaged Stochastic Advection by Lie Transport (LA-SALT) atmospheric model (LA-SAM) for the isothermal ideal gas discussed in section \ref{sec-SAM}. \\  \index{LA-SALT!atmospheric model (LA-SAM)}

The simplest way to derive the 2D LA-SALT Stochastic Atmospheric Model (LA-SAM) is to alter the \emph{Stratonovich} stochastic Lagrangian trajectory in equations \eqref{EPSD-geom-CAM} to the \emph{Stratonovich} LA-SALT stochastic path, which reads
 
\begin{equation}
\rmd x_{t}^{a}\rightarrow \ {{\rmd}\mathbf{X}_{t}^{a}} 
:= {\mathbb{E}[\mathbf{u}^{a}]}(\bs{x},t)dt + \sum_{i}\xi _{i}^{a}(\bs{x})\circ dW_{i}(t).
\label{LA-SALTpath}
\end{equation}

The LA-SAM system then emerges in \emph{Stratonovich} stochastic geometric form as
 \begin{align}
\begin{split}
\left({\rmd} + \mathcal{L}_{{{\rmd}\mathbf{X}_{t}^{a}}} \right)  
\big( \bs{u} + \frac{1}{Ro^{a}}\bs{R}(\bs{x}) \big) \cdot d\bs{x}  
 &=  \,{d\,}\left(\frac{1}{2} |\bs{u}|^2 +  \frac{1}{Ro^{a}}\bs{u}\cdot\bs{R}(\bs{x}) 
 -\frac{1}{Ro^{a}}(\kappa/\gamma) (D\theta)^{\alpha} \right) dt
\,,\\
\left({\rmd} + \mathcal{L}_{{{\rmd}\mathbf{X}_{t}^{a}}}\right) (D\theta \,dxdy) & =  0
\,,
\end{split}
\label{EPSD-geom-LASAM-Strat}
\end{align}
where $\mathcal{L}_{{{\rmd}\mathbf{X}_{t}^{a}}}$ denotes Lie derivative with respect to \emph{Stratonovich} LA-SALT stochastic path in \eqref{LA-SALTpath}.
The stochastic geometric LA-SAM system in equation \eqref{EPSD-geom-LASAM-Strat} is given in its corresponding \emph{It\^o} form by
 \begin{align}
\begin{split}
\left({\rmd} + \mathcal{L}_{{{\rmd}\mathbf{\wh{X}}_{t}^{a}}}\right)
\Big( \big( \bs{u} + \frac{1}{Ro^{a}}\bs{R}(\bs{x}) \big) \cdot d\bs{x} \Big) 
& - \frac12  \sum_{i} \mathcal{L}_{\xi _{i}^{a}}\left(\mathcal{L}_{\xi _{i}^{a}}
\Big( \big( \bs{u} + \frac{1}{Ro^{a}}\bs{R}(\bs{x}) \big) \cdot d\bs{x} \Big) \right)
\\&=  \,{d\,}\left(\frac{1}{2} |\bs{u}|^2 +  \frac{1}{Ro^{a}}\bs{u}\cdot\bs{R}(\bs{x}) 
 -\frac{1}{Ro^{a}}(\kappa/\gamma) (D\theta)^{\alpha} \right) dt
\,,\\
\left({\rmd} + \mathcal{L}_{{{\rmd}\mathbf{\wh{X}}_{t}^{a}}}\right) (D\theta \,dxdy) 
& - \frac12  \sum_{i} \mathcal{L}_{\xi _{i}^{a}}\left(\mathcal{L}_{\xi _{i}^{a}} (D\theta \,dxdy) \right)
 =  0
\,,
\end{split}
\label{EPSD-geom-LASAM-Ito}
\end{align}
where the \emph{It\^o} Lagrangian trajectory for the LA-SALT, reads
 
\begin{equation*}
{{\rmd}\mathbf{\wh{X}}_{t}^{a}} 
:= {\mathbb{E}[\mathbf{u}^{a}]}(\bs{x},t)dt + \sum_{i}\xi _{i}^{a}(\bs{x})dW_{i}(t).
\end{equation*}

In standard notation for fluid dynamics and with $\alpha=1=\gamma$, the \emph{Stratonovich} LA-SAM equations in \eqref{EPSD-geom-LA-SAM-Strat} become

\begin{align}
d\mathbf{u}^{a}+({\rmd\mathbf{X}_{t}}^{a}\cdot \nabla )\mathbf{u}
^{a}+\frac{1}{Ro^{a}}{\rmd\mathbf{X}_{t}}^{a\bot }
& +{
\sum_{i}\Big(u_{j}^{a}\nabla \xi _{i}^{j}+\frac{1}{Ro^{a}}\nabla \Big(R_{j}
\mathbf{(x})\xi _{i}^{j}\Big)\Big)\circ dW_{t}^{i}} 
 \notag\\
& \hspace{-22mm}{+\,u_{j}^{a}\nabla \mathbb{E}[{u^{a}}^{j}]dt+
\frac{1}{Ro^{a}}\nabla (\mathbb{E}[{\mathbf{u}}^{a}]\cdot \mathbf{R})dt}+
\frac{1}{Ro^{a}}\nabla \theta ^{a}\,dt=\frac{1}{Re^{a}}\triangle \mathbf{u}^{a}\,dt\,,
 \notag\\
& d\theta ^{a}+{\rmd\mathbf{X}_{t}}^{a}\cdot \nabla \theta
^{a}=-\gamma (\theta ^{o}-\theta ^{a})\,dt + \frac{1}{Pe^{a}}\triangle \theta ^{a}\,dt\,.
\label{COUPLED_SWE_T_A_STOCH_LA-Strat}
\end{align}
In the previous equation, we have used the continuity equation to eliminate the areal density $D$.

Likewise, in standard notation for fluid dynamics and with $\alpha=1=\gamma$, the \emph{It\^o} LA-SAM equations in \eqref{EPSD-geom-LA-SAM-Ito} become

\begin{align}{\rmd\mathbf{\wh{X}}_{t}}
& d\mathbf{u}^{a}+({\rmd\mathbf{\wh{X}}_{t}}^{a}\cdot \nabla )\mathbf{u}
^{a}+\frac{1}{Ro^{a}}{\rmd\mathbf{\wh{X}}_{t}}^{a\bot }
 + {
\sum_{i}\Big(u_{j}^{a}\nabla \xi _{i}^{j}+\frac{1}{Ro^{a}}\nabla \Big(R_{j}
\mathbf{(x})\xi _{i}^{j}\Big)\Big) dW_{t}^{i}}  \nonumber \\
&+\frac12 \bigg[ \mathbf{\hat{z}}\times \xi \Big( {\rm div}\Big(\xi\,\big(\,\mathbf{\hat{z}}\cdot{\rm curl}\,(\,\mathbb{E}[{\mathbf{u}}^{a}] + \frac{1}{Ro^{a}}\mathbf{R}(\mathbf{x}) \big)\Big) \,\,\Big)  
- \nabla \bigg( \xi\cdot\nabla\Big(\xi \cdot\big(\mathbb{E}[{\mathbf{u}}^{a}] + \frac{1}{Ro^{a}}\mathbf{R}(\mathbf{x})\big) \Big)
\bigg)\bigg]dt
\nonumber \\& \hspace{22mm}{+\,u_{j}^{a}\nabla \mathbb{E}[{u^{a}}^{j}]dt+
\frac{1}{Ro^{a}}\nabla (\mathbb{E}[{\mathbf{u}}^{a}]\cdot \mathbf{R})dt}+
\frac{1}{Ro^{a}}\nabla \theta ^{a}\,dt = \frac{1}{Re^{a}}\triangle \mathbf{u}^{a}\,dt\,,
\label{COUPLED_SWE_VELOC_A_STOCH_LA-Ito} \\
& d\theta ^{a}+{\rmd\mathbf{\wh{X}}_{t}}^{a}\cdot \nabla \theta
^{a} - \frac12 \Big(\xi\cdot\nabla(\xi\cdot\nabla \theta^a)  \Big)dt 
=-\gamma (\theta ^{o}-\theta ^{a})\,dt +  \frac{1}{Pe^{a}}\triangle \theta ^{a}\,dt
\,.
\label{COUPLED_SWE_T_A_STOCH_LA-Ito}
\end{align}
In the final equation, we have again used the continuity equation to eliminate the areal density $D$.



\bibliography{biblio}
\bibliographystyle{alpha}

\printindex

\end{document}